\documentclass[aps,prx,superscriptaddress,twocolumn,longbibliography,floatfix]{revtex4-2}
\usepackage{units}
\usepackage{amsmath,braket}
\usepackage{amsthm}
\usepackage{amssymb}
\usepackage{graphicx}
\usepackage{color}
\usepackage{xcolor}
\usepackage{bbold}
\usepackage{apptools}
\usepackage{appendix}
\usepackage{hhline}

\definecolor{myurlcolor}{rgb}{0,0,0.7}
\definecolor{myrefcolor}{rgb}{0.1,0,0.9}

\usepackage[
	breaklinks,
	pdftex,
	colorlinks=true, 
	linkcolor=myrefcolor,
	citecolor=myurlcolor,
	urlcolor=myurlcolor
]{hyperref}

\newcommand{\SMLong}{Appendix}

\newcommand{\SM}{Appendix}

\newtheorem{theorem}{Theorem}
\newtheorem{lemma}{Lemma}

\AtAppendix{\counterwithin{theorem}{section}}
\AtAppendix{\counterwithin{lemma}{section}}
\AtAppendix{\counterwithin{fact}{section}}
\AtAppendix{\counterwithin{definition}{section}}

\usepackage[linesnumbered, ruled,vlined]{algorithm2e}

\graphicspath{{./imagesTesting/}}

\renewcommand{\eqref}[1]{Eq.~(\ref{#1})} %

\def\app#1#2{%
  \mathrel{%
    \setbox0=\hbox{$#1\sim$}%
    \setbox2=\hbox{%
      \rlap{\hbox{$#1\propto$}}%
      \lower1.1\ht0\box0%
    }%
    \raise0.25\ht2\box2%
  }%
}

\ifx\proof\undefined
\newenvironment{proof}[1][\protect\proofname]{\par
	\normalfont\topsep6\p@\@plus6\p@\relax
	\trivlist
	\itemindent\parindent
	\item[\hskip\labelsep\scshape #1]\ignorespaces
}{%
	\endtrivlist\@endpefalse
}
\providecommand{\proofname}{Proof}
\fi

\newtheorem{proposition}{Proposition}
\makeatother

\newcommand{\idg}[1]{{\bfseries #1)}}

\newcommand\numberthis{\addtocounter{equation}{1}\tag{\theequation}}

\providecommand{\factname}{Fact}
\providecommand{\theoremname}{Theorem}
\providecommand{\claimname}{Claim}
\providecommand{\lemmaname}{Lemma}
\providecommand{\definitionname}{Definition}
\providecommand{\corollaryname}{Corollary}
\providecommand{\conjecturename}{Conjecture}

\definecolor{THc}{rgb}{0.9,0.3,0.2}

\newtheorem{definition}{\protect\definitionname}

\usepackage{physics}

\newcommand{\subfigimg}[3][,]{%
	\setbox1=\hbox{\includegraphics[#1]{#3}}%
	\leavevmode\rlap{\usebox1}%
	\rlap{\hspace*{2pt}\raisebox{\dimexpr\ht1-0.5\baselineskip}{{\bfseries \large\textsf{#2}}}}%
	\phantom{\usebox1}%
}

\newcommand{\sectionMain}[1]{
\let\oldaddcontentsline\addcontentsline%
\renewcommand{\addcontentsline}[3]{}%
\section{#1}
\let\addcontentsline\oldaddcontentsline
}

\newcommand{\prlsection}[1]{\section{#1}}

\begin{document}

\title{Efficient witnessing and testing of magic in mixed quantum states}

\author{Tobias Haug}
\email{tobias.haug@u.nus.edu}
\affiliation{Quantum Research Center, Technology Innovation Institute, Abu Dhabi, UAE}

\author{Poetri Sonya Tarabunga}
\email{poetri.tarabunga@tum.de}
\affiliation{Technical University of Munich, TUM School of Natural Sciences, Physics Department, 85748 Garching, Germany}
\affiliation{Munich Center for Quantum Science and Technology (MCQST), Schellingstr. 4, 80799 München, Germany}

\begin{abstract}
Nonstabilizerness or `magic' is a crucial resource for quantum computers which can be distilled from noisy quantum states. However, determining the magic of mixed quantum has been a notoriously difficult task.
Here, we provide efficient witnesses of magic based on the stabilizer R\'enyi entropy which robustly indicate the presence of magic and quantitatively estimate magic monotones.
We also design efficient property testing algorithms to reliably distinguish states with high and low magic, assuming the entropy is bounded.
We apply our methods to certify the number of noisy T-gates under a wide class of noise models.
Additionally, using the IonQ quantum computer, we experimentally verify the magic of noisy random quantum circuits. Surprisingly, we find that magic is highly robust, persisting even under exponentially strong noise.
Our witnesses can also be efficiently computed for matrix product states, revealing that  subsystems of many-body quantum states can contain extensive magic despite entanglement. 
Finally, our work also has direct implications for cryptography and pseudomagic: To mimic high magic states with as little magic as possible, one requires an extensive amount of entropy. This implies that entropy is a necessary resource to hide magic from eavesdroppers.
Our work uncovers powerful tools to verify and study the complexity of noisy quantum systems.
\end{abstract}

\maketitle

 \let\oldaddcontentsline\addcontentsline%
\renewcommand{\addcontentsline}[3]{}%

\prlsection{Introduction}
Nonstabilizerness, which is also colloquially known as magic, is the key quantum resource required to achieve universal quantum computation~\cite{bravyi2005universal,veitch2014resource}. To run a fault-tolerant quantum computer, non-universal Clifford operations~\cite{gottesman1997stabilizer,nielsen2011quantum,kitaev2003fault,eastin2009restriction,howard2014contextuality} are combined with non-Clifford gates, where the latter cannot be implemented transversally in a fault-tolerant manner. Instead, such non-Clifford gates must be implemented via complicated protocols~\cite{bravyi2005universal}. For example, magic state distillation purifies many noisy magic resource states into a less noisy magic state using a costly protocol~\cite{kitaev2003fault,eastin2009restriction,litinski2019magic}.

Crucially, to be able to distill useful magic resource states, one requires the noisy input states to contain at least some degree of magic. However, noise is prone to destroy magic, inhibiting the distillation process. As non-magic states can  be efficiently simulated, this also destroys any hope for quantum advantage~\cite{bravyi2016trading}.  Thus, it is essential to understand the relationship between noise and magic. 
For states $\rho$, noise is characterized by the $2$-R\'enyi entropy $S_2=-\ln\text{tr}(\rho^2)$.
For pure states with $S_2=0$, efficiently computable measures of magic exist~\cite{leone2022stabilizer,haug2022scalable,haug2023efficient,haug2022quantifying,haug2023stabilizer,tarabunga2024nonstabilizerness,lami2023quantum} and efficient quantum algorithms can test whether a given quantum state contains magic~\cite{gross2021schur,haug2022scalable,haug2023efficient,grewal2023improved,bao2024tolerant,arunachalam2024note,iyer2024tolerant,hinsche2024single}.
However, the restriction to pure states is a severe problem, as realistic experiments only prepare noisy mixed states, even with quantum error correction~\cite{bluvstein2024logical}. Previous experiments used magic measures that are only valid for pure states~\cite{oliviero2022measuring,haug2022scalable,niroula2023phase,bluvstein2024logical}, and thus could not fully certify whether the prepared mixed states actually contained magic.
A weaker notion of detecting magic is witnessing, which indicates the presence of magic, but is unable to reliable identify all types of magic states~\cite{dai2022detecting,warmuz2024magic,macedo2025witnessingmagicbellinequalities,macedo2025witnessing}. Such magic witnesses for mixed states have been proposed, but so far no efficiently implementable witness is known to our knowledge.

A different notion of estimating magic is property testing, which robustly determines whether a state has a high or low amount of magic~\cite{rubinfeld1996robust,goldreich1998property,buhrman2008quantum,montanaro2013survey}.
However, testing properties of mixed states is a notoriously difficult problem~\cite{montanaro2013survey}. In the limit where states are highly mixed with $S_2=\omega(\log n)$, testing magic is inherently inefficient~\cite{bansal2024pseudorandomdensitymatrices}. 
In contrast, for states with low entropy $S_2=O(\log n)$, whether testing can be efficient has remained an open problem. 

The poor understanding of how magic and noise interact mainly arises due to a lack of tools to study the magic of mixed states. This has been a major bottleneck not only for experiments, but also for numerical studies and analytics. In particular, while the magic of pure many-body states has been studied recently, magic has not been well understood for mixed many-body states such as entangled subsystems~\cite{sarkar2020characterization,liu2022many,haug2022quantifying,oliviero2022magic,tarabunga2023many,tarabunga2025efficientmutualmagicmagic,korbany2025longrangenonstabilizerness}.
Notably, this problem is shared with most other resource theories such as entanglement or coherence~\cite{montanaro2013survey,bansal2024pseudorandomdensitymatrices}.

The efficiency of property testing is fundamentally connected to quantum cryptography~\cite{ji2018pseudorandom}: Here, one wants to hide information about states from eavesdroppers, which requires testing to be inefficient as a prerequisite. In this context, the notion of pseudoresources has recently been introduced: 
Pseuodresources  mimic high-resource states using only a low amount of resources~\cite{bouland2022quantum,haug2023pseudorandom,gu2023little,bansal2024pseudorandomdensitymatrices,tanggara2025neartermpseudorandompseudoresourcequantum}.
For example, in pseudomagic one generates a state ensemble with low magic $g(n)$ that is indistinguishable (for any efficient quantum algorithm) from a high-magic ensemble with magic $f(n)$~\cite{gu2023little}. 
Pseudoresources can provide important cryptographic primitives such commitment and oblivious transfer~\cite{grilo2025quantumpseudoresourcesimplycryptography}.
Further, pseudomagic is crucial for the task of securely encrypting quantum states: only when the pseudomagic gap is maximal, no information about magic can be learned by eavesdroppers~\cite{haug2025pseudorandom}.
Of course, one would like the pseudomagic gap $f(n)$ vs $g(n)$ to be as large as possible. It has been shown that the gap depends on the entropy of a state. It is maximal for $S_2=\omega(\log n)$, however the size of the gap is unknown for $S_2=O(\log n)$~\cite{bansal2024pseudorandomdensitymatrices}.

Here, we introduce genuine magic witnesses derived from stabilizer R\'enyi entropies (SREs), which can be efficiently measured on quantum computers via Bell measurements.
Our witnesses robustly indicate that a given state indeed contains magic.
Notably, our witnesses fulfil the stronger notion of quantitative witnessing~\cite{eisert2007quantitative}, which gives quantitative predictions about magic monotones.
Beyond witnessing, we can also test magic, i.e. unambiguously distinguish states depending on their magic: Given an unknown $n$-qubit quantum state, we efficiently test whether it contains $O(\log n)$ or $\omega(\log n)$ magic, assuming the state has $2$-R\'enyi entropy $S_2=O(\log n)$. 
This also allows us to efficiently certify the number of T-states,  a common resource state for universal quantum computers, even when subject to mixed unital Clifford noise.
We show that magic is surprisingly robust in the presence of noise, and can survive even under exponentially strong depolarising noise. In particular, we can witness magic in noisy local random circuits up to a critical depth which is independent of qubit number.
As experimental demonstration, we witness the magic of noisy Clifford circuits doped with T-gates on the IonQ quantum computer, robustly certifying that magic has indeed been generated.
Our witness is also efficient for matrix product states (MPSs) and can be computed in $O(n\chi^3)$ time. We study the magic contained within entangled subsystems of the ground state of the transverse field Ising model, finding that magic survives even in the presence of entanglement.
Finally, we study the implications of our work for cryptography and pseudomagic. We show that for low-magic states to masquerade as high-magic states, there is a trade-off between entropy and magic: States with low entropy require $\omega(\log n)$ magic to mimic high-magic states, while high entropy states ($S_2=\omega(\log n)$) need no magic at all. This implies that entropy is a necessary resource for securely hiding information about magic from eavesdroppers.  We also make progress on the complexity of preparing pseudorandom density matrices, an important cryptographic notion, where we bound the number of T-gates   as $\omega(\log n)$ for $S_2=O(\log n)$.
Our work demonstrates that magic of mixed states can be efficiently characterized and can be highly robust to noise.

We summarize our main results in Fig.~\ref{fig:sketch}, while the complexity of testing and pseudomagic are shown in Tab.~\ref{tab:testing} and Tab.~\ref{tab:pseudoresource} respectively.

\begin{figure}[htbp]
	\centering	
    \subfigimg[width=0.49\textwidth]{}{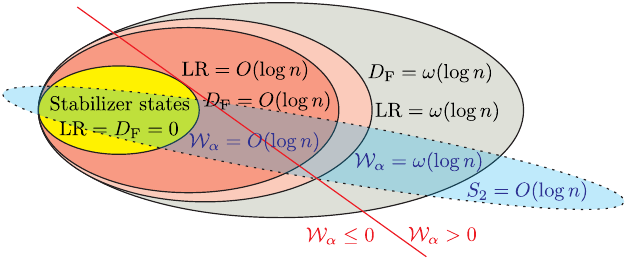}
	\caption{Schematic representation of $n$-qubit state space in terms of magic monotones, namely log-free robustness of magic $\text{LR}$ and stabilizer fidelity $D_\text{F}$. 
    Whenever witness $\mathcal{W}_\alpha>0$ ($\alpha\geq1/2)$ is positive (to the right of red line), it confirms that a state has magic, i.e. is not a stabilizer state. 
    For low entropy $S_2=O(\log n)$ (blue-shaded area), $\mathcal{W}_\alpha$ also gives quantitative predictions on $\text{LR}$ and $D_\text{F}$ and enables efficient testing of magic.  }
        \label{fig:sketch}
\end{figure}

\begin{table}[htbp]\centering
\begin{tabular}{| l|l| }
    \hline    
     Entropy &  Copies  \\
    \hline\hline
    $S_2=0$  & $O(\text{poly}(n))$~\cite{gross2021schur}\\\hline
 $S_2=O(\log n)$ & $O(\text{poly}(n))$ [this work] \\\hline
 $S_2=\omega(\log n)$ & $2^{\omega(\log n)}$~\cite{bansal2024pseudorandomdensitymatrices}\\ \hline
\end{tabular}
\caption{Number of copies of state $\rho$ needed to test whether $\rho$ has $M(\rho)=O(\log n)$ or $M(\rho)=\omega(\log n)$ magic. Complexity depends on $2$-R\'enyi entropy $S_2=-\ln\text{tr}(\rho^2)$, where we characterize magic $M$ by log-free robustness of magic and stabilizer fidelity.
}
\label{tab:testing}
\end{table}

\begin{table}[htbp]\centering
\begin{tabular}{| l|l| }
    \hline    
    Entropy &  $f(n)$ vs $g(n)$  \\
    \hline\hline
    $S_2=0$  & $\Theta(n)$ vs $\omega(\log n)$~\cite{gu2023little}\\\hline
 $S_2=O(\log n)$ & $\Theta(n)$ vs $\omega(\log n)$ [this work] \\\hline
 $S_2=\omega(\log n)$ & $\Theta(n)$ vs $0$~\cite{bansal2024pseudorandomdensitymatrices}\\ \hline
\end{tabular}
\caption{Pseudomagic gap $f(n)$ vs $g(n)$ between high-magic and low-magic state ensemble, depending on $2$-R\'enyi entropy $S_2=-\ln\text{tr}(\rho^2)$. Magic is characterized by the log-free robustness of magic and stabilizer fidelity.
}
\label{tab:pseudoresource}
\end{table}

\prlsection{Magic witness} \label{sec:sre}
We start by defining our magic witness for $n$-qubit state $\rho$ as
\begin{equation}\label{eq:witnessExp}
    \mathcal{W}_\alpha(\rho)= \frac{1}{1-\alpha}\ln A_\alpha(\rho) - \frac{1-2\alpha}{1-\alpha}S_2(\rho)\,,
\end{equation}
where $S_2(\rho)=-\ln\text{tr}(\rho^2)$ is the 2-R\'enyi entropy, and $A_\alpha(\rho)$ is the $\alpha$-moment of the Pauli spectrum~\cite{haug2023efficient}
\begin{equation}
    A_\alpha(\rho)=2^{-n}\sum_{P\in\mathcal{P}_n} \vert\text{tr}(\rho P)\vert^{2\alpha}\,.
\end{equation}
with $\mathcal{P}_n$ being the set of $4^n$ Pauli strings $P=\otimes_{j=1}^n \sigma^{a_j}$ which are tensor products of single-qubit Pauli operators $\sigma^0=I_1$, $\sigma^1=\sigma^x$, $\sigma^2=\sigma^z$, $\sigma^3=\sigma^y$. 
One can relate 
\begin{equation}
    \mathcal{W}_\alpha(\rho)\equiv M_\alpha(\rho)-2S_2(\rho)
\end{equation}
to the SRE for mixed states $M_\alpha=1/(1-\alpha)(\ln A_\alpha+S_2)$, which is the $\alpha$-R\'enyi entropy over the distribution $p_\rho(P)=2^{-n}\text{tr}(\rho P)^2/\text{tr}(\rho^2)$ up to a constant shift~\cite{tarabunga2025efficientmutualmagicmagic}. Further, for pure states $\rho=\ket{\psi}\bra{\psi}$, the witness is equivalent to the SRE $\mathcal{W}_\alpha(\ket{\psi})\equiv M_\alpha(\ket{\psi})$~\cite{leone2022stabilizer}. %

As we will see, $\mathcal{W}_\alpha$ allows us to witness the nonstabilizerness or magic of mixed quantum states, i.e. to what degree states cannot be generated from Clifford operations. Here, Clifford unitaries $U_\text{C}$ are unitaries generated from Hadamard, S-gate and CNOT gates. Pure stabilizer states $\ket{\psi_\text{C}}$ are states generated by applying $U_\text{C}$ onto the $\ket{0}$ state. Mixed stabilizer states are convex mixture of pure stabilizer states, i.e. $\rho_\text{C}=\sum_i p_i\vert\psi_\text{C}^{(i)}\rangle\langle\psi_\text{C}^{(i)}\vert$ with $p_i\geq0$.
$\mathcal{W}_\alpha$ has the following properties for all $\alpha$: i) Invariant under Clifford unitaries $U_\text{C}$, i.e. $\mathcal{W}_\alpha(U_\text{C} \rho U_\text{C}^\dagger)=\mathcal{W}_\alpha(\rho)$,
ii) Additive, i.e. $\mathcal{W}_\alpha(\rho\otimes\sigma)=\mathcal{W}_\alpha(\rho)+\mathcal{W}_\alpha(\sigma)$.
iii) $-2S_2(\rho)\leq\mathcal{W}_\alpha(\rho)\leq n\ln2-2S_2(\rho)$.
$\mathcal{W}_\alpha$ is not a magic monotone~\cite{veitch2014resource}, notably because it can be non-positive for magic states and increase under Clifford operations as shown in \SM{}~\ref{sec:monoton}.

However, $\mathcal{W}_\alpha(\rho)$ is a genuine witness of magic for mixed states for any $\alpha\geq1/2$. In particular, whenever $\mathcal{W}_\alpha(\rho)>0$, $\rho$ is not a stabilizer state $\rho_\text{C}$, i.e. cannot be written as a convex mixture of stabilizer states~\cite{dai2022detecting,warmuz2024magic,macedo2025witnessingmagicbellinequalities,macedo2025witnessing}. Further, any mixed stabilizer state $\rho_\text{C}$ must have $\mathcal{W}_\alpha(\rho_\text{C})\leq0$.
To see this, we note that
\begin{equation}
    \mathcal{W}_{1/2}(\rho) = 2 \ln \mathcal{D}(\rho)\,, 
\end{equation}
where $\mathcal{D}(\rho)\equiv A_{1/2}(\rho)$ is the stabilizer norm~\cite{campbell2011catalysis,howard2017robustness}\,.
The stabilizer norm is a witness of mixed-state magic as $\ln \mathcal{D}(\rho) >0 $ guarantees that $\rho$ is a nonstabilizer state~\cite{howard2017robustness,rall2019simulation}.
Then, via the hierarchy of R\'enyi entropies for $M_\alpha$, we have
\begin{equation}
    2\ln \mathcal{D}(\rho) \geq M_\alpha(\rho) - 2S_2(\rho)\equiv \mathcal{W}_\alpha(\rho) \quad (\alpha \geq 1/2)\,.
\end{equation}
Further, we define a variant of our witness which we call the filtered $\alpha$-magic witness (see \SM{}~\ref{sec:alt_witness})
\begin{equation}
    \Tilde{\mathcal{W}}_\alpha= \frac{1}{1-\alpha}\ln(\frac{2^nA_\alpha-1}{2^n-1}) + \frac{1-2\alpha}{1-\alpha}\ln(\frac{2^n \text{tr}(\rho^2)-1}{2^n-1})
\end{equation}
which is more sensitive to magic than $\mathcal{W}_\alpha$, though it has the same asymptotic scaling.

Notably, our magic witnesses provide bounds on genuine magic monotones.
First, the log-free robustness of magic~\cite{howard2017robustness,liu2022many} is given by
\begin{equation}\label{eq:robustness}
\text{LR}(\rho)=\text{min}_x\ln\left( \sum_i \vert x_i \vert: \rho=\sum_i x_i \vert\psi_{\text{C}}^{(i)}\rangle\langle\psi_{\text{C}}^{(i)}\vert\right)\,.
\end{equation}
In particular, we have
\begin{equation}\label{eq:robustnessbound}
    2\text{LR}(\rho) \geq 2 \ln\mathcal{D}(\rho)\geq \mathcal{W}_\alpha(\rho) \quad (\alpha \geq 1/2)\,,
\end{equation}
and similarly for $\Tilde{\mathcal{W}}_\alpha(\rho)$.
Additionally, $\mathcal{W}_\alpha$ also relates to another magic monotone, namely the stabilizer fidelity for mixed states~\cite{bravyi2019simulation,rubboli2024mixed}
\begin{equation} 
    D_\text{F}(\rho)=\min_{\sigma\in \text{STAB}}-\ln \mathcal{F}(\rho,\sigma)
\end{equation}
where we have the Uhlmann fidelity
    $\mathcal{F}(\rho,\sigma)=\text{tr}(\sqrt{\rho\sigma})^2$~\cite{baldwin2023efficiently}
and $D_\text{F}(\rho)\leq \text{LR}(\rho)$ (see \SM{}~\ref{sec:boundLR}). We find that whenever $D_\text{F}(\rho)=\omega(\log n)$, then $-\ln A_\alpha(\rho)=\omega(\log n)$ for $\alpha\geq2$ (see \SM{}~\ref{sec:testing} or Ref.~\cite{iyer2024tolerant}).

Importantly, our witnesses go beyond simply determining the presence of magic. In fact, $\mathcal{W}_\alpha$ are quantitative magic witnesses, i.e. 
its specific value provides quantitative information about the amount of magic present in the state~\cite{eisert2007quantitative}. 
Indeed, if $\mathcal{W}_\alpha(\rho)=\omega(\log n)$, then the bound in~\eqref{eq:robustnessbound} implies $\text{LR}(\rho)=\omega(\log n)$. Further, if $\mathcal{W}_\alpha(\rho)=O(\log n)$, then $D_\text{F}(\rho)=O(\log n)$. Therefore, our witnesses allow one not only to detect the magic, but also to infer whether the state is a low-magic or high-magic state, as quantified by a genuine magic monotone. %

Notably, for odd $\alpha$ one can efficiently measure $A_\alpha$~\cite{haug2023efficient} using Bell measurements, which to our knowledge makes $\mathcal{W}_\alpha$ and $\Tilde{\mathcal{W}}_\alpha$ (with odd $\alpha>1$) the first efficiently computable witnesses of magic:
\begin{theorem}[Efficient witness of magic (\SM{}~\ref{sec:measwitness})]\label{thm:witness}
For a given (mixed) $n$-qubit state $\rho$ and odd $\alpha$, there exist an efficient algorithm to measure $A_\alpha(\rho)$ to additive precision $\epsilon$ with failure probability $\delta$ using $O(\alpha\epsilon^{-2}\log(2/\delta))$ copies of $\rho$, $O(1)$ circuit depth, and $O(\alpha n\log(2/\delta))$ classical post-processing time. 
\end{theorem}
The algorithm is provided in \SM{}~\ref{sec:measwitness} or~\cite{haug2023efficient}, which can be implemented using $2$ copies of $\rho$ and Bell measurements. Note that the efficiency of $A_3$ was first shown in Ref.~\cite{gross2021schur}.

\prlsection{Efficient testing of magic}
While witnessing indicates the presence of magic, it is not a necessary criterion, as there are magic states that have non-positive witness. Moreover, the task of witnessing only determines whether the state contains magic, but usually cannot be directly related to quantitative values of a true magic monotone. 

Property tests ask a different question~\cite{buhrman2008quantum}: Does a given state have a property, or is far from it? 
In particular, a test for magic determines whether a state has low magic, or high magic,  where one measures magic in terms of magic monotones. Such tests have been demonstrated for pure states~\cite{iyer2024tolerant,bao2024tolerant,arunachalam2024note,grewal2023improved,leone2024learning}, yet were absent for mixed states with respect to magic monotones.
Witnesses can usually not be used for testing, as they may not detect magic for some classes of high magic states. 
Yet, by leveraging our quantitative witness, we design an efficient algorithm to test for magic which reliably distinguishes states depending on their magic.  Here, we give an efficient test for magic whenever $S_2(\rho)=O(\log n)$:
\begin{theorem}[Efficient testing of magic]\label{thm:testing}
Let $\rho$ be an $n$-qubit state with $S_2(\rho)=O(\log n)$ where it is promised that
\begin{align*}
\mathrm{either}\quad (a)& \,\,\mathrm{LR}(\rho)=O(\log n) \,\,\mathrm{and}\,\,D_\mathrm{F}(\rho)=O(\log n) \,,\\
\mathrm{or}\quad (b)& \,\, \mathrm{LR}(\rho)=\omega(\log n)\,\,\mathrm{and}\,\, D_\mathrm{F}(\rho)=\omega(\log n)\,.
\end{align*} 
Then, there exist an efficient quantum algorithm to distinguish case ($a$) and ($b$) using $\mathrm{poly}(n)$ copies of $\rho$ with high probability. 
\end{theorem}
We sketch the proof idea in the following, with the actual steps shown in \SM{}~\ref{sec:testing}.
First, we show that case ($a$) has $\mathcal{W}_3(\rho)=O(\log n)$, while case ($b$) has $\mathcal{W}_3(\rho)=\omega (\log n)$. Case ($a$) can be shown from our bound on $\text{LR}$ of~\eqref{eq:robustnessbound}, while case ($b$) is derived using bounds on the stabilizer fidelity of Ref.~\cite{iyer2024tolerant}.
Then, efficiency follows from  the algorithm to estimate $A_3(\rho)$ from Ref.~\cite{haug2023efficient} and Hoeffding's inequality.

Our test is tight with respect to $S_2$, as there can be no efficient tests when $S_2=\omega(\log n)$~\cite{bansal2024pseudorandomdensitymatrices}.
Note that for technical reasons we define magic in terms of both $\text{LR}$ and $D_\text{F}$.
Likely, our statements can be relaxed to $D_\text{F}$ only. 
The missing step is to find a lower on $D_\text{F}$ in terms of $\mathcal{W}_3$. Note that such bounds have been proven considering a relaxed definition of $D_\text{F}$ which is limited to pure stabilizer states only~\cite{iyer2024tolerant}, however this measure is not a magic monotone.

\prlsection{Certifying noisy T-states}
To run universal quantum computers, one requires magic resource states as input, such as T-states $\ket{T}=\frac{1}{\sqrt{2}}(\ket{0}+e^{-i\pi/4}\ket{1})$~\cite{bravyi2005universal}. Usually, as they cannot be directly prepared in a fault-tolerant manner due to the Eastin-Knill theorem~\cite{eastin2009restrictions}, they have to be distilled from many noisy T-states~\cite{bravyi2005universal}. T-states are a valuable resource that is expensive to generate. Thus, it is essential to be able to verify their correct preparation, especially when they are subject to noise.
We now show that one can indeed certify the number of T-states even when they are subject to a quite general class of noise channels, namely mixed unital Clifford channels: They are are given by $\Lambda_\text{C}(\rho)=\sum_{i} p_i U_\text{C}^{(i)} \rho {U_\text{C}^{(i)}}^\dagger$
where $\sum_i p_i=1$, $p_i\geq0$, and $U_\text{C}^{(i)}$ are Clifford unitaries. This  includes Pauli channels $\sum_i p_i P_i\rho P_i^\dagger$ with Pauli $P_i$ as a special case.
\begin{proposition}[Certifying noisy magic states (\SM{}~\ref{sec:cliffUnital})]\label{thm:certify}
There is an efficient quantum algorithm to certify whether given noisy $n$-qubit quantum state $\rho_t=\Lambda_\text{C}((\ket{T}\bra{T})^{\otimes t}\otimes (\ket{0} \bra{0})^{\otimes(n-t)})$ contains
either $(a)$ $t=O(\log n)$ or $(b)$ $t=\omega(\log n)$ T-states $\ket{T}=\frac{1}{\sqrt{2}}(\ket{0}+e^{-i\pi/4}\ket{1})$ when it is subject to arbitrary $n$-qubit mixed unital Clifford channels $\Lambda_\text{C}(\rho)$ and $S_2(\rho_t)=O(\log n)$. 
\end{proposition}
The proof for Prop.~\ref{thm:certify} is provided in \SM{}~\ref{sec:cliffUnital} and gives the explicit testing algorithm. To prove Prop.~\ref{thm:certify}, we first show that the robustness of magic upper bounds the number of T-states in terms of our witness. 
Then, we  also find a lower bound on the number of T-states via our witness. This involves showing that $-\ln A_3(\rho)$ can only increase under mixed unital Clifford channels. Finally, the number of T-states can be efficiently bounded by estimating $-\ln A_3(\rho)$, which we show via Hoeffding's inequality.

\prlsection{Noise-robustness of many-body states}
We have shown how to witness and test magic, where the efficient measurement requires a bounded entropy $S_2=O(\log n)$. The condition on entropy arises due to inherent inefficiency of testing highly mixed states, i.e. $S_2=\omega(\log n)$~\cite{wright2016learn,bansal2024pseudorandomdensitymatrices}. But does this imply that magic is destroyed in such high entropy states?
It turns out this is not the case: We find that even for exponentially strong noise and high entropy, magic can still survive.

Let us consider typical states drawn from the Haar measure are known to be highly magical~\cite{turkeshi2023pauli}. They can be approximated by random local circuits~\cite{brandao2016local,arute2019quantum}, which have been implemented on noisy intermediate scale quantum (NISQ) computers~\cite{arute2019quantum,bharti2022noisy,wang2021noise}.
However, such noisy quantum computers are subject to substantial noise. For simplicity, let us consider global depolarisation noise $\Gamma(\rho)=(1-p)\rho+p I_n/2^n$ with noise probability $p$. Further, we assume that the system is subject to noise with exponentially high probability, i.e.
\begin{equation}
    p= 1-2^{-\beta n}\,.
\end{equation}
with some factor $\beta$~\cite{bharti2022noisy,wang2021noise}. 
Now, by examining the scaling of the witness in the limit $n\gg1$ in \SM{}~\ref{sec:witness_noise}, we find that our filtered witness $\Tilde{\mathcal{W}}_\alpha>0$ whenever $\beta<1/2$.
This implies that magic can be highly robust to noise, surviving even for exponential noise and large number of qubits $n$. Further, we note that our efficient witness $\Tilde{\mathcal{W}}_3$ can certify magic with the same sensitivity as the intractable stabilizer norm $\mathcal{D}$.

While we so far assumed a simplified state and noise model, we find numerically that similar robustness holds also in settings closer to current experiments. 
In particular, we consider  random local circuits under local depolarising noise, a typical setup for NISQ devices~\cite{arute2019quantum,bharti2022noisy,chen2023complexity}.
In Fig.~\ref{fig:localdepol}, we study $d$ layers of random single-qubit rotations with CNOT gates arranged in a linear chain, where after each gate we apply single-qubit depolarisation noise $\Gamma(\rho)=(1-p)\rho+p \text{tr}_1(\rho) I_1/2$ with probability $p$ (see \SM{}~\ref{sec:witness_noise_local}). 
Curiously, we find in Fig.~\ref{fig:localdepol}a that for constant $p$, $\Tilde{\mathcal{W}}_\alpha$ is positive until a critical circuit depth $\tilde{d}_\text{c}$, which is independent of qubit number $n$. We plot $\tilde{d}_\text{c}$ in in Fig.~\ref{fig:localdepol}b against $p$, finding nearly linear scaling as $d_\text{c}\propto p^{-\eta}$, with $\eta\approx 0.96$. 
Thus, NISQ devices can robustly generate magic that is also efficiently verifiable. 
\begin{figure}[htbp]
	\centering	
    \subfigimg[width=0.24\textwidth]{a}{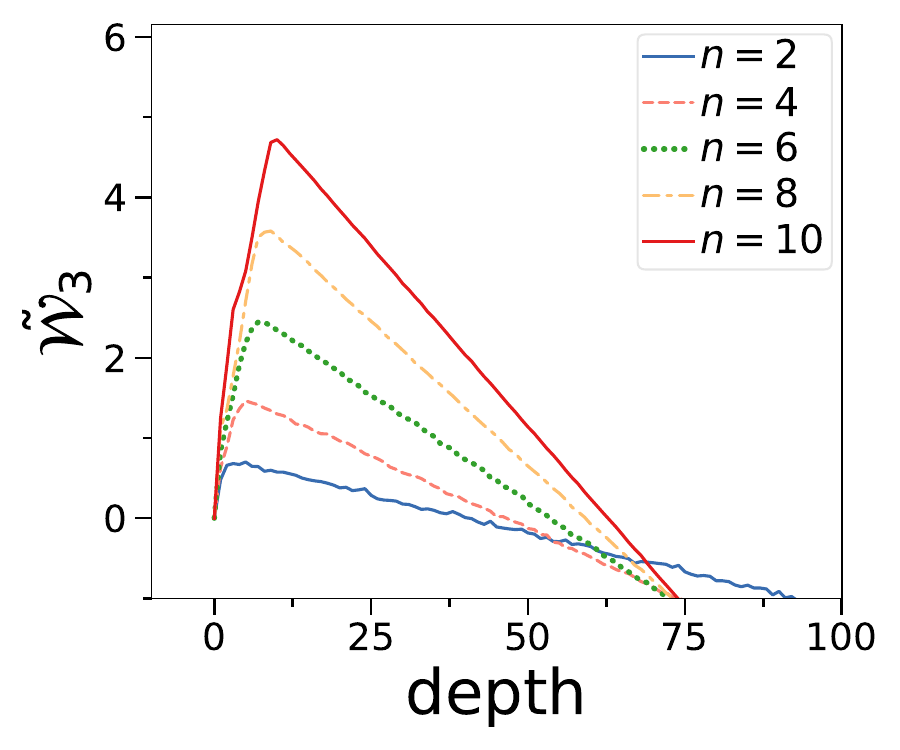}\hfill%
    \subfigimg[width=0.24\textwidth]{b}{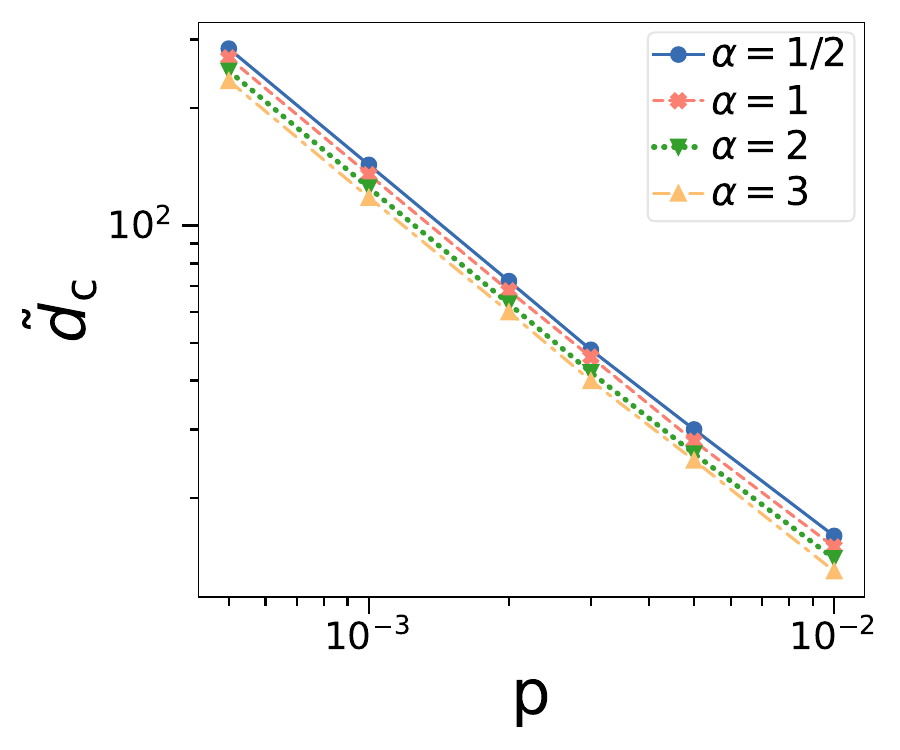}
    \caption{Magic of random local circuit of $d$ layers with local depolarising noise probability $p$. \idg{a} Efficient filtered magic witness $\Tilde{\mathcal{W}}_{3}$ against $d$ for different $n$ and $p$. \idg{b} Critical circuit depth $\tilde{d}_\text{c}$ where $\Tilde{\mathcal{W}}_\alpha$ becomes negative for different $p$ and $n=8$. By fitting we find $d_\text{c}\propto p^{-\eta}$ with $\eta\approx 0.96$. Each datapoint is averaged over 10 random circuit instances.
	}
	\label{fig:localdepol}
\end{figure}

\prlsection{Magic of noisy quantum computers} 
Now, we demonstrate our witness and the noise-robustness of magic in experiment. In Fig.~\ref{fig:Ionq}, we characterize noisy quantum states on the IonQ quantum computer. We study random random Clifford circuits $U_\text{C}$ interleaved with $N_\text{T}$ T-gates $T=\text{diag}(1,e^{-i\pi/4})$~\cite{haferkamp2022efficient,haug2023efficient}
\begin{equation}\label{eq:random_CliffT}
\ket{\psi(N_\text{T})}=U_\text{C}^{(0)}[\prod_{k=1}^{N_\text{T}} (T\otimes I_{n-1}) U_\text{C}^{(k)} ]\ket{0}\,.
\end{equation}
For experimental convenience, we compressed the circuits such that for all $N_\text{T}$ we have the same circuit depth and thus similar purity.
We show experimentally measured $\mathcal{W}_3$ in Fig.~\ref{fig:Ionq}a. The experiment closely matches our simulation assuming global depolarising noise $\Gamma(\rho)=(1-p)\rho+pI/2^n$ with noise strength $p$, where we extract $p$ from the experiment for each circuit. Although we have significant noise $p\approx0.2$, we can experimentally certify the presence of genuine mixed-state magic for all $N_\text{T}>0$.

In Fig.~\ref{fig:Ionq}b, we study our magic witness $\mathcal{W}_\alpha$ for different $\alpha$ and the log-free robustness of magic $\text{LR}$ by simulating the circuits of the experiment with depolarisation noise. We confirm the inequalities $2\text{LR}\geq \mathcal{W}_{1/2}\geq \mathcal{W}_{1}\geq\mathcal{W}_{2}\geq \mathcal{W}_{3}$. Notably, all measures consistently certify the magic of the noisy states. One can also equivalently use the filtered witness, which we find is more sensitive to T-gates.
We also experimentally certify the magic of important magic resource states such as the T-state in \SM{}~\ref{sec:experiment_sup}.
\begin{figure}[htbp]
	\centering	
	\subfigimg[width=0.24\textwidth]{a}{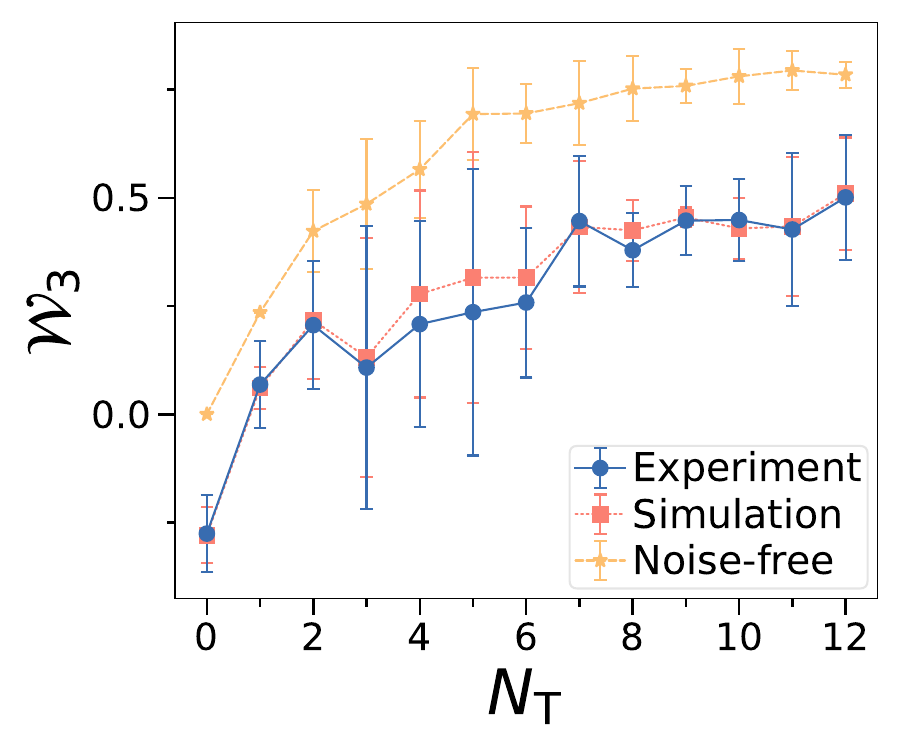}\hfill
    \subfigimg[width=0.24\textwidth]{b}{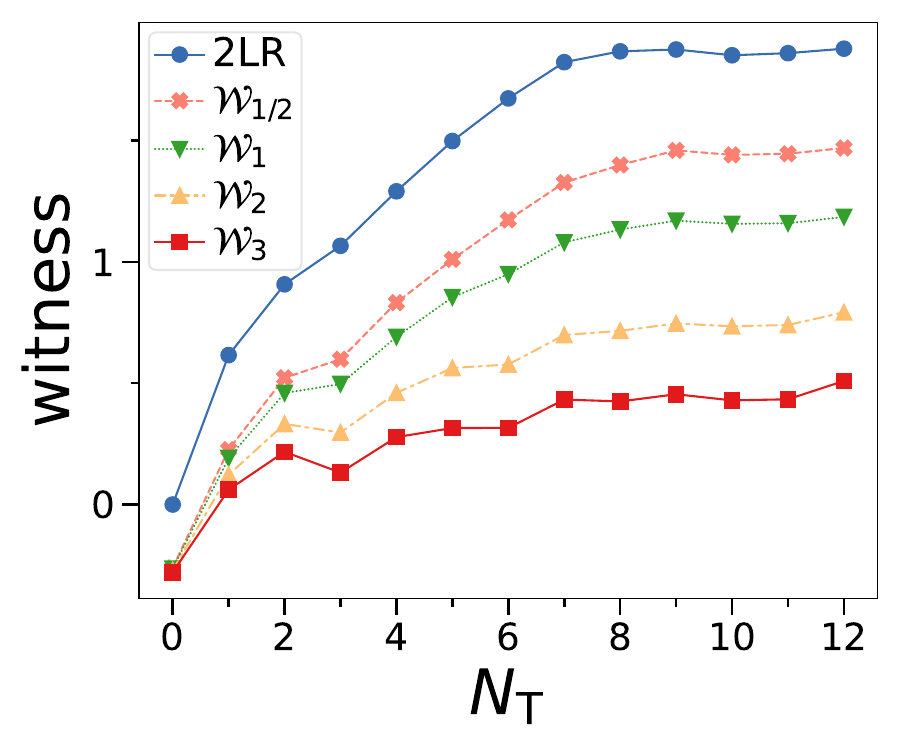}
    \caption{\idg{a} Magic witness $\mathcal{W}_3$ measured on the IonQ quantum computer for Clifford circuits doped with $N_\text{T}$ T-gates. We have $n=3$ qubits and estimate an effective average global depolarisation noise $p\approx 0.2$ from the purity. We show experimental result in blue and simulation with global depolarisation noise in orange, while yellow line is a simulation for noise-free states. We average over 10 random initialization of the circuit. 
    \idg{b} Simulation of $\mathcal{W}_\alpha$ for different $\alpha$ and log-free robustness of magic $\text{LR}$ with global depolarisation noise with parameters of experiment.
	}
	\label{fig:Ionq}
\end{figure}

\prlsection{Witness magic of many-body systems}
So far we considered witness as tool for quantum computers. Beyond, magic has become also an important method to study the complexity of many-body quantum systems, both numerically and analytically~\cite{leone2022stabilizer,haug2022quantifying,haug2023stabilizer,lami2023quantum,tarabunga2024nonstabilizernessmps,turkeshi2023pauli,tarabunga2025efficientmutualmagicmagic}. 
While for pure states genuine magic monotones have been used extensively~\cite{leone2024stabilizer}, so far for mixed states no efficient way to detect magic has been known. This has made the study of magic in entangled subsystem of extensive quantum systems difficult. For example, the interplay between entanglement and (non-local) magic has remained poorly understood~\cite{korbany2025longrangenonstabilizerness,wei2025longrangenonstabilizerness}.

We now find that our witness $\mathcal{W}_\alpha$ with integer $\alpha>0$, as well as its filtered variant $\tilde{\mathcal{W}}_\alpha$, can be efficiently computed for any mixed subsystem of MPS:
\begin{theorem}[Efficient witness of magic for MPS (\cite{haug2022quantifying,haug2023stabilizer,lami2023quantum,tarabunga2024nonstabilizernessmps})]\label{thm:witnessMPS}
Given MPS $\ket{\psi}$ with bond dimension $\chi$ and its $n$-qubit subsystem $\rho=\text{tr}_{\bar{n}}(\ket{\psi}\bra{\psi})$, then there exist efficient classical algorithms to compute $\mathcal{W}_\alpha(\rho)$ and $\Tilde{\mathcal{W}}_\alpha(\rho)$ for $\alpha=1$ in $O(n\chi^3\epsilon^{-2})$ time and $\epsilon$ additive precision, as well as exactly for integer $\alpha>1$ in $O(n\chi^{6\alpha})$ time.
\end{theorem}
Here, $\text{tr}_{\bar{n}}(.)$ is the partial trace over the complement of $n$ qubits, which we assume to be connected to the left or right boundary of the MPS.
The algorithms for MPS used in Thm.~\ref{thm:witnessMPS} have been proposed previously to compute $(1-\alpha)^{-1}\ln A_\alpha$ and can be straightforwardly extended to $\mathcal{W}_\alpha$.
In particular, 
\begin{equation}
\mathcal{W}_1(\rho)=-\sum_{P\in\mathcal{P}_n} 2^{-n}\frac{\text{tr}(\rho P)^2}{\text{tr}(\rho^2)} \ln(\text{tr}(\rho P)^2)-2S_2(\rho)
\end{equation}
can be efficiently calculated using Pauli sampling~\cite{haug2023stabilizer,lami2023quantum}, and integer $\alpha>1$ can be computed exactly using MPS replica tricks~\cite{haug2022quantifying,tarabunga2024nonstabilizernessmps}. Approximations via various Monte-Carlo sampling approaches can be used as well~\cite{tarabunga2025efficientmutualmagicmagic}.

We now apply our witness to study the magic of subsystems of many-body states $\ket{\psi}$. As subsystems $\rho=\text{tr}_A(\ket{\psi})$ are usually entangled and thus mixed, it has been difficult to determine whether they contain magic or it has been destroyed by the partial trace. Using our witness~\eqref{eq:witnessExp}, we can now efficiently detect the presence of magic in such subsystems. In particular, in Fig.~\ref{fig:witness_mps} we consider the  ground state of the transverse-field Ising model
\begin{equation}\label{eq:ising}
    H_\text{TFIM}=-\sum_{k=1}^{n-1}\sigma^x_k\sigma^x_{k+1}-h\sum_{k=1}^n \sigma_k^z
\end{equation}
where we have the field $h$ and choose open boundary conditions. 
The ground state is known to be a nonstabilizer state except at $h=0$ and $h\to \infty$. We compute $\mathcal{W}_2$ in subsystems of length $\ell$ with $\rho_\ell=\text{tr}_{\bar{\ell}}(\ket{\psi})$, where $\text{tr}_{\bar{\ell}}(.)$ is the partial trace over the complement of $\ell$. Note that, if $\mathcal{W}_2$ detects the magic in a subsystem, this also implies that the full state is a nonstabilizer state. 
In Fig.~\ref{fig:witness_mps}a, we find that $\mathcal{W}_2$ increases with subsystem size $\ell$. Further, $\mathcal{W}_2$ increases with field $h$ until the critical field $h=1$, and then decreases. Notably, for $h<1$ the witness becomes greater than zero a particular $\ell\geq \ell_\text{c}$, which in \SM{}~\ref{sec:manybody} we determine to follow $\ell_\text{c}\approx n^{-h+1}$. In contrast, for any $h\leq1$ the witness is non-zero already for constant $\ell$. This is the result of the different entanglement structure for $h<1$ and $h>1$.
In Fig.~\ref{fig:witness_mps}b, we plot $\mathcal{W}_2$ for $h=1$ for different $n$. We find that the witness is a robust feature that is nearly independent of $n$ and increases linearly with $\ell$. 

In general, when one considered a subsystem of a pure state, the second term in~\eqref{eq:witnessExp} corresponds to the entanglement of the subsystem. Therefore, the witness consists of two competing contributions: one which relates to magic in $\ln A_\alpha$ and another from the entanglement entropy $S_2$. The witness successfully detects magic when $\ln A_\alpha$ dominates. Given that ground states of local, gapped Hamiltonians exhibit area law scaling for entanglement and volume law scaling for magic, the witness proves particularly effective in detecting magic within these physically relevant systems.

\begin{figure}[htbp]
	\centering	
	\subfigimg[width=0.24\textwidth]{a}{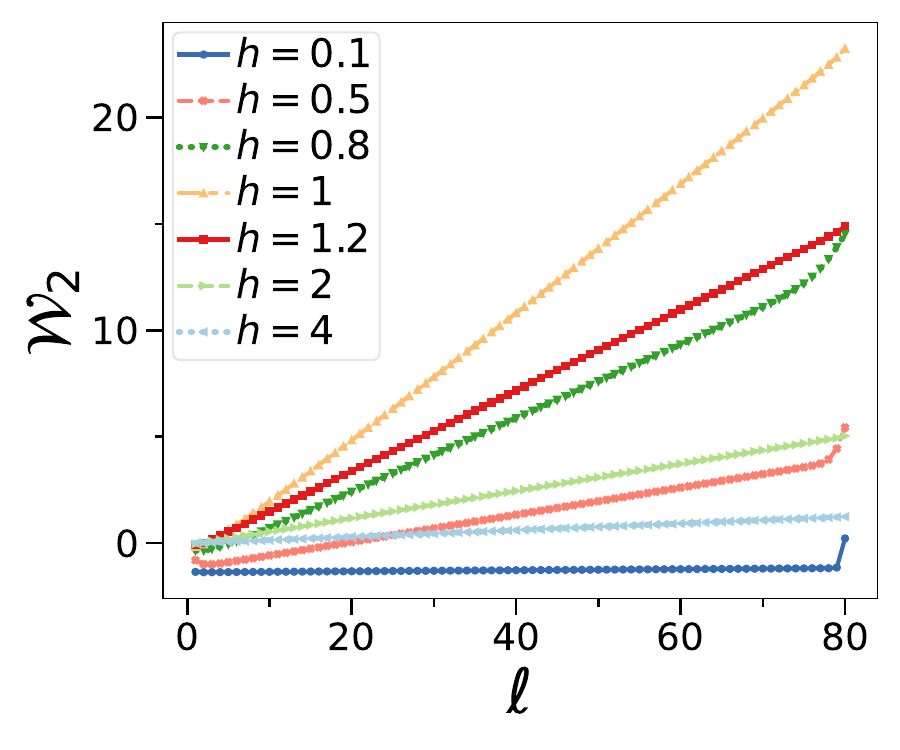}\hfill
    \subfigimg[width=0.24\textwidth]{b}{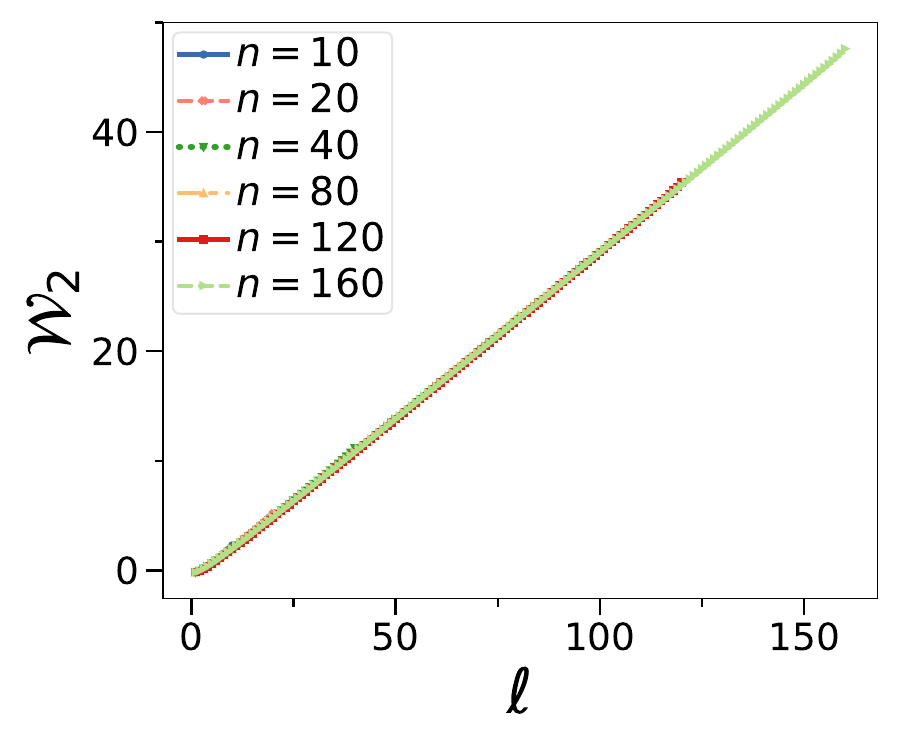}
    \caption{2-SRE magic witness $\mathcal{W}_2(\rho_\ell)$ for reduced density matrix $\rho_\ell$ of size $\ell$ of the groundstate of the TFIM~\eqref{eq:ising} for different fields $h$.
    \idg{a} $\mathcal{W}_2$ against $\ell$ for different $h$ with $n=80$ qubits and $\chi=10$.
    \idg{b} $\mathcal{W}_2$ against $\ell$ for different $n$ for $h=1$ and $\chi=10$.
	}
	\label{fig:witness_mps}
\end{figure}

\prlsection{Pseudorandom density matrices and pseudomagic}
Finally, we discuss the implications of our results in quantum cryptography~\cite{ji2018pseudorandom}.
A key task in cryptography is to efficiently create complex quantum states with as little resources as possible~\cite{bouland2022quantum,haug2023pseudorandom}. 
In particular, random quantum states are required for many applications, yet generating Haar random states and random mixed states requires exponentially deep circuits in general~\cite{nielsen2011quantum}. However, true randomness is often not required. 
Recently, pseudorandom states (PRSs) and the more general pseudorandom density matrices (PRDMs) have been proposed. PRDMs can be efficiently prepared, yet are indistinguishable from truly random mixed states for any efficient quantum algorithm~\cite{bansal2024pseudorandomdensitymatrices}. They allow for the low-cost preparation of pseudorandomness, with only a very low circuit depth~\cite{schuster2024randomunitariesextremelylow}.
However, what are the minimal resources needed to generate pseudorandomness?
For entropy $S_2=0$, it has been shown that one requires $\Omega(n)$ T-gates to prepare PRDMs~\cite{grewal2023improved}, however for the mixed case no general bound has been known. 
Here, we find that for PRDMs with $S_2=O(\log n)$, at least $\omega(\log n)$ T-gates are required:
\begin{proposition}[T-gates for PRDM (\SM{}~\ref{sec:Tgate})]
    Any family of circuits consisting of Clifford operations and $N_\text{T}$ T-gates requires $N_\text{T}=\omega(\log n)$ to prepare PRDMs with entropy $S_2=O(\log n)$.
\end{proposition}
This result follows from our efficient tester of Thm.~\ref{thm:testing} and is fully shown in \SM{}~\ref{sec:Tgate}: A (hypothetical) PRDM consisting of only $N_\text{T}=O(\log n)$ T-gates has low magic and thus can be efficiently distinguished from truly random mixed states, which are highly magical whenever $S_2=O(\log n)$.

A related task is to mimic high resource states by using only very little resources, which is called pseudoresources~\cite{bouland2022quantum,haug2023pseudorandom}. 
Here, in the context of magic, pseudomagic state ensembles have been proposed~\cite{gu2023little,haug2023pseudorandom}: They are two efficiently preparable ensembles which are indistinguishable for any efficient quantum algorithm, yet have widely different magic (see \SM{}~\ref{sec:pseudomagic}). In particular, ensemble ($a$) has high magic $f(n)$, while ($b$) has low magic $g(n)$~\cite{bouland2022quantum,gu2023little,haug2023pseudorandom}. The pseudomagic gap is the gap $f(n)$ vs $g(n)$ of such pseudomagic ensembles.
Such pseudomagic ensemble can mimic states with high magic by only using only very little magic resource themselves.

Recently, pseudomagic has also been related to encryption~\cite{haug2025pseudorandom}: Let us assume we want to securely transmit quantum states such that no eavesdropper can learn anything about the state or about the encryption process. To encrypt, one scrambles the state such that it becomes indistinguishable (for any efficient eavesdropper) from a random state. While scrambling can hide all information about the state itself, it turns out that meta-information about the encryption process can still be leaked. In particular, the pseudomagic gap determines how much an eavesdropper can learn about the magic of the state and the encryption process. A maximal pseudomagic gap ($f(n)=\Theta(n)$ vs $g(n)=0$) implies that no information about magic is leaked, while a sub-maximal gap implies leakage of information~\cite{haug2025pseudorandom}. In the latter case, an eavesdropper could potentially harvest magic from the encrypted state via black-box magic distillation~\cite{gu2023little,bansal2024pseudorandomdensitymatrices}.

Now, we ask what is the maximal possible pseudomagic gap $f(n)$ vs $g(n)$?
For pure states with $2$-R\'enyi entropy $S_2=0$, the pseudomagic is $f(n)=\Theta(n)$ vs $g(n)=\omega(\log n)$~\cite{gu2023little}. In contrast, highly mixed states with $S_2=\omega( \log n)$ can have the maximal possible gap of $f(n)=\Theta(n)$ vs $g(n)=0$~\cite{bansal2024pseudorandomdensitymatrices}, while the question for $S_2=O(\log n)$ has been an open problem.
Here, we find that states with $S_2=O(\log n)$ have in fact the same pseudomagic gap as completely pure states:
\begin{proposition}[Pseudomagic of mixed states (\SM{}~\ref{sec:pseudomagic})]
Pseudomagic state ensembles with $2$-R\'enyi entropy $S_2=O(\log n)$ can have a pseudomagic gap $f(n)=\Theta(n)$ vs $g(n)=\omega(\log n)$. 
\end{proposition}
\begin{proof}
We prove via contradiction: Let us assume there exist a pseudomagic state ensemble with magic $f(n)=\Theta(n)$ vs $g(n)=\Theta(\log n)$. Here, we measure magic in terms of log-free robustness of magic $\text{LR}$ and stabilizer fidelity $D_\text{F}$. Then, due to Thm.~\ref{thm:testing}, there exist an efficient quantum algorithm to distinguish states with $O(\log n)$ magic and $\omega(\log n)$ magic whenever $S_2=O(\log n)$.  Thus, this pseudomagic state ensemble cannot exist and we must have a pseudomagic gap of $f(n)=\Theta(n)$ vs $g(n)=\omega(\log n)$.
\end{proof}
Thus, for limited entropy $S_2=O(\log n)$, one requires $\omega(\log n)$ magic to mimic high magic states. Only with extensive entropy $S_2=\omega(\log n)$ one can achieve the maximal pseudomagic gap. Thus, extensive entropy is needed as resource for a maximal pseudomagic gap and prevent leakage of information about magic to eavesdroppers.

\prlsection{Discussion} \label{sec:discussion}
We demonstrate that the magic of mixed states can be efficiently witnessed and tested. 
As our main tool, we introduce efficient witnesses of magic $\mathcal{W}_\alpha$ and filtered witness $\Tilde{\mathcal{W}}_\alpha$ with $\alpha\geq1/2$, where $\mathcal{W}_\alpha(\rho)>0$ or $\Tilde{\mathcal{W}}_\alpha(\rho)>0$ guarantees that $\rho$ has magic. Beyond that, they offer useful quantitative information about the degree of magic possessed by the state, since they provide rigorous bounds on $\text{LR}(\rho)$ and $D_\text{F}(\rho)$. Notably, our witnesses are quantitative witnesses~\cite{eisert2007quantitative}. This is a stronger notion of witnessing that allows for quantitative statements about the amount of magic.
Furthermore, our witnesses allow for the analysis of smaller subsystems, thus significantly reducing experimental demands. This parallels the significant role of the $p_3$-PPT condition to detect mixed state entanglement~\cite{elben2020mixed}. 

We provide tests to determine whether states have $O(\log n)$ or $\omega(\log n)$ magic, which are efficient as long as the state is not too mixed, i.e. its $2$-R\'enyi entropy is bounded by $S_2=O(\log n)$. This bound  is tight, as testing is inherently inefficient for $S_2=\omega(\log n)$~\cite{bansal2024pseudorandomdensitymatrices}.
Our work resolves the long-standing question of testing as function of $S_2$~\cite{montanaro2013survey} (see Tab.~\ref{tab:testing}), where previously testing algorithms for magic monotones were only known for completely pure states, i.e. $S_2=0$~\cite{gross2021schur,haug2022quantifying}.
We note that our work defines magic in terms of $\text{LR}$ and $D_\text{F}$, which usually have the same scaling in $n$ (e.g. for Clifford+T circuits), but for special types of states can behave differently~\cite{haug2024probing}.
Likely, testing efficiency solely depends on $D_\text{F}$, and future work could find a lower bound on $D_\text{F}$ in terms of $\mathcal{W}_\alpha$ in order to prove this. 

With our witnesses, we experimentally certify the magic of mixed states on the IonQ quantum computer. Previous magic measures were only well defined for pure quantum states~\cite{gross2021schur,haug2022scalable,tirrito2023quantifying,turkeshi2023measuring} and thus did not unambiguously certify whether magic has actually been prepared in experiment~\cite{haug2022scalable,oliviero2022measuring,niroula2023phase,haug2023efficient,bluvstein2024logical}. 
In fact, our witnesses provide to our knowledge the first efficient, robust and scalable way to experimentally witness magic. 
We believe our tools are crucial to enable quantum computing: We can robustly certify the number of noisy T-states, the key resource to implement T-gates which enable universality in fault-tolerant quantum computers. In particular, we can efficiently verify the number of T-states even when subject to a quite general class of noise models. This task is essential in large-scale quantum systems where the noise is not well characterized.

We find that magic is surprisingly robust to noise. We observe a high degree of magic in our experiment even under physical noise. 
Further, we show that magic of random quantum states such as generated by local random circuits is highly robust to noise, being able to tolerate exponentially high depolarization noise. We find that there is a critical circuit $d_\text{c}$, below which magic can be witnessed independent of qubit number.
Thus, NISQ devices can serve as scalable generators of magic~\cite{bharti2022noisy,arute2019quantum,chen2023complexity}.

Our magic witness also enables the study of magic in mixed many-body systems, a regime for which so far no efficient methods existed to our knowledge. In particular, we efficiently compute the magic of (entangled) subsystems of MPS for any integer $\alpha\geq1$ (e.g. $O(n\chi^3\epsilon^{-2})$ for $\alpha=1$). %
We demonstrate that even small (mixed) subsystems of the ground state of the TFIM contain substantial amount of magic, even close to the critical point where there is increased entropy due to entanglement. Magic is detected only beyond a particular subsystem size, which depends on the transverse field. %
We find that there is an intrinsic competition between the moment of the Pauli spectrum $(1-\alpha)^{-1}\ln A_\alpha$ that increases magic, and entropy $S_2$ that decreases magic. For area law states such as ground states of local gapped 1D Hamiltonians where $S_2=O(\log n)$, a volume-law Pauli spectrum $(1-\alpha)^{-1}\ln A_\alpha\sim \omega(\log n)$ thus directly implies the presence of magic.

Our results also have direct implications on quantum cryptography: We show that the pseudomagic gap, i.e. the capability of low-magic states to mimic high-magic states, is fundamentally constrained for limited entropy $S_2=O(\log n)$. In particular, we find a pseudomagic gap of $f(n)=\Theta(n)$ vs $g(n)=\omega(\log n)$, matching the one previously found for completely pure states~\cite{gu2023little}. The optimal gap is only achieved for high entropy $S_2=\omega(\log n)$, where one finds $f(n)=\Theta(n)$ vs $g(n)=0$~\cite{bansal2024pseudorandomdensitymatrices}. We resolve the question of the pseudomagic gap as a function of entropy, with the complete summary given in Tab.~\ref{tab:pseudoresource}.
While usually entropy in quantum states is associated negatively with noise and destroying quantumness, we find the opposite is true in quantum cryptography: Here, entropy is an important resource needed to hide information about quantum resources such as magic from eavesdroppers~\cite{haug2025pseudorandom}.
To completely hide magic, one requires $S_2=\omega(\log n)$, which we find is tight.

Finally, we note that, while the property testing of pure states has been extensively studied, relatively little work relates to the testing of mixed states~\cite{montanaro2013survey,bansal2024pseudorandomdensitymatrices}. 
Our work shows that magic is indeed testable as long as the entropy is bounded, while future work could study whether the same holds also for other properties such as entanglement or coherence. Further, it would be interesting to study testing via single-copy state access~\cite{hinsche2024single}, where the gaps between learning tasks on mixed and pure states are known~\cite{liu2024exponential}.
In addition, one could optimize our witnesses to other kind of moments of the Pauli spectrum, as has similarly been done for mixed state entanglement~\cite{yu2021optimal,neven2021symmetryresolved}. In particular, other witnesses could be constructed by considering boundaries of the stabilizer polytope~\cite{warmuz2024magic}, which may improve over our witness. Moreover, an interesting question arises: is it possible to construct proper mixed state monotones that are also efficiently computable? 
Finally, it would be interesting to extend our analysis to qudits, where for odd and prime local Hilbert space dimensions genuine and computable magic monotones are known~\cite{veitch2014resource,veitch2012negative,tarabunga2024critical}.

\begin{acknowledgments}
We thank Lorenzo Piroli, David Aram Korbany, Marcello Dalmonte, and Emanuele Tirrito for insightful discussions. P.S.T. acknowledges funding by the Deutsche Forschungsgemeinschaft (DFG, German Research Foundation) under Germany’s Excellence Strategy – EXC-2111 – 390814868.
Our MPS simulations have been performed using the iTensor library~\cite{itensor}.

\end{acknowledgments}

\bibliography{bibliography}

\begin{thebibliography}{84}%
\makeatletter
\providecommand \@ifxundefined [1]{%
 \@ifx{#1\undefined}
}%
\providecommand \@ifnum [1]{%
 \ifnum #1\expandafter \@firstoftwo
 \else \expandafter \@secondoftwo
 \fi
}%
\providecommand \@ifx [1]{%
 \ifx #1\expandafter \@firstoftwo
 \else \expandafter \@secondoftwo
 \fi
}%
\providecommand \natexlab [1]{#1}%
\providecommand \enquote  [1]{``#1''}%
\providecommand \bibnamefont  [1]{#1}%
\providecommand \bibfnamefont [1]{#1}%
\providecommand \citenamefont [1]{#1}%
\providecommand \href@noop [0]{\@secondoftwo}%
\providecommand \href [0]{\begingroup \@sanitize@url \@href}%
\providecommand \@href[1]{\@@startlink{#1}\@@href}%
\providecommand \@@href[1]{\endgroup#1\@@endlink}%
\providecommand \@sanitize@url [0]{\catcode `\\12\catcode `\$12\catcode `\&12\catcode `\#12\catcode `\^12\catcode `\_12\catcode `\%12\relax}%
\providecommand \@@startlink[1]{}%
\providecommand \@@endlink[0]{}%
\providecommand \url  [0]{\begingroup\@sanitize@url \@url }%
\providecommand \@url [1]{\endgroup\@href {#1}{\urlprefix }}%
\providecommand \urlprefix  [0]{URL }%
\providecommand \Eprint [0]{\href }%
\providecommand \doibase [0]{https://doi.org/}%
\providecommand \selectlanguage [0]{\@gobble}%
\providecommand \bibinfo  [0]{\@secondoftwo}%
\providecommand \bibfield  [0]{\@secondoftwo}%
\providecommand \translation [1]{[#1]}%
\providecommand \BibitemOpen [0]{}%
\providecommand \bibitemStop [0]{}%
\providecommand \bibitemNoStop [0]{.\EOS\space}%
\providecommand \EOS [0]{\spacefactor3000\relax}%
\providecommand \BibitemShut  [1]{\csname bibitem#1\endcsname}%
\let\auto@bib@innerbib\@empty
\bibitem [{\citenamefont {Bravyi}\ and\ \citenamefont {Kitaev}(2005)}]{bravyi2005universal}%
  \BibitemOpen
  \bibfield  {author} {\bibinfo {author} {\bibfnamefont {S.}~\bibnamefont {Bravyi}}\ and\ \bibinfo {author} {\bibfnamefont {A.}~\bibnamefont {Kitaev}},\ }\bibfield  {title} {\bibinfo {title} {Universal quantum computation with ideal clifford gates and noisy ancillas},\ }\href {https://doi.org/10.1103/PhysRevA.71.022316} {\bibfield  {journal} {\bibinfo  {journal} {Phys. Rev. A}\ }\textbf {\bibinfo {volume} {71}},\ \bibinfo {pages} {022316} (\bibinfo {year} {2005})}\BibitemShut {NoStop}%
\bibitem [{\citenamefont {Veitch}\ \emph {et~al.}(2014)\citenamefont {Veitch}, \citenamefont {Mousavian}, \citenamefont {Gottesman},\ and\ \citenamefont {Emerson}}]{veitch2014resource}%
  \BibitemOpen
  \bibfield  {author} {\bibinfo {author} {\bibfnamefont {V.}~\bibnamefont {Veitch}}, \bibinfo {author} {\bibfnamefont {S.~H.}\ \bibnamefont {Mousavian}}, \bibinfo {author} {\bibfnamefont {D.}~\bibnamefont {Gottesman}},\ and\ \bibinfo {author} {\bibfnamefont {J.}~\bibnamefont {Emerson}},\ }\bibfield  {title} {\bibinfo {title} {The resource theory of stabilizer quantum computation},\ }\href {https://doi.org/10.1088/1367-2630/16/1/013009} {\bibfield  {journal} {\bibinfo  {journal} {New J. Phys.}\ }\textbf {\bibinfo {volume} {16}},\ \bibinfo {pages} {013009} (\bibinfo {year} {2014})}\BibitemShut {NoStop}%
\bibitem [{\citenamefont {Gottesman}(1997)}]{gottesman1997stabilizer}%
  \BibitemOpen
  \bibfield  {author} {\bibinfo {author} {\bibfnamefont {D.}~\bibnamefont {Gottesman}},\ }\emph {\bibinfo {title} {Stabilizer codes and quantum error correction. Caltech Ph. D}},\ \href@noop {} {Ph.D. thesis},\ \bibinfo  {school} {Thesis, eprint: quant-ph/9705052} (\bibinfo {year} {1997})\BibitemShut {NoStop}%
\bibitem [{\citenamefont {Nielsen}\ and\ \citenamefont {Chuang}(2011)}]{nielsen2011quantum}%
  \BibitemOpen
  \bibfield  {author} {\bibinfo {author} {\bibfnamefont {M.~A.}\ \bibnamefont {Nielsen}}\ and\ \bibinfo {author} {\bibfnamefont {I.~L.}\ \bibnamefont {Chuang}},\ }\href@noop {} {\emph {\bibinfo {title} {Quantum Computation and Quantum Information: 10th Anniversary Edition}}}\ (\bibinfo  {publisher} {Cambridge University Press},\ \bibinfo {year} {2011})\BibitemShut {NoStop}%
\bibitem [{\citenamefont {Kitaev}(2003)}]{kitaev2003fault}%
  \BibitemOpen
  \bibfield  {author} {\bibinfo {author} {\bibfnamefont {A.~Y.}\ \bibnamefont {Kitaev}},\ }\bibfield  {title} {\bibinfo {title} {Fault-tolerant quantum computation by anyons},\ }\href {https://doi.org/10.1016/S0003-4916(02)00018-0} {\bibfield  {journal} {\bibinfo  {journal} {Ann. Phys.}\ }\textbf {\bibinfo {volume} {303}},\ \bibinfo {pages} {2} (\bibinfo {year} {2003})}\BibitemShut {NoStop}%
\bibitem [{\citenamefont {Eastin}\ and\ \citenamefont {Knill}(2009{\natexlab{a}})}]{eastin2009restriction}%
  \BibitemOpen
  \bibfield  {author} {\bibinfo {author} {\bibfnamefont {B.}~\bibnamefont {Eastin}}\ and\ \bibinfo {author} {\bibfnamefont {E.}~\bibnamefont {Knill}},\ }\bibfield  {title} {\bibinfo {title} {Restrictions on transversal encoded quantum gate sets},\ }\href {https://doi.org/10.1103/PhysRevLett.102.110502} {\bibfield  {journal} {\bibinfo  {journal} {Phys. Rev. Lett.}\ }\textbf {\bibinfo {volume} {102}},\ \bibinfo {pages} {110502} (\bibinfo {year} {2009}{\natexlab{a}})}\BibitemShut {NoStop}%
\bibitem [{\citenamefont {Howard}\ \emph {et~al.}(2014)\citenamefont {Howard}, \citenamefont {Wallman}, \citenamefont {Veitch},\ and\ \citenamefont {Emerson}}]{howard2014contextuality}%
  \BibitemOpen
  \bibfield  {author} {\bibinfo {author} {\bibfnamefont {M.}~\bibnamefont {Howard}}, \bibinfo {author} {\bibfnamefont {J.}~\bibnamefont {Wallman}}, \bibinfo {author} {\bibfnamefont {V.}~\bibnamefont {Veitch}},\ and\ \bibinfo {author} {\bibfnamefont {J.}~\bibnamefont {Emerson}},\ }\bibfield  {title} {\bibinfo {title} {Contextuality supplies the ‘magic’for quantum computation},\ }\href {https://doi.org/10.1038/nature13460} {\bibfield  {journal} {\bibinfo  {journal} {Nature}\ }\textbf {\bibinfo {volume} {510}},\ \bibinfo {pages} {351} (\bibinfo {year} {2014})}\BibitemShut {NoStop}%
\bibitem [{\citenamefont {Litinski}(2019)}]{litinski2019magic}%
  \BibitemOpen
  \bibfield  {author} {\bibinfo {author} {\bibfnamefont {D.}~\bibnamefont {Litinski}},\ }\bibfield  {title} {\bibinfo {title} {Magic state distillation: Not as costly as you think},\ }\href@noop {} {\bibfield  {journal} {\bibinfo  {journal} {Quantum}\ }\textbf {\bibinfo {volume} {3}},\ \bibinfo {pages} {205} (\bibinfo {year} {2019})}\BibitemShut {NoStop}%
\bibitem [{\citenamefont {Bravyi}\ \emph {et~al.}(2016)\citenamefont {Bravyi}, \citenamefont {Smith},\ and\ \citenamefont {Smolin}}]{bravyi2016trading}%
  \BibitemOpen
  \bibfield  {author} {\bibinfo {author} {\bibfnamefont {S.}~\bibnamefont {Bravyi}}, \bibinfo {author} {\bibfnamefont {G.}~\bibnamefont {Smith}},\ and\ \bibinfo {author} {\bibfnamefont {J.~A.}\ \bibnamefont {Smolin}},\ }\bibfield  {title} {\bibinfo {title} {Trading classical and quantum computational resources},\ }\href {https://doi.org/10.1103/PhysRevX.6.021043} {\bibfield  {journal} {\bibinfo  {journal} {Phys. Rev. X}\ }\textbf {\bibinfo {volume} {6}},\ \bibinfo {pages} {021043} (\bibinfo {year} {2016})}\BibitemShut {NoStop}%
\bibitem [{\citenamefont {Leone}\ \emph {et~al.}(2022)\citenamefont {Leone}, \citenamefont {Oliviero},\ and\ \citenamefont {Hamma}}]{leone2022stabilizer}%
  \BibitemOpen
  \bibfield  {author} {\bibinfo {author} {\bibfnamefont {L.}~\bibnamefont {Leone}}, \bibinfo {author} {\bibfnamefont {S.~F.~E.}\ \bibnamefont {Oliviero}},\ and\ \bibinfo {author} {\bibfnamefont {A.}~\bibnamefont {Hamma}},\ }\bibfield  {title} {\bibinfo {title} {Stabilizer r\'enyi entropy},\ }\href {https://doi.org/10.1103/PhysRevLett.128.050402} {\bibfield  {journal} {\bibinfo  {journal} {Phys. Rev. Lett.}\ }\textbf {\bibinfo {volume} {128}},\ \bibinfo {pages} {050402} (\bibinfo {year} {2022})}\BibitemShut {NoStop}%
\bibitem [{\citenamefont {Haug}\ and\ \citenamefont {Kim}(2023)}]{haug2022scalable}%
  \BibitemOpen
  \bibfield  {author} {\bibinfo {author} {\bibfnamefont {T.}~\bibnamefont {Haug}}\ and\ \bibinfo {author} {\bibfnamefont {M.}~\bibnamefont {Kim}},\ }\bibfield  {title} {\bibinfo {title} {Scalable measures of magic resource for quantum computers},\ }\href {https://doi.org/10.1103/PRXQuantum.4.010301} {\bibfield  {journal} {\bibinfo  {journal} {PRX Quantum}\ }\textbf {\bibinfo {volume} {4}},\ \bibinfo {pages} {010301} (\bibinfo {year} {2023})}\BibitemShut {NoStop}%
\bibitem [{\citenamefont {Haug}\ \emph {et~al.}(2024{\natexlab{a}})\citenamefont {Haug}, \citenamefont {Lee},\ and\ \citenamefont {Kim}}]{haug2023efficient}%
  \BibitemOpen
  \bibfield  {author} {\bibinfo {author} {\bibfnamefont {T.}~\bibnamefont {Haug}}, \bibinfo {author} {\bibfnamefont {S.}~\bibnamefont {Lee}},\ and\ \bibinfo {author} {\bibfnamefont {M.~S.}\ \bibnamefont {Kim}},\ }\bibfield  {title} {\bibinfo {title} {Efficient quantum algorithms for stabilizer entropies},\ }\href {https://doi.org/10.1103/PhysRevLett.132.240602} {\bibfield  {journal} {\bibinfo  {journal} {Phys. Rev. Lett.}\ }\textbf {\bibinfo {volume} {132}},\ \bibinfo {pages} {240602} (\bibinfo {year} {2024}{\natexlab{a}})}\BibitemShut {NoStop}%
\bibitem [{\citenamefont {Haug}\ and\ \citenamefont {Piroli}(2023{\natexlab{a}})}]{haug2022quantifying}%
  \BibitemOpen
  \bibfield  {author} {\bibinfo {author} {\bibfnamefont {T.}~\bibnamefont {Haug}}\ and\ \bibinfo {author} {\bibfnamefont {L.}~\bibnamefont {Piroli}},\ }\bibfield  {title} {\bibinfo {title} {Quantifying nonstabilizerness of matrix product states},\ }\href {https://doi.org/10.1103/PhysRevB.107.035148} {\bibfield  {journal} {\bibinfo  {journal} {Phys. Rev. B}\ }\textbf {\bibinfo {volume} {107}},\ \bibinfo {pages} {035148} (\bibinfo {year} {2023}{\natexlab{a}})}\BibitemShut {NoStop}%
\bibitem [{\citenamefont {Haug}\ and\ \citenamefont {Piroli}(2023{\natexlab{b}})}]{haug2023stabilizer}%
  \BibitemOpen
  \bibfield  {author} {\bibinfo {author} {\bibfnamefont {T.}~\bibnamefont {Haug}}\ and\ \bibinfo {author} {\bibfnamefont {L.}~\bibnamefont {Piroli}},\ }\bibfield  {title} {\bibinfo {title} {Stabilizer entropies and nonstabilizerness monotones},\ }\href@noop {} {\bibfield  {journal} {\bibinfo  {journal} {Quantum}\ }\textbf {\bibinfo {volume} {7}},\ \bibinfo {pages} {1092} (\bibinfo {year} {2023}{\natexlab{b}})}\BibitemShut {NoStop}%
\bibitem [{\citenamefont {Tarabunga}\ \emph {et~al.}(2024{\natexlab{a}})\citenamefont {Tarabunga}, \citenamefont {Frau}, \citenamefont {Haug}, \citenamefont {Tirrito},\ and\ \citenamefont {Piroli}}]{tarabunga2024nonstabilizerness}%
  \BibitemOpen
  \bibfield  {author} {\bibinfo {author} {\bibfnamefont {P.~S.}\ \bibnamefont {Tarabunga}}, \bibinfo {author} {\bibfnamefont {M.}~\bibnamefont {Frau}}, \bibinfo {author} {\bibfnamefont {T.}~\bibnamefont {Haug}}, \bibinfo {author} {\bibfnamefont {E.}~\bibnamefont {Tirrito}},\ and\ \bibinfo {author} {\bibfnamefont {L.}~\bibnamefont {Piroli}},\ }\bibfield  {title} {\bibinfo {title} {A nonstabilizerness monotone from stabilizerness asymmetry},\ }\href@noop {} {\bibfield  {journal} {\bibinfo  {journal} {arXiv:2411.05766}\ } (\bibinfo {year} {2024}{\natexlab{a}})}\BibitemShut {NoStop}%
\bibitem [{\citenamefont {Lami}\ and\ \citenamefont {Collura}(2023)}]{lami2023quantum}%
  \BibitemOpen
  \bibfield  {author} {\bibinfo {author} {\bibfnamefont {G.}~\bibnamefont {Lami}}\ and\ \bibinfo {author} {\bibfnamefont {M.}~\bibnamefont {Collura}},\ }\bibfield  {title} {\bibinfo {title} {Nonstabilizerness via perfect pauli sampling of matrix product states},\ }\href {https://doi.org/10.1103/PhysRevLett.131.180401} {\bibfield  {journal} {\bibinfo  {journal} {Phys. Rev. Lett.}\ }\textbf {\bibinfo {volume} {131}},\ \bibinfo {pages} {180401} (\bibinfo {year} {2023})}\BibitemShut {NoStop}%
\bibitem [{\citenamefont {Gross}\ \emph {et~al.}(2021)\citenamefont {Gross}, \citenamefont {Nezami},\ and\ \citenamefont {Walter}}]{gross2021schur}%
  \BibitemOpen
  \bibfield  {author} {\bibinfo {author} {\bibfnamefont {D.}~\bibnamefont {Gross}}, \bibinfo {author} {\bibfnamefont {S.}~\bibnamefont {Nezami}},\ and\ \bibinfo {author} {\bibfnamefont {M.}~\bibnamefont {Walter}},\ }\bibfield  {title} {\bibinfo {title} {Schur--weyl duality for the clifford group with applications: Property testing, a robust hudson theorem, and de finetti representations},\ }\href@noop {} {\bibfield  {journal} {\bibinfo  {journal} {Communications in Mathematical Physics}\ }\textbf {\bibinfo {volume} {385}},\ \bibinfo {pages} {1325} (\bibinfo {year} {2021})}\BibitemShut {NoStop}%
\bibitem [{\citenamefont {Grewal}\ \emph {et~al.}(2023)\citenamefont {Grewal}, \citenamefont {Iyer}, \citenamefont {Kretschmer},\ and\ \citenamefont {Liang}}]{grewal2023improved}%
  \BibitemOpen
  \bibfield  {author} {\bibinfo {author} {\bibfnamefont {S.}~\bibnamefont {Grewal}}, \bibinfo {author} {\bibfnamefont {V.}~\bibnamefont {Iyer}}, \bibinfo {author} {\bibfnamefont {W.}~\bibnamefont {Kretschmer}},\ and\ \bibinfo {author} {\bibfnamefont {D.}~\bibnamefont {Liang}},\ }\bibfield  {title} {\bibinfo {title} {Improved stabilizer estimation via bell difference sampling},\ }\href@noop {} {\bibfield  {journal} {\bibinfo  {journal} {arXiv:2304.13915}\ } (\bibinfo {year} {2023})}\BibitemShut {NoStop}%
\bibitem [{\citenamefont {Bao}\ \emph {et~al.}(2024)\citenamefont {Bao}, \citenamefont {van Dordrecht},\ and\ \citenamefont {Helsen}}]{bao2024tolerant}%
  \BibitemOpen
  \bibfield  {author} {\bibinfo {author} {\bibfnamefont {Z.}~\bibnamefont {Bao}}, \bibinfo {author} {\bibfnamefont {P.}~\bibnamefont {van Dordrecht}},\ and\ \bibinfo {author} {\bibfnamefont {J.}~\bibnamefont {Helsen}},\ }\bibfield  {title} {\bibinfo {title} {Tolerant testing of stabilizer states with a polynomial gap via a generalized uncertainty relation},\ }\href@noop {} {\bibfield  {journal} {\bibinfo  {journal} {arXiv preprint arXiv:2410.21811}\ } (\bibinfo {year} {2024})}\BibitemShut {NoStop}%
\bibitem [{\citenamefont {Arunachalam}\ \emph {et~al.}(2024)\citenamefont {Arunachalam}, \citenamefont {Bravyi},\ and\ \citenamefont {Dutt}}]{arunachalam2024note}%
  \BibitemOpen
  \bibfield  {author} {\bibinfo {author} {\bibfnamefont {S.}~\bibnamefont {Arunachalam}}, \bibinfo {author} {\bibfnamefont {S.}~\bibnamefont {Bravyi}},\ and\ \bibinfo {author} {\bibfnamefont {A.}~\bibnamefont {Dutt}},\ }\bibfield  {title} {\bibinfo {title} {A note on polynomial-time tolerant testing stabilizer states},\ }\href@noop {} {\bibfield  {journal} {\bibinfo  {journal} {arXiv preprint arXiv:2410.22220}\ } (\bibinfo {year} {2024})}\BibitemShut {NoStop}%
\bibitem [{\citenamefont {Iyer}\ and\ \citenamefont {Liang}(2024)}]{iyer2024tolerant}%
  \BibitemOpen
  \bibfield  {author} {\bibinfo {author} {\bibfnamefont {V.}~\bibnamefont {Iyer}}\ and\ \bibinfo {author} {\bibfnamefont {D.}~\bibnamefont {Liang}},\ }\bibfield  {title} {\bibinfo {title} {Tolerant testing of stabilizer states with mixed state inputs},\ }\href@noop {} {\bibfield  {journal} {\bibinfo  {journal} {arXiv preprint arXiv:2411.08765}\ } (\bibinfo {year} {2024})}\BibitemShut {NoStop}%
\bibitem [{\citenamefont {Hinsche}\ and\ \citenamefont {Helsen}(2024)}]{hinsche2024single}%
  \BibitemOpen
  \bibfield  {author} {\bibinfo {author} {\bibfnamefont {M.}~\bibnamefont {Hinsche}}\ and\ \bibinfo {author} {\bibfnamefont {J.}~\bibnamefont {Helsen}},\ }\bibfield  {title} {\bibinfo {title} {Single-copy stabilizer testing},\ }\href@noop {} {\bibfield  {journal} {\bibinfo  {journal} {arXiv preprint arXiv:2410.07986}\ } (\bibinfo {year} {2024})}\BibitemShut {NoStop}%
\bibitem [{\citenamefont {Bluvstein}\ \emph {et~al.}(2024)\citenamefont {Bluvstein}, \citenamefont {Evered}, \citenamefont {Geim}, \citenamefont {Li}, \citenamefont {Zhou}, \citenamefont {Manovitz}, \citenamefont {Ebadi}, \citenamefont {Cain}, \citenamefont {Kalinowski}, \citenamefont {Hangleiter} \emph {et~al.}}]{bluvstein2024logical}%
  \BibitemOpen
  \bibfield  {author} {\bibinfo {author} {\bibfnamefont {D.}~\bibnamefont {Bluvstein}}, \bibinfo {author} {\bibfnamefont {S.~J.}\ \bibnamefont {Evered}}, \bibinfo {author} {\bibfnamefont {A.~A.}\ \bibnamefont {Geim}}, \bibinfo {author} {\bibfnamefont {S.~H.}\ \bibnamefont {Li}}, \bibinfo {author} {\bibfnamefont {H.}~\bibnamefont {Zhou}}, \bibinfo {author} {\bibfnamefont {T.}~\bibnamefont {Manovitz}}, \bibinfo {author} {\bibfnamefont {S.}~\bibnamefont {Ebadi}}, \bibinfo {author} {\bibfnamefont {M.}~\bibnamefont {Cain}}, \bibinfo {author} {\bibfnamefont {M.}~\bibnamefont {Kalinowski}}, \bibinfo {author} {\bibfnamefont {D.}~\bibnamefont {Hangleiter}}, \emph {et~al.},\ }\bibfield  {title} {\bibinfo {title} {Logical quantum processor based on reconfigurable atom arrays},\ }\href {https://doi.org/10.1038/s41586-023-06927-3} {\bibfield  {journal} {\bibinfo  {journal} {Nature}\ }\textbf {\bibinfo {volume} {626}},\ \bibinfo {pages} {58} (\bibinfo {year} {2024})}\BibitemShut {NoStop}%
\bibitem [{\citenamefont {Oliviero}\ \emph {et~al.}(2022{\natexlab{a}})\citenamefont {Oliviero}, \citenamefont {Leone}, \citenamefont {Hamma},\ and\ \citenamefont {Lloyd}}]{oliviero2022measuring}%
  \BibitemOpen
  \bibfield  {author} {\bibinfo {author} {\bibfnamefont {S.~F.~E.}\ \bibnamefont {Oliviero}}, \bibinfo {author} {\bibfnamefont {L.}~\bibnamefont {Leone}}, \bibinfo {author} {\bibfnamefont {A.}~\bibnamefont {Hamma}},\ and\ \bibinfo {author} {\bibfnamefont {S.}~\bibnamefont {Lloyd}},\ }\bibfield  {title} {\bibinfo {title} {Measuring magic on a quantum processor},\ }\href {https://doi.org/10.1038/s41534-022-00666-5} {\bibfield  {journal} {\bibinfo  {journal} {npj Quantum Information}\ }\textbf {\bibinfo {volume} {8}},\ \bibinfo {pages} {148} (\bibinfo {year} {2022}{\natexlab{a}})}\BibitemShut {NoStop}%
\bibitem [{\citenamefont {Niroula}\ \emph {et~al.}(2023)\citenamefont {Niroula}, \citenamefont {White}, \citenamefont {Wang}, \citenamefont {Johri}, \citenamefont {Zhu}, \citenamefont {Monroe}, \citenamefont {Noel},\ and\ \citenamefont {Gullans}}]{niroula2023phase}%
  \BibitemOpen
  \bibfield  {author} {\bibinfo {author} {\bibfnamefont {P.}~\bibnamefont {Niroula}}, \bibinfo {author} {\bibfnamefont {C.~D.}\ \bibnamefont {White}}, \bibinfo {author} {\bibfnamefont {Q.}~\bibnamefont {Wang}}, \bibinfo {author} {\bibfnamefont {S.}~\bibnamefont {Johri}}, \bibinfo {author} {\bibfnamefont {D.}~\bibnamefont {Zhu}}, \bibinfo {author} {\bibfnamefont {C.}~\bibnamefont {Monroe}}, \bibinfo {author} {\bibfnamefont {C.}~\bibnamefont {Noel}},\ and\ \bibinfo {author} {\bibfnamefont {M.~J.}\ \bibnamefont {Gullans}},\ }\bibfield  {title} {\bibinfo {title} {Phase transition in magic with random quantum circuits},\ }\href@noop {} {\bibfield  {journal} {\bibinfo  {journal} {arXiv:2304.10481}\ } (\bibinfo {year} {2023})}\BibitemShut {NoStop}%
\bibitem [{\citenamefont {Dai}\ \emph {et~al.}(2022)\citenamefont {Dai}, \citenamefont {Fu},\ and\ \citenamefont {Luo}}]{dai2022detecting}%
  \BibitemOpen
  \bibfield  {author} {\bibinfo {author} {\bibfnamefont {H.}~\bibnamefont {Dai}}, \bibinfo {author} {\bibfnamefont {S.}~\bibnamefont {Fu}},\ and\ \bibinfo {author} {\bibfnamefont {S.}~\bibnamefont {Luo}},\ }\bibfield  {title} {\bibinfo {title} {Detecting magic states via characteristic functions},\ }\href@noop {} {\bibfield  {journal} {\bibinfo  {journal} {International Journal of Theoretical Physics}\ }\textbf {\bibinfo {volume} {61}},\ \bibinfo {pages} {35} (\bibinfo {year} {2022})}\BibitemShut {NoStop}%
\bibitem [{\citenamefont {Warmuz}\ \emph {et~al.}(2024)\citenamefont {Warmuz}, \citenamefont {Dokudowiec}, \citenamefont {Radhakrishnan},\ and\ \citenamefont {Byrnes}}]{warmuz2024magic}%
  \BibitemOpen
  \bibfield  {author} {\bibinfo {author} {\bibfnamefont {K.}~\bibnamefont {Warmuz}}, \bibinfo {author} {\bibfnamefont {E.}~\bibnamefont {Dokudowiec}}, \bibinfo {author} {\bibfnamefont {C.}~\bibnamefont {Radhakrishnan}},\ and\ \bibinfo {author} {\bibfnamefont {T.}~\bibnamefont {Byrnes}},\ }\bibfield  {title} {\bibinfo {title} {A magic monotone for faithful detection of non-stabilizerness in mixed states},\ }\href@noop {} {\bibfield  {journal} {\bibinfo  {journal} {arXiv preprint arXiv:2409.18570}\ } (\bibinfo {year} {2024})}\BibitemShut {NoStop}%
\bibitem [{\citenamefont {Macedo}\ \emph {et~al.}(2025{\natexlab{a}})\citenamefont {Macedo}, \citenamefont {Andriolo}, \citenamefont {Zamora}, \citenamefont {Poderini},\ and\ \citenamefont {Chaves}}]{macedo2025witnessingmagicbellinequalities}%
  \BibitemOpen
  \bibfield  {author} {\bibinfo {author} {\bibfnamefont {R.~A.}\ \bibnamefont {Macedo}}, \bibinfo {author} {\bibfnamefont {P.}~\bibnamefont {Andriolo}}, \bibinfo {author} {\bibfnamefont {S.}~\bibnamefont {Zamora}}, \bibinfo {author} {\bibfnamefont {D.}~\bibnamefont {Poderini}},\ and\ \bibinfo {author} {\bibfnamefont {R.}~\bibnamefont {Chaves}},\ }\bibfield  {title} {\bibinfo {title} {Witnessing magic with bell inequalities},\ }\href {https://arxiv.org/abs/2503.18734} {\bibfield  {journal} {\bibinfo  {journal} {arXiv:2503.18734}\ } (\bibinfo {year} {2025}{\natexlab{a}})}\BibitemShut {NoStop}%
\bibitem [{\citenamefont {Macedo}\ \emph {et~al.}(2025{\natexlab{b}})\citenamefont {Macedo}, \citenamefont {Andriolo}, \citenamefont {Zamora}, \citenamefont {Poderini},\ and\ \citenamefont {Chaves}}]{macedo2025witnessing}%
  \BibitemOpen
  \bibfield  {author} {\bibinfo {author} {\bibfnamefont {R.~A.}\ \bibnamefont {Macedo}}, \bibinfo {author} {\bibfnamefont {P.}~\bibnamefont {Andriolo}}, \bibinfo {author} {\bibfnamefont {S.}~\bibnamefont {Zamora}}, \bibinfo {author} {\bibfnamefont {D.}~\bibnamefont {Poderini}},\ and\ \bibinfo {author} {\bibfnamefont {R.}~\bibnamefont {Chaves}},\ }\bibfield  {title} {\bibinfo {title} {Witnessing magic with bell inequalities},\ }\href@noop {} {\bibfield  {journal} {\bibinfo  {journal} {arXiv preprint arXiv:2503.18734}\ } (\bibinfo {year} {2025}{\natexlab{b}})}\BibitemShut {NoStop}%
\bibitem [{\citenamefont {Rubinfeld}\ and\ \citenamefont {Sudan}(1996)}]{rubinfeld1996robust}%
  \BibitemOpen
  \bibfield  {author} {\bibinfo {author} {\bibfnamefont {R.}~\bibnamefont {Rubinfeld}}\ and\ \bibinfo {author} {\bibfnamefont {M.}~\bibnamefont {Sudan}},\ }\bibfield  {title} {\bibinfo {title} {Robust characterizations of polynomials with applications to program testing},\ }\href {https://doi.org/10.1137/S0097539793255151} {\bibfield  {journal} {\bibinfo  {journal} {SIAM Journal on Computing}\ }\textbf {\bibinfo {volume} {25}},\ \bibinfo {pages} {252} (\bibinfo {year} {1996})}\BibitemShut {NoStop}%
\bibitem [{\citenamefont {Goldreich}\ \emph {et~al.}(1998)\citenamefont {Goldreich}, \citenamefont {Goldwasser},\ and\ \citenamefont {Ron}}]{goldreich1998property}%
  \BibitemOpen
  \bibfield  {author} {\bibinfo {author} {\bibfnamefont {O.}~\bibnamefont {Goldreich}}, \bibinfo {author} {\bibfnamefont {S.}~\bibnamefont {Goldwasser}},\ and\ \bibinfo {author} {\bibfnamefont {D.}~\bibnamefont {Ron}},\ }\bibfield  {title} {\bibinfo {title} {Property testing and its connection to learning and approximation},\ }\href {https://doi.org/10.1145/285055.285060} {\bibfield  {journal} {\bibinfo  {journal} {Journal of the ACM (JACM)}\ }\textbf {\bibinfo {volume} {45}},\ \bibinfo {pages} {653} (\bibinfo {year} {1998})}\BibitemShut {NoStop}%
\bibitem [{\citenamefont {Buhrman}\ \emph {et~al.}(2008)\citenamefont {Buhrman}, \citenamefont {Fortnow}, \citenamefont {Newman},\ and\ \citenamefont {R{\"o}hrig}}]{buhrman2008quantum}%
  \BibitemOpen
  \bibfield  {author} {\bibinfo {author} {\bibfnamefont {H.}~\bibnamefont {Buhrman}}, \bibinfo {author} {\bibfnamefont {L.}~\bibnamefont {Fortnow}}, \bibinfo {author} {\bibfnamefont {I.}~\bibnamefont {Newman}},\ and\ \bibinfo {author} {\bibfnamefont {H.}~\bibnamefont {R{\"o}hrig}},\ }\bibfield  {title} {\bibinfo {title} {Quantum property testing},\ }\href {https://doi.org/10.1137/S009753970444241} {\bibfield  {journal} {\bibinfo  {journal} {SIAM Journal on Computing}\ }\textbf {\bibinfo {volume} {37}},\ \bibinfo {pages} {1387} (\bibinfo {year} {2008})}\BibitemShut {NoStop}%
\bibitem [{\citenamefont {Montanaro}\ and\ \citenamefont {de~Wolf}(2013)}]{montanaro2013survey}%
  \BibitemOpen
  \bibfield  {author} {\bibinfo {author} {\bibfnamefont {A.}~\bibnamefont {Montanaro}}\ and\ \bibinfo {author} {\bibfnamefont {R.}~\bibnamefont {de~Wolf}},\ }\bibfield  {title} {\bibinfo {title} {A survey of quantum property testing},\ }\href@noop {} {\bibfield  {journal} {\bibinfo  {journal} {arXiv preprint arXiv:1310.2035}\ } (\bibinfo {year} {2013})}\BibitemShut {NoStop}%
\bibitem [{\citenamefont {Bansal}\ \emph {et~al.}(2024)\citenamefont {Bansal}, \citenamefont {Mok}, \citenamefont {Bharti}, \citenamefont {Koh},\ and\ \citenamefont {Haug}}]{bansal2024pseudorandomdensitymatrices}%
  \BibitemOpen
  \bibfield  {author} {\bibinfo {author} {\bibfnamefont {N.}~\bibnamefont {Bansal}}, \bibinfo {author} {\bibfnamefont {W.-K.}\ \bibnamefont {Mok}}, \bibinfo {author} {\bibfnamefont {K.}~\bibnamefont {Bharti}}, \bibinfo {author} {\bibfnamefont {D.~E.}\ \bibnamefont {Koh}},\ and\ \bibinfo {author} {\bibfnamefont {T.}~\bibnamefont {Haug}},\ }\bibfield  {title} {\bibinfo {title} {Pseudorandom density matrices},\ }\href {https://doi.org/10.48550/arXiv.2407.11607} {\bibfield  {journal} {\bibinfo  {journal} {arXiv preprint arXiv:2407.11607}\ } (\bibinfo {year} {2024})}\BibitemShut {NoStop}%
\bibitem [{\citenamefont {Sarkar}\ \emph {et~al.}(2020)\citenamefont {Sarkar}, \citenamefont {Mukhopadhyay},\ and\ \citenamefont {Bayat}}]{sarkar2020characterization}%
  \BibitemOpen
  \bibfield  {author} {\bibinfo {author} {\bibfnamefont {S.}~\bibnamefont {Sarkar}}, \bibinfo {author} {\bibfnamefont {C.}~\bibnamefont {Mukhopadhyay}},\ and\ \bibinfo {author} {\bibfnamefont {A.}~\bibnamefont {Bayat}},\ }\bibfield  {title} {\bibinfo {title} {Characterization of an operational quantum resource in a critical many-body system},\ }\href {https://doi.org/10.1088/1367-2630/aba919} {\bibfield  {journal} {\bibinfo  {journal} {New J. Phys.}\ }\textbf {\bibinfo {volume} {22}},\ \bibinfo {pages} {083077} (\bibinfo {year} {2020})}\BibitemShut {NoStop}%
\bibitem [{\citenamefont {Liu}\ and\ \citenamefont {Winter}(2022)}]{liu2022many}%
  \BibitemOpen
  \bibfield  {author} {\bibinfo {author} {\bibfnamefont {Z.-W.}\ \bibnamefont {Liu}}\ and\ \bibinfo {author} {\bibfnamefont {A.}~\bibnamefont {Winter}},\ }\bibfield  {title} {\bibinfo {title} {Many-body quantum magic},\ }\href {https://doi.org/10.1103/PRXQuantum.3.020333} {\bibfield  {journal} {\bibinfo  {journal} {PRX Quantum}\ }\textbf {\bibinfo {volume} {3}},\ \bibinfo {pages} {020333} (\bibinfo {year} {2022})}\BibitemShut {NoStop}%
\bibitem [{\citenamefont {Oliviero}\ \emph {et~al.}(2022{\natexlab{b}})\citenamefont {Oliviero}, \citenamefont {Leone},\ and\ \citenamefont {Hamma}}]{oliviero2022magic}%
  \BibitemOpen
  \bibfield  {author} {\bibinfo {author} {\bibfnamefont {S.~F.~E.}\ \bibnamefont {Oliviero}}, \bibinfo {author} {\bibfnamefont {L.}~\bibnamefont {Leone}},\ and\ \bibinfo {author} {\bibfnamefont {A.}~\bibnamefont {Hamma}},\ }\bibfield  {title} {\bibinfo {title} {Magic-state resource theory for the ground state of the transverse-field ising model},\ }\href {https://doi.org/10.1103/PhysRevA.106.042426} {\bibfield  {journal} {\bibinfo  {journal} {Phys. Rev. A}\ }\textbf {\bibinfo {volume} {106}},\ \bibinfo {pages} {042426} (\bibinfo {year} {2022}{\natexlab{b}})}\BibitemShut {NoStop}%
\bibitem [{\citenamefont {Tarabunga}\ \emph {et~al.}(2023)\citenamefont {Tarabunga}, \citenamefont {Tirrito}, \citenamefont {Chanda},\ and\ \citenamefont {Dalmonte}}]{tarabunga2023many}%
  \BibitemOpen
  \bibfield  {author} {\bibinfo {author} {\bibfnamefont {P.~S.}\ \bibnamefont {Tarabunga}}, \bibinfo {author} {\bibfnamefont {E.}~\bibnamefont {Tirrito}}, \bibinfo {author} {\bibfnamefont {T.}~\bibnamefont {Chanda}},\ and\ \bibinfo {author} {\bibfnamefont {M.}~\bibnamefont {Dalmonte}},\ }\bibfield  {title} {\bibinfo {title} {Many-body magic via pauli-markov chains---from criticality to gauge theories},\ }\href {https://doi.org/10.1103/PRXQuantum.4.040317} {\bibfield  {journal} {\bibinfo  {journal} {PRX Quantum}\ }\textbf {\bibinfo {volume} {4}},\ \bibinfo {pages} {040317} (\bibinfo {year} {2023})}\BibitemShut {NoStop}%
\bibitem [{\citenamefont {Tarabunga}\ and\ \citenamefont {Haug}(2025)}]{tarabunga2025efficientmutualmagicmagic}%
  \BibitemOpen
  \bibfield  {author} {\bibinfo {author} {\bibfnamefont {P.~S.}\ \bibnamefont {Tarabunga}}\ and\ \bibinfo {author} {\bibfnamefont {T.}~\bibnamefont {Haug}},\ }\bibfield  {title} {\bibinfo {title} {Efficient mutual magic and magic capacity with matrix product states},\ }\href {https://arxiv.org/abs/2504.07230} {\bibfield  {journal} {\bibinfo  {journal} {arXiv:2504.07230}\ } (\bibinfo {year} {2025})}\BibitemShut {NoStop}%
\bibitem [{\citenamefont {Korbany}\ \emph {et~al.}(2025)\citenamefont {Korbany}, \citenamefont {Gullans},\ and\ \citenamefont {Piroli}}]{korbany2025longrangenonstabilizerness}%
  \BibitemOpen
  \bibfield  {author} {\bibinfo {author} {\bibfnamefont {D.~A.}\ \bibnamefont {Korbany}}, \bibinfo {author} {\bibfnamefont {M.~J.}\ \bibnamefont {Gullans}},\ and\ \bibinfo {author} {\bibfnamefont {L.}~\bibnamefont {Piroli}},\ }\bibfield  {title} {\bibinfo {title} {Long-range nonstabilizerness and phases of matter},\ }\href {https://arxiv.org/abs/2502.19504} {\bibfield  {journal} {\bibinfo  {journal} {arxiv:2502.19504}\ } (\bibinfo {year} {2025})}\BibitemShut {NoStop}%
\bibitem [{\citenamefont {Ji}\ \emph {et~al.}(2018)\citenamefont {Ji}, \citenamefont {Liu},\ and\ \citenamefont {Song}}]{ji2018pseudorandom}%
  \BibitemOpen
  \bibfield  {author} {\bibinfo {author} {\bibfnamefont {Z.}~\bibnamefont {Ji}}, \bibinfo {author} {\bibfnamefont {Y.-K.}\ \bibnamefont {Liu}},\ and\ \bibinfo {author} {\bibfnamefont {F.}~\bibnamefont {Song}},\ }\bibfield  {title} {\bibinfo {title} {Pseudorandom quantum states},\ }in\ \href {https://doi.org/10.1007/978-3-319-96878-0_5} {\emph {\bibinfo {booktitle} {Annual International Cryptology Conference}}}\ (\bibinfo {organization} {Springer},\ \bibinfo {year} {2018})\ pp.\ \bibinfo {pages} {126--152}\BibitemShut {NoStop}%
\bibitem [{\citenamefont {Aaronson}\ \emph {et~al.}(2022)\citenamefont {Aaronson}, \citenamefont {Bouland}, \citenamefont {Fefferman}, \citenamefont {Ghosh}, \citenamefont {Vazirani}, \citenamefont {Zhang},\ and\ \citenamefont {Zhou}}]{bouland2022quantum}%
  \BibitemOpen
  \bibfield  {author} {\bibinfo {author} {\bibfnamefont {S.}~\bibnamefont {Aaronson}}, \bibinfo {author} {\bibfnamefont {A.}~\bibnamefont {Bouland}}, \bibinfo {author} {\bibfnamefont {B.}~\bibnamefont {Fefferman}}, \bibinfo {author} {\bibfnamefont {S.}~\bibnamefont {Ghosh}}, \bibinfo {author} {\bibfnamefont {U.}~\bibnamefont {Vazirani}}, \bibinfo {author} {\bibfnamefont {C.}~\bibnamefont {Zhang}},\ and\ \bibinfo {author} {\bibfnamefont {Z.}~\bibnamefont {Zhou}},\ }\bibfield  {title} {\bibinfo {title} {Quantum pseudoentanglement},\ }\href {https://doi.org/10.48550/arXiv.2211.00747} {\bibfield  {journal} {\bibinfo  {journal} {arXiv preprint arXiv:2211.00747}\ } (\bibinfo {year} {2022})}\BibitemShut {NoStop}%
\bibitem [{\citenamefont {Haug}\ \emph {et~al.}(2023)\citenamefont {Haug}, \citenamefont {Bharti},\ and\ \citenamefont {Koh}}]{haug2023pseudorandom}%
  \BibitemOpen
  \bibfield  {author} {\bibinfo {author} {\bibfnamefont {T.}~\bibnamefont {Haug}}, \bibinfo {author} {\bibfnamefont {K.}~\bibnamefont {Bharti}},\ and\ \bibinfo {author} {\bibfnamefont {D.~E.}\ \bibnamefont {Koh}},\ }\bibfield  {title} {\bibinfo {title} {Pseudorandom unitaries are neither real nor sparse nor noise-robust},\ }\href@noop {} {\bibfield  {journal} {\bibinfo  {journal} {arXiv:2306.11677}\ } (\bibinfo {year} {2023})}\BibitemShut {NoStop}%
\bibitem [{\citenamefont {Gu}\ \emph {et~al.}(2024)\citenamefont {Gu}, \citenamefont {Leone}, \citenamefont {Ghosh}, \citenamefont {Eisert}, \citenamefont {Yelin},\ and\ \citenamefont {Quek}}]{gu2023little}%
  \BibitemOpen
  \bibfield  {author} {\bibinfo {author} {\bibfnamefont {A.}~\bibnamefont {Gu}}, \bibinfo {author} {\bibfnamefont {L.}~\bibnamefont {Leone}}, \bibinfo {author} {\bibfnamefont {S.}~\bibnamefont {Ghosh}}, \bibinfo {author} {\bibfnamefont {J.}~\bibnamefont {Eisert}}, \bibinfo {author} {\bibfnamefont {S.~F.}\ \bibnamefont {Yelin}},\ and\ \bibinfo {author} {\bibfnamefont {Y.}~\bibnamefont {Quek}},\ }\bibfield  {title} {\bibinfo {title} {Pseudomagic quantum states},\ }\href@noop {} {\bibfield  {journal} {\bibinfo  {journal} {Physical Review Letters}\ }\textbf {\bibinfo {volume} {132}},\ \bibinfo {pages} {210602} (\bibinfo {year} {2024})}\BibitemShut {NoStop}%
\bibitem [{\citenamefont {Tanggara}\ \emph {et~al.}(2025)\citenamefont {Tanggara}, \citenamefont {Gu},\ and\ \citenamefont {Bharti}}]{tanggara2025neartermpseudorandompseudoresourcequantum}%
  \BibitemOpen
  \bibfield  {author} {\bibinfo {author} {\bibfnamefont {A.}~\bibnamefont {Tanggara}}, \bibinfo {author} {\bibfnamefont {M.}~\bibnamefont {Gu}},\ and\ \bibinfo {author} {\bibfnamefont {K.}~\bibnamefont {Bharti}},\ }\bibfield  {title} {\bibinfo {title} {Near-term pseudorandom and pseudoresource quantum states},\ }\href {https://arxiv.org/abs/2504.17650} {\bibfield  {journal} {\bibinfo  {journal} {arXiv:2504.17650}\ } (\bibinfo {year} {2025})}\BibitemShut {NoStop}%
\bibitem [{\citenamefont {Grilo}\ and\ \citenamefont {Álvaro Yángüez}(2025)}]{grilo2025quantumpseudoresourcesimplycryptography}%
  \BibitemOpen
  \bibfield  {author} {\bibinfo {author} {\bibfnamefont {A.~B.}\ \bibnamefont {Grilo}}\ and\ \bibinfo {author} {\bibnamefont {Álvaro Yángüez}},\ }\bibfield  {title} {\bibinfo {title} {Quantum pseudoresources imply cryptography},\ }\href {https://arxiv.org/abs/2504.15025} {\bibfield  {journal} {\bibinfo  {journal} {arXiv:2504.15025}\ } (\bibinfo {year} {2025})}\BibitemShut {NoStop}%
\bibitem [{\citenamefont {Haug}\ \emph {et~al.}(2025)\citenamefont {Haug}, \citenamefont {Bansal}, \citenamefont {Mok}, \citenamefont {Koh},\ and\ \citenamefont {Bharti}}]{haug2025pseudorandom}%
  \BibitemOpen
  \bibfield  {author} {\bibinfo {author} {\bibfnamefont {T.}~\bibnamefont {Haug}}, \bibinfo {author} {\bibfnamefont {N.}~\bibnamefont {Bansal}}, \bibinfo {author} {\bibfnamefont {W.-K.}\ \bibnamefont {Mok}}, \bibinfo {author} {\bibfnamefont {D.~E.}\ \bibnamefont {Koh}},\ and\ \bibinfo {author} {\bibfnamefont {K.}~\bibnamefont {Bharti}},\ }\bibfield  {title} {\bibinfo {title} {Pseudorandom quantum authentication},\ }\href@noop {} {\bibfield  {journal} {\bibinfo  {journal} {arXiv preprint arXiv:2501.00951}\ } (\bibinfo {year} {2025})}\BibitemShut {NoStop}%
\bibitem [{\citenamefont {Eisert}\ \emph {et~al.}(2007)\citenamefont {Eisert}, \citenamefont {Brandao},\ and\ \citenamefont {Audenaert}}]{eisert2007quantitative}%
  \BibitemOpen
  \bibfield  {author} {\bibinfo {author} {\bibfnamefont {J.}~\bibnamefont {Eisert}}, \bibinfo {author} {\bibfnamefont {F.~G.}\ \bibnamefont {Brandao}},\ and\ \bibinfo {author} {\bibfnamefont {K.~M.}\ \bibnamefont {Audenaert}},\ }\bibfield  {title} {\bibinfo {title} {Quantitative entanglement witnesses},\ }\href@noop {} {\bibfield  {journal} {\bibinfo  {journal} {New Journal of Physics}\ }\textbf {\bibinfo {volume} {9}},\ \bibinfo {pages} {46} (\bibinfo {year} {2007})}\BibitemShut {NoStop}%
\bibitem [{\citenamefont {Campbell}(2011)}]{campbell2011catalysis}%
  \BibitemOpen
  \bibfield  {author} {\bibinfo {author} {\bibfnamefont {E.~T.}\ \bibnamefont {Campbell}},\ }\bibfield  {title} {\bibinfo {title} {Catalysis and activation of magic states in fault-tolerant architectures},\ }\href {https://doi.org/10.1103/PhysRevA.83.032317} {\bibfield  {journal} {\bibinfo  {journal} {Phys. Rev. A}\ }\textbf {\bibinfo {volume} {83}},\ \bibinfo {pages} {032317} (\bibinfo {year} {2011})}\BibitemShut {NoStop}%
\bibitem [{\citenamefont {Howard}\ and\ \citenamefont {Campbell}(2017)}]{howard2017robustness}%
  \BibitemOpen
  \bibfield  {author} {\bibinfo {author} {\bibfnamefont {M.}~\bibnamefont {Howard}}\ and\ \bibinfo {author} {\bibfnamefont {E.}~\bibnamefont {Campbell}},\ }\bibfield  {title} {\bibinfo {title} {Application of a resource theory for magic states to fault-tolerant quantum computing},\ }\href {https://doi.org/10.1103/PhysRevLett.118.090501} {\bibfield  {journal} {\bibinfo  {journal} {Phys. Rev. Lett.}\ }\textbf {\bibinfo {volume} {118}},\ \bibinfo {pages} {090501} (\bibinfo {year} {2017})}\BibitemShut {NoStop}%
\bibitem [{\citenamefont {Rall}\ \emph {et~al.}(2019)\citenamefont {Rall}, \citenamefont {Liang}, \citenamefont {Cook},\ and\ \citenamefont {Kretschmer}}]{rall2019simulation}%
  \BibitemOpen
  \bibfield  {author} {\bibinfo {author} {\bibfnamefont {P.}~\bibnamefont {Rall}}, \bibinfo {author} {\bibfnamefont {D.}~\bibnamefont {Liang}}, \bibinfo {author} {\bibfnamefont {J.}~\bibnamefont {Cook}},\ and\ \bibinfo {author} {\bibfnamefont {W.}~\bibnamefont {Kretschmer}},\ }\bibfield  {title} {\bibinfo {title} {Simulation of qubit quantum circuits via pauli propagation},\ }\href {https://doi.org/10.1103/PhysRevA.99.062337} {\bibfield  {journal} {\bibinfo  {journal} {Phys. Rev. A}\ }\textbf {\bibinfo {volume} {99}},\ \bibinfo {pages} {062337} (\bibinfo {year} {2019})}\BibitemShut {NoStop}%
\bibitem [{\citenamefont {Bravyi}\ \emph {et~al.}(2019)\citenamefont {Bravyi}, \citenamefont {Browne}, \citenamefont {Calpin}, \citenamefont {Campbell}, \citenamefont {Gosset},\ and\ \citenamefont {Howard}}]{bravyi2019simulation}%
  \BibitemOpen
  \bibfield  {author} {\bibinfo {author} {\bibfnamefont {S.}~\bibnamefont {Bravyi}}, \bibinfo {author} {\bibfnamefont {D.}~\bibnamefont {Browne}}, \bibinfo {author} {\bibfnamefont {P.}~\bibnamefont {Calpin}}, \bibinfo {author} {\bibfnamefont {E.}~\bibnamefont {Campbell}}, \bibinfo {author} {\bibfnamefont {D.}~\bibnamefont {Gosset}},\ and\ \bibinfo {author} {\bibfnamefont {M.}~\bibnamefont {Howard}},\ }\bibfield  {title} {\bibinfo {title} {Simulation of quantum circuits by low-rank stabilizer decompositions},\ }\href {https://doi.org/10.22331/q-2019-09-02-181} {\bibfield  {journal} {\bibinfo  {journal} {Quantum}\ }\textbf {\bibinfo {volume} {3}},\ \bibinfo {pages} {181} (\bibinfo {year} {2019})}\BibitemShut {NoStop}%
\bibitem [{\citenamefont {Rubboli}\ \emph {et~al.}(2024)\citenamefont {Rubboli}, \citenamefont {Takagi},\ and\ \citenamefont {Tomamichel}}]{rubboli2024mixed}%
  \BibitemOpen
  \bibfield  {author} {\bibinfo {author} {\bibfnamefont {R.}~\bibnamefont {Rubboli}}, \bibinfo {author} {\bibfnamefont {R.}~\bibnamefont {Takagi}},\ and\ \bibinfo {author} {\bibfnamefont {M.}~\bibnamefont {Tomamichel}},\ }\bibfield  {title} {\bibinfo {title} {Mixed-state additivity properties of magic monotones based on quantum relative entropies for single-qubit states and beyond},\ }\href@noop {} {\bibfield  {journal} {\bibinfo  {journal} {Quantum}\ }\textbf {\bibinfo {volume} {8}},\ \bibinfo {pages} {1492} (\bibinfo {year} {2024})}\BibitemShut {NoStop}%
\bibitem [{\citenamefont {Baldwin}\ and\ \citenamefont {Jones}(2023)}]{baldwin2023efficiently}%
  \BibitemOpen
  \bibfield  {author} {\bibinfo {author} {\bibfnamefont {A.~J.}\ \bibnamefont {Baldwin}}\ and\ \bibinfo {author} {\bibfnamefont {J.~A.}\ \bibnamefont {Jones}},\ }\bibfield  {title} {\bibinfo {title} {Efficiently computing the uhlmann fidelity for density matrices},\ }\href@noop {} {\bibfield  {journal} {\bibinfo  {journal} {Physical Review A}\ }\textbf {\bibinfo {volume} {107}},\ \bibinfo {pages} {012427} (\bibinfo {year} {2023})}\BibitemShut {NoStop}%
\bibitem [{\citenamefont {Leone}\ \emph {et~al.}(2024)\citenamefont {Leone}, \citenamefont {Oliviero},\ and\ \citenamefont {Hamma}}]{leone2024learning}%
  \BibitemOpen
  \bibfield  {author} {\bibinfo {author} {\bibfnamefont {L.}~\bibnamefont {Leone}}, \bibinfo {author} {\bibfnamefont {S.~F.}\ \bibnamefont {Oliviero}},\ and\ \bibinfo {author} {\bibfnamefont {A.}~\bibnamefont {Hamma}},\ }\bibfield  {title} {\bibinfo {title} {Learning t-doped stabilizer states},\ }\href@noop {} {\bibfield  {journal} {\bibinfo  {journal} {Quantum}\ }\textbf {\bibinfo {volume} {8}},\ \bibinfo {pages} {1361} (\bibinfo {year} {2024})}\BibitemShut {NoStop}%
\bibitem [{\citenamefont {Eastin}\ and\ \citenamefont {Knill}(2009{\natexlab{b}})}]{eastin2009restrictions}%
  \BibitemOpen
  \bibfield  {author} {\bibinfo {author} {\bibfnamefont {B.}~\bibnamefont {Eastin}}\ and\ \bibinfo {author} {\bibfnamefont {E.}~\bibnamefont {Knill}},\ }\bibfield  {title} {\bibinfo {title} {Restrictions on transversal encoded quantum gate sets},\ }\href@noop {} {\bibfield  {journal} {\bibinfo  {journal} {Physical review letters}\ }\textbf {\bibinfo {volume} {102}},\ \bibinfo {pages} {110502} (\bibinfo {year} {2009}{\natexlab{b}})}\BibitemShut {NoStop}%
\bibitem [{\citenamefont {Wright}(2016)}]{wright2016learn}%
  \BibitemOpen
  \bibfield  {author} {\bibinfo {author} {\bibfnamefont {J.}~\bibnamefont {Wright}},\ }\emph {\bibinfo {title} {How to learn a quantum state}},\ \href@noop {} {Ph.D. thesis},\ \bibinfo  {school} {Carnegie Mellon University} (\bibinfo {year} {2016})\BibitemShut {NoStop}%
\bibitem [{\citenamefont {Turkeshi}\ \emph {et~al.}(2025)\citenamefont {Turkeshi}, \citenamefont {Dymarsky},\ and\ \citenamefont {Sierant}}]{turkeshi2023pauli}%
  \BibitemOpen
  \bibfield  {author} {\bibinfo {author} {\bibfnamefont {X.}~\bibnamefont {Turkeshi}}, \bibinfo {author} {\bibfnamefont {A.}~\bibnamefont {Dymarsky}},\ and\ \bibinfo {author} {\bibfnamefont {P.}~\bibnamefont {Sierant}},\ }\bibfield  {title} {\bibinfo {title} {Pauli spectrum and nonstabilizerness of typical quantum many-body states},\ }\bibfield  {journal} {\bibinfo  {journal} {Physical Review B}\ }\textbf {\bibinfo {volume} {111}},\ \href {https://doi.org/10.1103/physrevb.111.054301} {10.1103/physrevb.111.054301} (\bibinfo {year} {2025})\BibitemShut {NoStop}%
\bibitem [{\citenamefont {Brandao}\ \emph {et~al.}(2016)\citenamefont {Brandao}, \citenamefont {Harrow},\ and\ \citenamefont {Horodecki}}]{brandao2016local}%
  \BibitemOpen
  \bibfield  {author} {\bibinfo {author} {\bibfnamefont {F.~G.}\ \bibnamefont {Brandao}}, \bibinfo {author} {\bibfnamefont {A.~W.}\ \bibnamefont {Harrow}},\ and\ \bibinfo {author} {\bibfnamefont {M.}~\bibnamefont {Horodecki}},\ }\bibfield  {title} {\bibinfo {title} {Local random quantum circuits are approximate polynomial-designs},\ }\href@noop {} {\bibfield  {journal} {\bibinfo  {journal} {Communications in Mathematical Physics}\ }\textbf {\bibinfo {volume} {346}},\ \bibinfo {pages} {397} (\bibinfo {year} {2016})}\BibitemShut {NoStop}%
\bibitem [{\citenamefont {Arute}\ \emph {et~al.}(2019)\citenamefont {Arute}, \citenamefont {Arya}, \citenamefont {Babbush}, \citenamefont {Bacon}, \citenamefont {Bardin}, \citenamefont {Barends}, \citenamefont {Biswas}, \citenamefont {Boixo}, \citenamefont {Brandao}, \citenamefont {Buell} \emph {et~al.}}]{arute2019quantum}%
  \BibitemOpen
  \bibfield  {author} {\bibinfo {author} {\bibfnamefont {F.}~\bibnamefont {Arute}}, \bibinfo {author} {\bibfnamefont {K.}~\bibnamefont {Arya}}, \bibinfo {author} {\bibfnamefont {R.}~\bibnamefont {Babbush}}, \bibinfo {author} {\bibfnamefont {D.}~\bibnamefont {Bacon}}, \bibinfo {author} {\bibfnamefont {J.~C.}\ \bibnamefont {Bardin}}, \bibinfo {author} {\bibfnamefont {R.}~\bibnamefont {Barends}}, \bibinfo {author} {\bibfnamefont {R.}~\bibnamefont {Biswas}}, \bibinfo {author} {\bibfnamefont {S.}~\bibnamefont {Boixo}}, \bibinfo {author} {\bibfnamefont {F.~G.}\ \bibnamefont {Brandao}}, \bibinfo {author} {\bibfnamefont {D.~A.}\ \bibnamefont {Buell}}, \emph {et~al.},\ }\bibfield  {title} {\bibinfo {title} {Quantum supremacy using a programmable superconducting processor},\ }\href@noop {} {\bibfield  {journal} {\bibinfo  {journal} {Nature}\ }\textbf {\bibinfo {volume} {574}},\ \bibinfo {pages} {505} (\bibinfo {year} {2019})}\BibitemShut {NoStop}%
\bibitem [{\citenamefont {Bharti}\ \emph {et~al.}(2022)\citenamefont {Bharti}, \citenamefont {Cervera-Lierta}, \citenamefont {Kyaw}, \citenamefont {Haug}, \citenamefont {Alperin-Lea}, \citenamefont {Anand}, \citenamefont {Degroote}, \citenamefont {Heimonen}, \citenamefont {Kottmann}, \citenamefont {Menke} \emph {et~al.}}]{bharti2022noisy}%
  \BibitemOpen
  \bibfield  {author} {\bibinfo {author} {\bibfnamefont {K.}~\bibnamefont {Bharti}}, \bibinfo {author} {\bibfnamefont {A.}~\bibnamefont {Cervera-Lierta}}, \bibinfo {author} {\bibfnamefont {T.~H.}\ \bibnamefont {Kyaw}}, \bibinfo {author} {\bibfnamefont {T.}~\bibnamefont {Haug}}, \bibinfo {author} {\bibfnamefont {S.}~\bibnamefont {Alperin-Lea}}, \bibinfo {author} {\bibfnamefont {A.}~\bibnamefont {Anand}}, \bibinfo {author} {\bibfnamefont {M.}~\bibnamefont {Degroote}}, \bibinfo {author} {\bibfnamefont {H.}~\bibnamefont {Heimonen}}, \bibinfo {author} {\bibfnamefont {J.~S.}\ \bibnamefont {Kottmann}}, \bibinfo {author} {\bibfnamefont {T.}~\bibnamefont {Menke}}, \emph {et~al.},\ }\bibfield  {title} {\bibinfo {title} {Noisy intermediate-scale quantum algorithms},\ }\href@noop {} {\bibfield  {journal} {\bibinfo  {journal} {Reviews of Modern Physics}\ }\textbf {\bibinfo {volume} {94}},\ \bibinfo {pages} {015004} (\bibinfo {year} {2022})}\BibitemShut {NoStop}%
\bibitem [{\citenamefont {Wang}\ \emph {et~al.}(2021)\citenamefont {Wang}, \citenamefont {Fontana}, \citenamefont {Cerezo}, \citenamefont {Sharma}, \citenamefont {Sone}, \citenamefont {Cincio},\ and\ \citenamefont {Coles}}]{wang2021noise}%
  \BibitemOpen
  \bibfield  {author} {\bibinfo {author} {\bibfnamefont {S.}~\bibnamefont {Wang}}, \bibinfo {author} {\bibfnamefont {E.}~\bibnamefont {Fontana}}, \bibinfo {author} {\bibfnamefont {M.}~\bibnamefont {Cerezo}}, \bibinfo {author} {\bibfnamefont {K.}~\bibnamefont {Sharma}}, \bibinfo {author} {\bibfnamefont {A.}~\bibnamefont {Sone}}, \bibinfo {author} {\bibfnamefont {L.}~\bibnamefont {Cincio}},\ and\ \bibinfo {author} {\bibfnamefont {P.~J.}\ \bibnamefont {Coles}},\ }\bibfield  {title} {\bibinfo {title} {Noise-induced barren plateaus in variational quantum algorithms},\ }\href@noop {} {\bibfield  {journal} {\bibinfo  {journal} {Nature communications}\ }\textbf {\bibinfo {volume} {12}},\ \bibinfo {pages} {6961} (\bibinfo {year} {2021})}\BibitemShut {NoStop}%
\bibitem [{\citenamefont {Chen}\ \emph {et~al.}(2023)\citenamefont {Chen}, \citenamefont {Cotler}, \citenamefont {Huang},\ and\ \citenamefont {Li}}]{chen2023complexity}%
  \BibitemOpen
  \bibfield  {author} {\bibinfo {author} {\bibfnamefont {S.}~\bibnamefont {Chen}}, \bibinfo {author} {\bibfnamefont {J.}~\bibnamefont {Cotler}}, \bibinfo {author} {\bibfnamefont {H.-Y.}\ \bibnamefont {Huang}},\ and\ \bibinfo {author} {\bibfnamefont {J.}~\bibnamefont {Li}},\ }\bibfield  {title} {\bibinfo {title} {The complexity of nisq},\ }\href@noop {} {\bibfield  {journal} {\bibinfo  {journal} {Nature Communications}\ }\textbf {\bibinfo {volume} {14}},\ \bibinfo {pages} {6001} (\bibinfo {year} {2023})}\BibitemShut {NoStop}%
\bibitem [{\citenamefont {Haferkamp}\ \emph {et~al.}(2022)\citenamefont {Haferkamp}, \citenamefont {Montealegre-Mora}, \citenamefont {Heinrich}, \citenamefont {Eisert}, \citenamefont {Gross},\ and\ \citenamefont {Roth}}]{haferkamp2022efficient}%
  \BibitemOpen
  \bibfield  {author} {\bibinfo {author} {\bibfnamefont {J.}~\bibnamefont {Haferkamp}}, \bibinfo {author} {\bibfnamefont {F.}~\bibnamefont {Montealegre-Mora}}, \bibinfo {author} {\bibfnamefont {M.}~\bibnamefont {Heinrich}}, \bibinfo {author} {\bibfnamefont {J.}~\bibnamefont {Eisert}}, \bibinfo {author} {\bibfnamefont {D.}~\bibnamefont {Gross}},\ and\ \bibinfo {author} {\bibfnamefont {I.}~\bibnamefont {Roth}},\ }\bibfield  {title} {\bibinfo {title} {Efficient unitary designs with a system-size independent number of non-clifford gates},\ }\href@noop {} {\bibfield  {journal} {\bibinfo  {journal} {Communications in Mathematical Physics}\ ,\ \bibinfo {pages} {1}} (\bibinfo {year} {2022})}\BibitemShut {NoStop}%
\bibitem [{\citenamefont {Tarabunga}\ \emph {et~al.}(2024{\natexlab{b}})\citenamefont {Tarabunga}, \citenamefont {Tirrito}, \citenamefont {Bañuls},\ and\ \citenamefont {Dalmonte}}]{tarabunga2024nonstabilizernessmps}%
  \BibitemOpen
  \bibfield  {author} {\bibinfo {author} {\bibfnamefont {P.~S.}\ \bibnamefont {Tarabunga}}, \bibinfo {author} {\bibfnamefont {E.}~\bibnamefont {Tirrito}}, \bibinfo {author} {\bibfnamefont {M.~C.}\ \bibnamefont {Bañuls}},\ and\ \bibinfo {author} {\bibfnamefont {M.}~\bibnamefont {Dalmonte}},\ }\bibfield  {title} {\bibinfo {title} {Nonstabilizerness via matrix product states in the pauli basis},\ }\href {https://doi.org/10.1103/PhysRevLett.133.010601} {\bibfield  {journal} {\bibinfo  {journal} {Phys. Rev. Lett.}\ }\textbf {\bibinfo {volume} {133}},\ \bibinfo {pages} {010601} (\bibinfo {year} {2024}{\natexlab{b}})}\BibitemShut {NoStop}%
\bibitem [{\citenamefont {Leone}\ and\ \citenamefont {Bittel}(2024)}]{leone2024stabilizer}%
  \BibitemOpen
  \bibfield  {author} {\bibinfo {author} {\bibfnamefont {L.}~\bibnamefont {Leone}}\ and\ \bibinfo {author} {\bibfnamefont {L.}~\bibnamefont {Bittel}},\ }\bibfield  {title} {\bibinfo {title} {Stabilizer entropies are monotones for magic-state resource theory},\ }\href@noop {} {\bibfield  {journal} {\bibinfo  {journal} {Physical Review A}\ }\textbf {\bibinfo {volume} {110}},\ \bibinfo {pages} {L040403} (\bibinfo {year} {2024})}\BibitemShut {NoStop}%
\bibitem [{\citenamefont {Wei}\ and\ \citenamefont {Liu}(2025)}]{wei2025longrangenonstabilizerness}%
  \BibitemOpen
  \bibfield  {author} {\bibinfo {author} {\bibfnamefont {F.}~\bibnamefont {Wei}}\ and\ \bibinfo {author} {\bibfnamefont {Z.-W.}\ \bibnamefont {Liu}},\ }\bibfield  {title} {\bibinfo {title} {Long-range nonstabilizerness from topology and correlation},\ }\href {https://arxiv.org/abs/2503.04566} {\bibfield  {journal} {\bibinfo  {journal} {arxiv:2503.04566}\ } (\bibinfo {year} {2025})}\BibitemShut {NoStop}%
\bibitem [{\citenamefont {Schuster}\ \emph {et~al.}(2024)\citenamefont {Schuster}, \citenamefont {Haferkamp},\ and\ \citenamefont {Huang}}]{schuster2024randomunitariesextremelylow}%
  \BibitemOpen
  \bibfield  {author} {\bibinfo {author} {\bibfnamefont {T.}~\bibnamefont {Schuster}}, \bibinfo {author} {\bibfnamefont {J.}~\bibnamefont {Haferkamp}},\ and\ \bibinfo {author} {\bibfnamefont {H.-Y.}\ \bibnamefont {Huang}},\ }\bibfield  {title} {\bibinfo {title} {Random unitaries in extremely low depth},\ }\href {https://arxiv.org/abs/2407.07754} {\bibfield  {journal} {\bibinfo  {journal} {arXiv preprint arXiv:2407.07754}\ } (\bibinfo {year} {2024})}\BibitemShut {NoStop}%
\bibitem [{\citenamefont {Elben}\ \emph {et~al.}(2020)\citenamefont {Elben}, \citenamefont {Kueng}, \citenamefont {Huang}, \citenamefont {van Bijnen}, \citenamefont {Kokail}, \citenamefont {Dalmonte}, \citenamefont {Calabrese}, \citenamefont {Kraus}, \citenamefont {Preskill}, \citenamefont {Zoller},\ and\ \citenamefont {Vermersch}}]{elben2020mixed}%
  \BibitemOpen
  \bibfield  {author} {\bibinfo {author} {\bibfnamefont {A.}~\bibnamefont {Elben}}, \bibinfo {author} {\bibfnamefont {R.}~\bibnamefont {Kueng}}, \bibinfo {author} {\bibfnamefont {H.-Y.~R.}\ \bibnamefont {Huang}}, \bibinfo {author} {\bibfnamefont {R.}~\bibnamefont {van Bijnen}}, \bibinfo {author} {\bibfnamefont {C.}~\bibnamefont {Kokail}}, \bibinfo {author} {\bibfnamefont {M.}~\bibnamefont {Dalmonte}}, \bibinfo {author} {\bibfnamefont {P.}~\bibnamefont {Calabrese}}, \bibinfo {author} {\bibfnamefont {B.}~\bibnamefont {Kraus}}, \bibinfo {author} {\bibfnamefont {J.}~\bibnamefont {Preskill}}, \bibinfo {author} {\bibfnamefont {P.}~\bibnamefont {Zoller}},\ and\ \bibinfo {author} {\bibfnamefont {B.}~\bibnamefont {Vermersch}},\ }\bibfield  {title} {\bibinfo {title} {Mixed-state entanglement from local randomized measurements},\ }\bibfield  {journal} {\bibinfo  {journal} {Physical Review Letters}\ }\textbf {\bibinfo {volume} {125}},\ \href {https://doi.org/10.1103/physrevlett.125.200501}
  {10.1103/physrevlett.125.200501} (\bibinfo {year} {2020})\BibitemShut {NoStop}%
\bibitem [{\citenamefont {Haug}\ \emph {et~al.}(2024{\natexlab{b}})\citenamefont {Haug}, \citenamefont {Aolita},\ and\ \citenamefont {Kim}}]{haug2024probing}%
  \BibitemOpen
  \bibfield  {author} {\bibinfo {author} {\bibfnamefont {T.}~\bibnamefont {Haug}}, \bibinfo {author} {\bibfnamefont {L.}~\bibnamefont {Aolita}},\ and\ \bibinfo {author} {\bibfnamefont {M.}~\bibnamefont {Kim}},\ }\bibfield  {title} {\bibinfo {title} {Probing quantum complexity via universal saturation of stabilizer entropies},\ }\href {https://arxiv.org/abs/2406.04190} {\bibfield  {journal} {\bibinfo  {journal} {arXiv:2406.04190}\ } (\bibinfo {year} {2024}{\natexlab{b}})}\BibitemShut {NoStop}%
\bibitem [{\citenamefont {Tirrito}\ \emph {et~al.}(2024)\citenamefont {Tirrito}, \citenamefont {Tarabunga}, \citenamefont {Lami}, \citenamefont {Chanda}, \citenamefont {Leone}, \citenamefont {Oliviero}, \citenamefont {Dalmonte}, \citenamefont {Collura},\ and\ \citenamefont {Hamma}}]{tirrito2023quantifying}%
  \BibitemOpen
  \bibfield  {author} {\bibinfo {author} {\bibfnamefont {E.}~\bibnamefont {Tirrito}}, \bibinfo {author} {\bibfnamefont {P.~S.}\ \bibnamefont {Tarabunga}}, \bibinfo {author} {\bibfnamefont {G.}~\bibnamefont {Lami}}, \bibinfo {author} {\bibfnamefont {T.}~\bibnamefont {Chanda}}, \bibinfo {author} {\bibfnamefont {L.}~\bibnamefont {Leone}}, \bibinfo {author} {\bibfnamefont {S.~F.}\ \bibnamefont {Oliviero}}, \bibinfo {author} {\bibfnamefont {M.}~\bibnamefont {Dalmonte}}, \bibinfo {author} {\bibfnamefont {M.}~\bibnamefont {Collura}},\ and\ \bibinfo {author} {\bibfnamefont {A.}~\bibnamefont {Hamma}},\ }\bibfield  {title} {\bibinfo {title} {Quantifying nonstabilizerness through entanglement spectrum flatness},\ }\href {https://doi.org/10.1103/PhysRevA.109.L040401} {\bibfield  {journal} {\bibinfo  {journal} {Physical Review A}\ }\textbf {\bibinfo {volume} {109}},\ \bibinfo {pages} {L040401} (\bibinfo {year} {2024})}\BibitemShut {NoStop}%
\bibitem [{\citenamefont {Turkeshi}\ \emph {et~al.}(2023)\citenamefont {Turkeshi}, \citenamefont {Schir{\`o}},\ and\ \citenamefont {Sierant}}]{turkeshi2023measuring}%
  \BibitemOpen
  \bibfield  {author} {\bibinfo {author} {\bibfnamefont {X.}~\bibnamefont {Turkeshi}}, \bibinfo {author} {\bibfnamefont {M.}~\bibnamefont {Schir{\`o}}},\ and\ \bibinfo {author} {\bibfnamefont {P.}~\bibnamefont {Sierant}},\ }\bibfield  {title} {\bibinfo {title} {Measuring nonstabilizerness via multifractal flatness},\ }\href@noop {} {\bibfield  {journal} {\bibinfo  {journal} {Physical Review A}\ }\textbf {\bibinfo {volume} {108}},\ \bibinfo {pages} {042408} (\bibinfo {year} {2023})}\BibitemShut {NoStop}%
\bibitem [{\citenamefont {Liu}\ \emph {et~al.}(2024)\citenamefont {Liu}, \citenamefont {Gong}, \citenamefont {Du},\ and\ \citenamefont {Cai}}]{liu2024exponential}%
  \BibitemOpen
  \bibfield  {author} {\bibinfo {author} {\bibfnamefont {Z.}~\bibnamefont {Liu}}, \bibinfo {author} {\bibfnamefont {W.}~\bibnamefont {Gong}}, \bibinfo {author} {\bibfnamefont {Z.}~\bibnamefont {Du}},\ and\ \bibinfo {author} {\bibfnamefont {Z.}~\bibnamefont {Cai}},\ }\bibfield  {title} {\bibinfo {title} {Exponential separations between quantum learning with and without purification},\ }\href@noop {} {\bibfield  {journal} {\bibinfo  {journal} {arXiv preprint arXiv:2410.17718}\ } (\bibinfo {year} {2024})}\BibitemShut {NoStop}%
\bibitem [{\citenamefont {Yu}\ \emph {et~al.}(2021)\citenamefont {Yu}, \citenamefont {Imai},\ and\ \citenamefont {G\"{u}hne}}]{yu2021optimal}%
  \BibitemOpen
  \bibfield  {author} {\bibinfo {author} {\bibfnamefont {X.-D.}\ \bibnamefont {Yu}}, \bibinfo {author} {\bibfnamefont {S.}~\bibnamefont {Imai}},\ and\ \bibinfo {author} {\bibfnamefont {O.}~\bibnamefont {G\"{u}hne}},\ }\bibfield  {title} {\bibinfo {title} {Optimal entanglement certification from moments of the partial transpose},\ }\bibfield  {journal} {\bibinfo  {journal} {Physical Review Letters}\ }\textbf {\bibinfo {volume} {127}},\ \href {https://doi.org/10.1103/physrevlett.127.060504} {10.1103/physrevlett.127.060504} (\bibinfo {year} {2021})\BibitemShut {NoStop}%
\bibitem [{\citenamefont {Neven}\ \emph {et~al.}(2021)\citenamefont {Neven}, \citenamefont {Carrasco}, \citenamefont {Vitale}, \citenamefont {Kokail}, \citenamefont {Elben}, \citenamefont {Dalmonte}, \citenamefont {Calabrese}, \citenamefont {Zoller}, \citenamefont {Vermersch}, \citenamefont {Kueng},\ and\ \citenamefont {Kraus}}]{neven2021symmetryresolved}%
  \BibitemOpen
  \bibfield  {author} {\bibinfo {author} {\bibfnamefont {A.}~\bibnamefont {Neven}}, \bibinfo {author} {\bibfnamefont {J.}~\bibnamefont {Carrasco}}, \bibinfo {author} {\bibfnamefont {V.}~\bibnamefont {Vitale}}, \bibinfo {author} {\bibfnamefont {C.}~\bibnamefont {Kokail}}, \bibinfo {author} {\bibfnamefont {A.}~\bibnamefont {Elben}}, \bibinfo {author} {\bibfnamefont {M.}~\bibnamefont {Dalmonte}}, \bibinfo {author} {\bibfnamefont {P.}~\bibnamefont {Calabrese}}, \bibinfo {author} {\bibfnamefont {P.}~\bibnamefont {Zoller}}, \bibinfo {author} {\bibfnamefont {B.}~\bibnamefont {Vermersch}}, \bibinfo {author} {\bibfnamefont {R.}~\bibnamefont {Kueng}},\ and\ \bibinfo {author} {\bibfnamefont {B.}~\bibnamefont {Kraus}},\ }\bibfield  {title} {\bibinfo {title} {Symmetry-resolved entanglement detection using partial transpose moments},\ }\bibfield  {journal} {\bibinfo  {journal} {npj Quantum Information}\ }\textbf {\bibinfo {volume} {7}},\ \href {https://doi.org/10.1038/s41534-021-00487-y} {10.1038/s41534-021-00487-y}
  (\bibinfo {year} {2021})\BibitemShut {NoStop}%
\bibitem [{\citenamefont {Veitch}\ \emph {et~al.}(2012)\citenamefont {Veitch}, \citenamefont {Ferrie}, \citenamefont {Gross},\ and\ \citenamefont {Emerson}}]{veitch2012negative}%
  \BibitemOpen
  \bibfield  {author} {\bibinfo {author} {\bibfnamefont {V.}~\bibnamefont {Veitch}}, \bibinfo {author} {\bibfnamefont {C.}~\bibnamefont {Ferrie}}, \bibinfo {author} {\bibfnamefont {D.}~\bibnamefont {Gross}},\ and\ \bibinfo {author} {\bibfnamefont {J.}~\bibnamefont {Emerson}},\ }\bibfield  {title} {\bibinfo {title} {Negative quasi-probability as a resource for quantum computation},\ }\href {https://doi.org/10.1088/1367-2630/14/11/113011} {\bibfield  {journal} {\bibinfo  {journal} {New Journal of Physics}\ }\textbf {\bibinfo {volume} {14}},\ \bibinfo {pages} {113011} (\bibinfo {year} {2012})}\BibitemShut {NoStop}%
\bibitem [{\citenamefont {Tarabunga}(2024)}]{tarabunga2024critical}%
  \BibitemOpen
  \bibfield  {author} {\bibinfo {author} {\bibfnamefont {P.~S.}\ \bibnamefont {Tarabunga}},\ }\bibfield  {title} {\bibinfo {title} {Critical behaviors of non-stabilizerness in quantum spin chains},\ }\href {https://doi.org/10.22331/q-2024-07-17-1413} {\bibfield  {journal} {\bibinfo  {journal} {Quantum}\ }\textbf {\bibinfo {volume} {8}},\ \bibinfo {pages} {1413} (\bibinfo {year} {2024})}\BibitemShut {NoStop}%
\bibitem [{\citenamefont {Fishman}\ \emph {et~al.}(2022)\citenamefont {Fishman}, \citenamefont {White},\ and\ \citenamefont {Stoudenmire}}]{itensor}%
  \BibitemOpen
  \bibfield  {author} {\bibinfo {author} {\bibfnamefont {M.}~\bibnamefont {Fishman}}, \bibinfo {author} {\bibfnamefont {S.}~\bibnamefont {White}},\ and\ \bibinfo {author} {\bibfnamefont {E.}~\bibnamefont {Stoudenmire}},\ }\bibfield  {title} {\bibinfo {title} {The itensor software library for tensor network calculations},\ }\href@noop {} {\bibfield  {journal} {\bibinfo  {journal} {SciPost Physics Codebases}\ ,\ \bibinfo {pages} {004}} (\bibinfo {year} {2022})}\BibitemShut {NoStop}%
\bibitem [{\citenamefont {Lin}\ and\ \citenamefont {Tomamichel}(2015)}]{lin2015investigating}%
  \BibitemOpen
  \bibfield  {author} {\bibinfo {author} {\bibfnamefont {S.~M.}\ \bibnamefont {Lin}}\ and\ \bibinfo {author} {\bibfnamefont {M.}~\bibnamefont {Tomamichel}},\ }\bibfield  {title} {\bibinfo {title} {Investigating properties of a family of quantum r{\'e}nyi divergences},\ }\href@noop {} {\bibfield  {journal} {\bibinfo  {journal} {Quantum Information Processing}\ }\textbf {\bibinfo {volume} {14}},\ \bibinfo {pages} {1501} (\bibinfo {year} {2015})}\BibitemShut {NoStop}%
\bibitem [{\citenamefont {Seddon}\ \emph {et~al.}(2021)\citenamefont {Seddon}, \citenamefont {Regula}, \citenamefont {Pashayan}, \citenamefont {Ouyang},\ and\ \citenamefont {Campbell}}]{seddon2021quantifying}%
  \BibitemOpen
  \bibfield  {author} {\bibinfo {author} {\bibfnamefont {J.~R.}\ \bibnamefont {Seddon}}, \bibinfo {author} {\bibfnamefont {B.}~\bibnamefont {Regula}}, \bibinfo {author} {\bibfnamefont {H.}~\bibnamefont {Pashayan}}, \bibinfo {author} {\bibfnamefont {Y.}~\bibnamefont {Ouyang}},\ and\ \bibinfo {author} {\bibfnamefont {E.~T.}\ \bibnamefont {Campbell}},\ }\bibfield  {title} {\bibinfo {title} {Quantifying quantum speedups: Improved classical simulation from tighter magic monotones},\ }\href {https://doi.org/10.1103/PRXQuantum.2.010345} {\bibfield  {journal} {\bibinfo  {journal} {PRX Quantum}\ }\textbf {\bibinfo {volume} {2}},\ \bibinfo {pages} {010345} (\bibinfo {year} {2021})}\BibitemShut {NoStop}%
\bibitem [{\citenamefont {Chen}(2014)}]{chen2014brief}%
  \BibitemOpen
  \bibfield  {author} {\bibinfo {author} {\bibfnamefont {E.}~\bibnamefont {Chen}},\ }\bibfield  {title} {\bibinfo {title} {A brief introduction to olympiad inequalities},\ }\href@noop {} {\bibfield  {journal} {\bibinfo  {journal} {https://ghoshadi.wordpress.com/wp-content/uploads/2018/05/evan-chens-notes.pdf}\ } (\bibinfo {year} {2014})}\BibitemShut {NoStop}%
\bibitem [{\citenamefont {Braunstein}(1996)}]{braunstein1996geometry}%
  \BibitemOpen
  \bibfield  {author} {\bibinfo {author} {\bibfnamefont {S.~L.}\ \bibnamefont {Braunstein}},\ }\bibfield  {title} {\bibinfo {title} {Geometry of quantum inference},\ }\href {https://doi.org/10.1016/0375-9601(96)00365-9} {\bibfield  {journal} {\bibinfo  {journal} {Physics Letters A}\ }\textbf {\bibinfo {volume} {219}},\ \bibinfo {pages} {169} (\bibinfo {year} {1996})}\BibitemShut {NoStop}%
\bibitem [{\citenamefont {Hall}(1998)}]{hall1998random}%
  \BibitemOpen
  \bibfield  {author} {\bibinfo {author} {\bibfnamefont {M.~J.}\ \bibnamefont {Hall}},\ }\bibfield  {title} {\bibinfo {title} {Random quantum correlations and density operator distributions},\ }\href {https://doi.org/10.1016/S0375-9601(98)00190-X} {\bibfield  {journal} {\bibinfo  {journal} {Physics Letters A}\ }\textbf {\bibinfo {volume} {242}},\ \bibinfo {pages} {123} (\bibinfo {year} {1998})}\BibitemShut {NoStop}%
\bibitem [{\citenamefont {Zyczkowski}\ and\ \citenamefont {Sommers}(2001)}]{Zyczkowski_2001}%
  \BibitemOpen
  \bibfield  {author} {\bibinfo {author} {\bibfnamefont {K.}~\bibnamefont {Zyczkowski}}\ and\ \bibinfo {author} {\bibfnamefont {H.-J.}\ \bibnamefont {Sommers}},\ }\bibfield  {title} {\bibinfo {title} {Induced measures in the space of mixed quantum states},\ }\href {https://doi.org/10.1088/0305-4470/34/35/335} {\bibfield  {journal} {\bibinfo  {journal} {Journal of Physics A: Mathematical and General}\ }\textbf {\bibinfo {volume} {34}},\ \bibinfo {pages} {7111} (\bibinfo {year} {2001})}\BibitemShut {NoStop}%
\end{thebibliography}%

\let\addcontentsline\oldaddcontentsline

\appendix

\onecolumngrid
\newpage 

\setcounter{secnumdepth}{2}
\setcounter{equation}{0}
\setcounter{figure}{0}
\setcounter{section}{0}

\renewcommand{\thesection}{\Alph{section}}
\renewcommand{\thesubsection}{\arabic{subsection}}
\renewcommand*{\theHsection}{\thesection}

\clearpage
\begin{center}

\textbf{\large \SMLong{}}
\end{center}
\setcounter{equation}{0}
\setcounter{figure}{0}
\setcounter{table}{0}

\makeatletter

\renewcommand{\thefigure}{S\arabic{figure}}

We provide proofs and additional details supporting the claims in the main text.

\makeatletter
\@starttoc{toc}

\makeatother
\section{Definitions}
We shortly review the notation for the scaling of functions:
\begin{itemize}
    \item Big-O notation $O(.)$: For a function $f(n)$, if there exists a constant $c$ and a specific input size $n_0$ such that $f(n) \leq c \cdot g(n)$ for all $n \geq n_0$, where $g(n)$ is a well-defined function, then we express it as $f(n) = O(g(n))$. This signifies the upper limit of how fast a function grows in respect to $g(n)$.
    \item Big-Omega notation $\Omega(.)$: For a function $f(n)$, if there exists a constant $c$ and a specific input size $n_0$ such that $f(n) \geq c \cdot g(n)$ for all $n \geq n_0$, where $g(n)$ is a well-defined function, then we express it as $f(n) = \Omega(g(n))$. This signifies the lower limit of how fast a function grows in respect to $g(n)$.
    \item Big-Theta notation $\Theta(.)$: For a function $f(n)$, if $f(n) = O(g(n))$ and if $f(n) = \Omega(g(n))$, where $g(n)$ is a well-defined function, then we express it as $f(n) = \Theta(g(n))$. This implies that the function grows with the same scaling as $g(n)$.
    \item Little-Omega notation $\omega(.)$: For a function $f(n)$, when for any constant $c>0$ there exist a specific input size $n_0$ such that $f(n) \geq c \cdot g(n)$ for all $n \geq n_0$, where $g(n)$ is a well-defined function, then we express it as $f(n) = \omega(g(n))$. This implies that the function grows strictly faster than $g(n)$.
    \item Little-o notation $o(.)$: For a function $f(n)$, when for any constant $c>0$ there exist a specific input size $n_0$ such that $f(n) \leq c \cdot g(n)$ for all $n \geq n_0$, where $g(n)$ is a well-defined function, then we express it as $f(n) = o(g(n))$. This signifies that the function grows strictly slower than $g(n)$.
    \item Negligible functions $\text{negl}(.)$:  Positive real-valued functions $\mu:\mathbb \mathbb{N} \to \mathbb R$ are negligible if and only if $\forall c \in \mathbb{N}$, $\exists n_0 \in \mathbb{N}$ such that $\forall n > n_0$, $\mu(n) < n^{-c}$. This means that the function decays faster than any inverse polynomial. Alternatively, one can write $o(1/\text{poly}(n))$ or $2^{-\omega(\log n)}$.
    
\end{itemize}

\section{Note on violation of monotonicity}\label{sec:monoton}
While $M_\alpha$ (and thus $\mathcal{W}_\alpha)$ for $\alpha\ge2$ is a magic monotone for pure states under Clifford operations that map pure states to pure states~\cite{leone2024stabilizer}, we note that $\mathcal{W}_\alpha$ for any $\alpha$ is not monotone under more general Clifford channels~\cite{haug2023stabilizer}.

Here, we note that a magic monotone such as log-free robustness of magic $\text{LR}(\rho)$ or the stabilizer fidelity $D_\text{F}(\rho)$ fulfil two conditions~\cite{veitch2014resource}: First, they must be zero if and only if the state is a mixture of stabilizer states $\rho_\text{C}$, i.e. $\text{LR}(\rho_\text{C})=0$. Second, they must be non-increasing under any Clifford channel $\Gamma_\text{C}$, i.e. $\text{LR}(\Gamma_\text{C}(\rho))\leq\text{LR}(\rho)$. Here, Clifford channels are any channel that can be implemented by Clifford unitaries, classical communication, classical randomness and measurements in the $z$-basis.

As an example, consider $\ket{\psi}=\cos(\pi/16)\ket{0}+\sin(\pi/16)\ket{1}$ and the mixed unital Clifford channel $\Gamma_\text{C}(\rho)=\frac{1}{2}\rho+\frac{1}{2}H_\text{d} \rho H_\text{d}$, where $H_\text{d}$ is the Hadamard gate. Here, we find that $\mathcal{W}_\alpha(\Gamma_\text{C}(\ket{\psi}))>\mathcal{W}_\alpha(\ket{\psi})$ for $1/2<\alpha<2.5$, i.e. $\mathcal{W}_\alpha$ violates monotonicity as it increases under $\Gamma_\text{C}$. Similar examples can be found for all $\alpha$. 

\section{Filtered magic witness} \label{sec:alt_witness}
We define the filtered magic witness as $\Tilde{\mathcal{W}}_\alpha$:
\begin{equation}\label{eq:filtered_witness}
    \Tilde{\mathcal{W}}_\alpha(\rho)= \frac{1}{1-\alpha}\ln \Tilde{A}_\alpha(\rho) + \frac{1-2\alpha}{1-\alpha}\ln \Tilde{A}_1(\rho)\,,
\end{equation}
where
\begin{equation}
\Tilde{A}_\alpha(\rho)=\frac{\sum_{P\in\Tilde{\mathcal{P}}_n} \vert\text{tr}(\rho P)\vert^{2\alpha}} {2^n - 1} \equiv \frac{2^nA_\alpha-1}{2^n-1} \,,
\end{equation}
where $\Tilde{\mathcal{P}}_n = \mathcal{P}_n/ \{ I\}$ and $A_1=\text{tr}(\rho^2)$. Note that $\Tilde{\mathcal{W}}_\alpha$ is not well-defined for the maximally mixed state $\rho=I/2^n$. The von Neumann limit $\alpha \to 1$ is given by
\begin{equation}
\Tilde{\mathcal{W}}_1(\rho)=-\sum_{P\in\Tilde{\mathcal{P}}_n} 
\frac{\text{tr}(\rho P)^2}{2^n\text{tr}(\rho^2)-1} \ln(\text{tr}(\rho P)^2)-2\ln \Tilde{A}_1(\rho)
\end{equation}
For pure states $\rho = \ket{\psi} \bra{\psi}$, $\Tilde{\mathcal{W}}_\alpha$ becomes the filtered SRE~\cite{turkeshi2023pauli}
\begin{equation} \label{eq:filtered_sre}
 \Tilde{\mathcal{W}}_\alpha(\ket{\psi})\equiv \Tilde{M}_\alpha(\ket{\psi})  =\frac{1}{1-\alpha}\ln(\frac{\sum_{P\in\Tilde{\mathcal{P}}_n} \bra{\psi}P\ket{\psi}^{2\alpha}} {2^n - 1})\,.
\end{equation}
Similarly to $\mathcal{W}_\alpha$, $\Tilde{\mathcal{W}}_\alpha$ is a witness of magic for mixed states for any $\alpha\geq1/2$: if $\Tilde{\mathcal{W}}_\alpha > 0$, $\rho$ is a nonstabilizer state. The proof goes analogously to $\mathcal{W}_\alpha$, by noting that 
\begin{equation}
    \Tilde{\mathcal{W}}_{1/2}(\rho) = 2 \ln \Tilde{\mathcal{D}}(\rho)\,, 
\end{equation}
where we define $\Tilde{\mathcal{D}}(\rho)$ as
\begin{equation}
\Tilde{\mathcal{D}}(\rho)=
\frac{\sum_{P\in\Tilde{\mathcal{P}}_n} \lvert \text{tr}(\rho P) \rvert} {2^n - 1} \equiv \frac{2^n \mathcal{D}(\rho)-1}{2^n-1} \,.
\end{equation}
Like the stabilizer norm, $\Tilde{\mathcal{D}}(\rho)$ is a witness of mixed-state magic as $\ln \Tilde{\mathcal{D}}(\rho) >0 $ guarantees that $\rho$ is a nonstabilizer state~\cite{howard2017robustness}. Note that the set of nonstabilizer states that can be detected by $\Tilde{\mathcal{D}}(\rho)$ and $\mathcal{D}(\rho)$ is the same since $\Tilde{\mathcal{D}}(\rho)>1$ if and only if $\mathcal{D}(\rho)>1$. Nevertheless, $\Tilde{\mathcal{D}}(\rho)$ provides a tighter lower bound to $\text{LR}(\rho)$ since $\Tilde{\mathcal{D}}(\rho) > \mathcal{D}(\rho)$ whenever $\Tilde{\mathcal{D}}(\rho)>1$.

We can then show by the hierarchy of R\'enyi entropies over the probability distribution $\Tilde{p}_\rho(P)=\text{tr}(\rho P)^2/(2^n\text{tr}(\rho^2)-1)$ for $P \neq I$ that
\begin{equation}
    2\ln \Tilde{\mathcal{D}}(\rho) \geq  \Tilde{\mathcal{W}}_\alpha(\rho) \quad (\alpha \geq 1/2)\,.
\end{equation}
It follows that the alternative magic witness also provides a lower bound on $\text{LR}(\rho)$ by
\begin{equation}
    2\text{LR}(\rho) \geq 2 \ln\Tilde{\mathcal{D}}(\rho)\geq \Tilde{\mathcal{W}}_\alpha(\rho) \quad (\alpha \geq 1/2)\,.
\end{equation}
Note that $\mathcal{W}_\alpha$ and $\Tilde{\mathcal{W}}_\alpha$ have the same asymptotic behavior for large $n$.

To illustrate the distinct behavior of $\mathcal{W}_\alpha$ and $\Tilde{\mathcal{W}}_\alpha$, we consider the single-qubit $T$ state $\ket{T} = (\ket{0} + e^{i \pi/4} \ket{1})/\sqrt{2}$ subjected to depolarizing noise $\rho_\mathrm{dp}
 = (1-p)\ket{T}\bra{T} + p I/2^n$ 
 with a probability $p$. This effectively rescales the expectation values of all Pauli strings $P \in \mathcal{P}/ \{ I\}$ by $(1-p)$. Note that, for single-qubit states, $\mathcal{D}(\rho)>1$ (and $\Tilde{\mathcal{D}}(\rho)>1$) is a necessary and sufficient condition of magic~\cite{rall2019simulation}. By direct calculation, we obtain
 \begin{equation}
     \mathcal{W}_\alpha = \frac{1}{1-\alpha}\ln (\frac{1+(1-p)^{2\alpha} 2^{1-\alpha}}{2}) + \frac{1-2\alpha}{1-\alpha}\ln(\frac{p^2-2p+2}{2})
 \end{equation}
 and
 \begin{equation}
     \Tilde{\mathcal{W}}_\alpha = \ln(2(1-p)^2).
 \end{equation}
 In particular, we find that the transition to a stabilizer state occurs at $p_c=1-1/\sqrt{2}$, where both $\mathcal{W}_{1/2}$ and $\Tilde{\mathcal{W}}_{1/2}$ change sign.
 Moreover, notice that $\Tilde{\mathcal{W}}_\alpha$ is independent of the R\'enyi index $\alpha$. This is a consequence of the fact that the distribution $\Tilde{p}_\rho(P)$ is flat for the $T$ state. This also implies that $\Tilde{\mathcal{W}}_\alpha$ obtains the same transition as $\Tilde{\mathcal{W}}_{1/2}$, thus successfully witnessing the magic at $p>p_c$. In contrast, $\mathcal{W}_\alpha$ for $\alpha>1/2$ has different transitions from $\mathcal{W}_{1/2}$. For example, we find that $\mathcal{W}_2$ changes sign at $p_{c,2}= 0.1565... < p_c$. Thus, $\mathcal{W}_2$ fails to detect the magic at $p_c > p \geq p_{c,2}$.

\section{Relationship of stabilizer fidelity and log-free robustness of magic}\label{sec:boundLR}
In this section, we show that the stabilizer fidelity is lower bounded by the log-free robustness of magic:
\begin{lemma}
The stabilizer fidelity of mixed states $D_\text{F}$ lower bounds the log-free robustness of magic $\text{LR}$
\begin{equation}
    D_\text{F}(\rho)\leq \text{LR}(\rho)\,.
\end{equation}
\end{lemma}
\begin{proof}
First, we note that $\alpha-z$ R\'enyi entropies can be written as~\cite{rubboli2024mixed}
\begin{equation}
    D_{\alpha,z}(\rho)=\min_{\sigma\in\text{STAB}}\frac{1}{\alpha-1}\ln(\text{tr}[(\rho^{\frac{\alpha}{2z}} \sigma^{\frac{1-\alpha}{z}} \rho^{\frac{\alpha}{2z}})^{z}])\,.
\end{equation}
In particular, we have the stabilizer fidelity
\begin{equation}
    D_\text{F}(\rho)=D_{1/2,1/2}(\rho)
\end{equation}
and the generalized robustness of magic
\begin{equation}
     \Lambda^+(\rho)=\lim_{\alpha\rightarrow \infty} D_{\alpha,\alpha-1}(\rho)\,.
\end{equation}
First, from Ref.~\cite{lin2015investigating} it follows that $D_\text{F}(\rho)\leq \Lambda^+(\rho)$. This is due to the fact that $D_{\alpha,\alpha}$ increases monotonically with $\alpha$ as well as decreases with increasing $z$ when $\alpha>1$.
Then, Ref.~\cite{seddon2021quantifying}  (Lemma 11) showed that $\Lambda^+(\rho)\leq \text{LR}(\rho)$.
\end{proof}

\section{Efficient measurement of witness}\label{sec:measwitness}
Here, we show how to efficiently measure our witness $\mathcal{W}_\alpha$ for odd integer $\alpha>1$. Note that the scheme to efficiently measure $A_\alpha$ was originally provided in Ref.~\cite{haug2023efficient}.

First, we note that $A_\alpha$ can be re-written into the expectation value of a multi-copy observable via the replica trick~\cite{haug2022quantifying}
\begin{align}\label{eq:contraction2}
	 A_\alpha(\rho)=2^{-n}\sum_{P \in \mathcal{P}_n} \text{tr}(\rho P)^{2\alpha}=\text{tr}(\rho^{\otimes 2\alpha}\zeta_{\alpha}^{\otimes n})\,,
\end{align}
where $\zeta_{\alpha}=\frac{1}{2}\sum_{k=0}^3 (\sigma^k)^{\otimes 2\alpha}$.
One can show that for integer odd $\alpha$, $\zeta_{\alpha}$ has eigenvalues $\{+1,-1\}$, while for even $\alpha$ it is a projector with eigenvalues $\{0,2\}$. 
One can identify $A_1(\rho)=\text{tr}(\rho^2)$ with the purity.

To measure $A_\alpha$, let us first recall the Bell transformation as shown in Fig.~\ref{fig:Bell}
acting on two qubits
\begin{equation}
U_\text{Bell}=(H\otimes I_1) \text{CNOT}\,,
\end{equation}
where $H=\frac{1}{\sqrt{2}}(\sigma^x+\sigma^z)$ is the Hadamard gate, and the CNOT gate $\text{CNOT}=\exp(i\frac{\pi}{4}(I_1-\sigma^z)\otimes(I_1-\sigma^x))$. 

\begin{figure}[htbp]
	\centering	
	\subfigimg[width=0.24\textwidth]{}{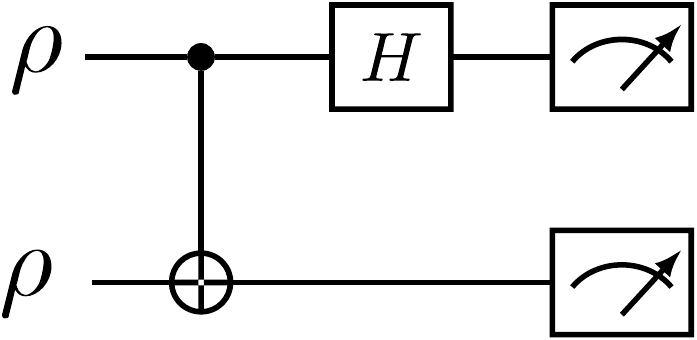}
	\caption{Circuit for Bell transformation on two copies of $n$-qubit state $\rho$, where for each qubit pair one applies a CNOT and Hadamard gate (i.e. $n$ CNOT and Hadamard gates in total). This is followed by measurements in the computational basis.
	}
	\label{fig:Bell}
\end{figure}
It turns out $\zeta_{\alpha}$ is diagonalized  by $U_\text{Bell}$
\begin{align*}
&A_\alpha=\text{tr}(\rho^{\otimes 2\alpha}({U_\text{Bell}^{\otimes \alpha}}^\dagger \frac{1}{2}( (I_1\otimes I_1)^{\otimes \alpha}+(\sigma^z\otimes I_1)^{\otimes \alpha}\numberthis\label{eq:meas}\\&
+(I_1\otimes \sigma^z)^{\otimes \alpha}+(-1)^{\alpha}(\sigma^z\otimes \sigma^z)^{\otimes \alpha})U_\text{Bell}^{\otimes \alpha})^{\otimes n})\,.
\end{align*}
Notably, while $A_\alpha$ is written in terms of $2\alpha$ copies of $\rho$, it suffices to actually only prepare $2$ copies of $\rho$ at the same time.
Using this result, Ref.~\cite{haug2023efficient} constructed Algorithm~\ref{alg:SEmeas}, which is efficient for odd integer $\alpha>1$.
Basically, one prepares two copies of $\rho$, applies Bell transformation, measures in computational basis, repeats this $O(\alpha)$ times, and  performs some classical post-processing in $O(\alpha n)$ time.
Note that the purity, which is needed for $\mathcal{W}_\alpha$, is given by $A_1(\rho)=\text{tr}(\rho^2)$, and can be computed with the same algorithm and complexity.

The total efficiency is as follows:
\begin{theorem}[Efficient witness of magic]\label{thm:witness_sup}
For a given (mixed) $n$-qubit state $\rho$ and odd $\alpha$, there exist an efficient algorithm to measure $A_\alpha(\rho)$ to additive precision $\epsilon$ with failure probability $\delta$ using $O(\alpha\epsilon^{-2}\log(2/\delta))$ copies of $\rho$, $O(1)$ circuit depth, and $O(\alpha n\log(2/\delta))$ classical post-processing time. 
\end{theorem}
\begin{proof}
This has been previously shown using Hoeffding's inequality~\cite{haug2023efficient}, which we reproduce for completeness:
Given estimator $\hat{A}_\alpha$ from Algorithm~\ref{alg:SEmeas}, mean value $A_\alpha$, additive error $\epsilon$, failure probability $\delta$, range of outcomes $\Delta\lambda$ and number $L$ of repetitions of the algorithm, Hoeffding's inequality is given by
\begin{equation}\label{eq:Hoeffding2}
    \delta\equiv\text{Pr}(\vert \hat{A}_\alpha - A_\alpha\vert\ge \epsilon)\le 2\exp\left(-\frac{2\epsilon^2 L}{\Delta\lambda^2}\right)\,.
\end{equation}
Now, we have $\Delta\lambda=2$ as the possible outcomes are $\pm1$. Then, by inverting, we get
\begin{equation}
    L \leq 2\epsilon^{-2}\log(2/\delta)\,.
\end{equation}
As each run of the algorithm requires $2\alpha$ copies of $\rho$, we need in total $O(\alpha \epsilon^{-2}\log(2/\delta))$ copies. The post-processing is linear in $n$, thus we have $O(\alpha n\epsilon^{-2}\log(2/\delta))$ post-processing time. Finally, Bell transformations require only constant circuit depth, with $n$ CNOT gates and Hadamard gates in parallel.

\end{proof}

\begin{algorithm}[h]
 \SetAlgoLined
 \LinesNumbered
  \SetKwInOut{Input}{Input}
  \SetKwInOut{Output}{Output}
   \Input{ $n$-qubit state $\rho$; integer $\alpha$; $L$ repetitions;%
   }
    \Output{$A_\alpha(\rho)$
    }

$A_\alpha= 0$

 \SetKwRepeat{Do}{do}{while}
    \For{$k=1,\dots,L$}{

    \For{$j=1,\dots,\alpha$}{
        Prepare $\eta=U_\text{Bell}^{\otimes n}\rho^{\otimes 2} {U_\text{Bell}^{\otimes n}}^\dagger$
    
        Sample in computational basis $\boldsymbol{r}^{(j)}\sim  \bra{r}\eta\ket{r}$
    }
    $b= 1$

        \For{$\ell=1,\dots,n$}{
            
            $\nu_1=\bigoplus_{j=1}^\alpha r^{(j)}_{2\ell-1}$;             $\nu_2=\bigoplus_{j=1}^\alpha r^{(j)}_{2\ell}$
            
             \uIf{$\alpha$ $\mathrm{is\,\,odd}$}{
               $b= b\cdot (-2\nu_1 \cdot\nu_2+1)$
              }
              \Else{
                $b= b\cdot 2(\nu_1-1) \cdot(\nu_2-1)$
              }
        }
        
    $A_\alpha= A_\alpha+b/L$
    }
 \caption{Measure $A_\alpha$ using Bell measurements.}
 \label{alg:SEmeas}
\end{algorithm}

\section{Efficient testing of magic for mixed states}\label{sec:testing}
Here, we give the detailed on proof on testing magic for mixed states with entropy $S_2=O(\log n)$. 

\begin{theorem}[Efficient testing of magic]\label{thm:testing_sup}
Let $\rho$ be an $n$-qubit state with $S_2(\rho)=O(\log n)$ where it is promised that
\begin{align*}
\mathrm{either}\quad (a)& \,\,\mathrm{LR}(\rho)=O(\log n) \,\,\mathrm{and}\,\,D_\mathrm{F}(\rho)=O(\log n) \,,\\
\mathrm{or}\quad (b)& \,\, \mathrm{LR}(\rho)=\omega(\log n)\,\,\mathrm{and}\,\, D_\mathrm{F}(\rho)=\omega(\log n)\,.
\end{align*} 
Then, there exist an efficient quantum algorithm to distinguish case ($a$) and ($b$) using $\mathrm{poly}(n)$ copies of $\rho$ with high probability. 
\end{theorem}
\begin{proof}
First, let us consider case ($a$), where we have $\text{LR}(\rho)=O(\log n)$. We recall that $\mathcal{W}_3(\rho)\leq 2\text{LR}(\rho)$. From this, it immediately follows that $\mathcal{W}_3(\rho)=O (\log n)$. As we demand that $S_2(\rho)=O(\log n)$, we have $-\frac{1}{2}\ln A_3(\rho)=O(\log n)$.

Next, let us consider case ($b$): Here, we have that $D_\text{F}(\rho)=\omega(\log n)$. 
We now give a lower bound on $-\frac{1}{2}\ln A_3(\rho)$ via $D_\text{F}$. In particular, we have
\begin{equation}
    D_\text{F}(\rho)=\min_{\sigma\in \text{STAB}}-\ln \mathcal{F}(\rho,\sigma) \leq \min_{\ket{\phi}\in \text{STAB}_\text{pure}}-\ln \bra{\phi}\rho\ket{\phi}\,.
\end{equation}
Here, we relaxed the minimization over the set of mixed stabilizer states to the minimization of pure stabilizer states only.
Then, Ref.~\cite{iyer2024tolerant} (Theorem 5.14) showed that if $A_3\geq \gamma$ with some $\gamma$, then $\max_{\ket{\phi}} \bra{\phi}\rho\ket{\phi}\geq \Omega(\gamma^{1089})$. Here, the power can be significantly improved in future works.
This implies that $\max_{\ket{\phi}} \bra{\phi}\rho\ket{\phi}\geq \Omega(A_3^{1089})$, and thus
\begin{equation}
D_\text{F}(\rho)\leq \min_{\ket{\phi}\in \text{STAB}_{\text{pure}}}-\ln \bra{\phi}\rho\ket{\phi} \leq \Omega(-\log A_3(\rho))\,.
\end{equation}
Due to $D_\text{F}(\rho)=\omega(\log n)$ by assumption, we have $-\frac{1}{2}\ln A_3(\rho)=\omega(\log n)$.

Now, it remains to show that one can efficiently distinguish $-\ln A_3=O(\log n)$ from $-\ln A_3=\omega(\log n)$. There exist an efficient algorithm to estimate $A_3$ to $\epsilon$ additive precision using $O(\epsilon^{-2})$ copies~\cite{haug2023efficient}, as we summarized in \SM{}~\ref{sec:measwitness}.
We now use Hoeffding's inequality 
\begin{equation}\label{eq:Hoeffding}
    \text{Pr}(\vert \hat{A}_3 - A_3\vert\ge \epsilon)\equiv\delta\le 2\exp\left(-\frac{2\epsilon^2 L}{\Delta\lambda^2}\right)\,,
\end{equation}
where $\hat{A}_3$ is the estimator for $A_3$ using $L$ measurements, $\delta$ is the failure probability, i.e. that the estimator $\hat{A}_3$ gives an error larger than $\epsilon$, and $\Delta\lambda$ the range of possible measurement outcomes. For the algorithm to estimate $A_3$~\cite{haug2023efficient}, the possible measurement outcomes are $\pm1$, and thus $\Delta\lambda=2$.  

Now, we need to distinguish the low-magic case $-\ln A_3^{\text{low}}=O(\log n)$ (e.g. $A_3=n^{-c}$ with constant $c>0$) from the high magic case $-\ln A_3^{\text{high}}=\omega(\log n)$ (e.g. $A_3^{\text{high}}=2^{-\ln^2(n)}$) while allowing only a small failure probability $\delta$.
Let us choose the decision boundary as $\epsilon_0=\frac{1}{2}n^{-c}$, i.e. when we measure $\hat{A}_3(\rho)>\epsilon_0$ then we say the state has low magic, else if we measure $\hat{A}_3(\rho)<\epsilon_0$ we say that the state has high magic. Further, let us choose the number of measurements as $L=n^{2c+1}=\text{poly}(n)$.

Then, the probability to wrongly classify a low-magic state as high magic is given by
\begin{gather}
    \delta^{\text{low}}\equiv\text{Pr}(\hat{A}_3^{\text{low}}\leq A_3^{\text{low}}-\epsilon ) =\text{Pr}(\hat{A}_3^{\text{low}}\leq \epsilon_0 )\\
    =2\exp\left(-\frac{(A_3^{\text{low}}-\epsilon_0)^2 L}{2}\right)=2\exp\left(-\frac{(n^{-c}-\frac{1}{2}n^{-c})^2 n^{2c+1}}{2}\right)=2\exp\left(-\frac{n}{8}\right)=O(2^{-n})\,.
\end{gather}
Next, the probability to wrongly classify a high-magic state as low-magic is given by
\begin{gather}
    \delta^{\text{high}}\equiv\text{Pr}(\hat{A}_3^{\text{high}}\geq A_3^{\text{high}}+\epsilon ) =\text{Pr}(\hat{A}_3^{\text{high}}\geq \epsilon_0 )\\
    =2\exp\left(-\frac{(A_3^{\text{high}}-\epsilon_0)^2 L}{2}\right)=2\exp\left(-\frac{(2^{-\ln^2(n)}-\frac{1}{2}n^{-c})^2 n^{2c+1}}{2}\right)=O(2^{-n})\,.
\end{gather}
Thus, we can distinguish low and high magic states with exponentially low failure rate using a polynomial number of measurements.
\end{proof}

\section{Certifying noisy magic resource states}\label{sec:cliffUnital}
Now, we show that using our methods we can efficiently certify magic resource states under noise. In particular, we want to certify whether a given noisy state contains a low or high number of T-states $\ket{T}=(\ket{0}+e^{-i\pi/4}\ket{1})/\sqrt{2}$. T-states are used to generate magic T-gates for universal fault-tolerant quantum computers~\cite{bravyi2005universal}. 

Such certification task of T-states are relevant for fault-tolerant quantum computing.  To run universal algorithms, one requires many copies of single-qubit T-states which can be used to generate T-gates. 
For example, let us assume that we are given a magic resource state from an untrusted source.  They claim that the given state is a tensor product of many single-qubit T-states. 
However, the T-states were affected by noise, where the noise model is not known. As full tomography is prohibitively expensive, one would like to have a simple check of the contained magic states.  In particular, can one verify whether the untrusted source indeed provided many T-states, or actually just mostly provided non-magical states?

In particular, we consider $n$-qubit states with $t$ T-states where a noise channel $\Lambda_\text{C}(\rho)$ is applied to
\begin{equation}
    \rho_t= \Lambda_\text{C}(\ket{T}^{\otimes t} \ket{0}^{\otimes(n-t)}\bra{T}^{\otimes t}\otimes \bra{0}^{\otimes(n-t)})
\end{equation}
with T-state $\ket{T}=\frac{1}{\sqrt{2}}(\ket{0}+e^{-i\pi/4}\ket{1})$ and stabilizer state $\ket{0}$.
As noise channel, we consider mixed unital Clifford channels, which are given by
\begin{equation}
    \Lambda_\text{C}(\rho)=\sum_{i} p_i U_\text{C}^{(i)} \rho {U_\text{C}^{(i)}}^\dagger
\end{equation}
where $\sum_i p_i=1$, $p_i\geq0$, and $U_\text{C}^{(i)}$ are Clifford unitaries, i.e. unitaries that can be generated from CNOT gates, S-gates and Hadamard gates only. For example, as special case, $\Lambda_\text{C}(\rho)$ contains Pauli channels. Note that mixed unital Clifford channels are Clifford channels, and thus cannot increase magic.

Now, we ask whether one can distinguish states with different $t$. Here, we assume low entropy $S_2(\rho_t)=-\ln\text{tr}(\rho_t^2)=O(\log n)$. This is necessary as the magic of states with high entropy cannot be tested in general~\cite{bansal2024pseudorandomdensitymatrices}.

Now, using our test we can efficiently distinguish $\rho_t$ between small and large $t$ when $S_2(\rho_t)=O(\log n)$ for any mixed unital Clifford channel:

\begin{proposition}[Testing noisy magic states]
There is an efficient quantum algorithm to certify whether given noisy $n$-qubit quantum state $\rho_t=\Lambda_\text{C}((\ket{T}\bra{T})^{\otimes t}\otimes (\ket{0} \bra{0})^{\otimes(n-t)})$ contains
either $(a)$ $t=O(\log n)$ or $(b)$ $t=\omega(\log n)$ T-states $\ket{T}=\frac{1}{\sqrt{2}}(\ket{0}+e^{-i\pi/4}\ket{1})$ when it is subject to arbitrary $n$-qubit mixed unital Clifford channels $\Lambda_\text{C}(\rho)$ and $S_2(\rho_t)=O(\log n)$. 
\end{proposition}
\begin{proof}
We can efficiently distinguish $(a)$ and $(b)$ using our tester by measuring $A_3(\rho_t)$. 

For $(a)$ with $t=O(\log n)$, we have $\text{LR}=O(\log n)$. This follows from the fact that $\ket{\psi_t}=\ket{T}^{t}\otimes \ket{0}^{n-t}\bra{T}^{t}\otimes \bra{0}^{n-t}$ can be prepared from $t$ T-gates, and thus has $\text{LR}(\ket{\psi_t})=O(t)=O(\log n)$ due to sub-additivity~\cite{howard2017robustness}. As $\text{LR}$ is a magic monotone, it can only decrease under Clifford channels $\text{LR}(\Lambda_\text{C}(\ket{\psi_t}))\leq \text{LR}(\ket{\psi_t})$. Finally, we have 
\begin{equation}
    \mathcal{W}_3=-\frac{1}{2}\ln A_3-\frac{5}{2}S_2\leq 2 \text{LR}
\end{equation}
and  $S_2(\rho_t)=O(\log n)$ which implies $-\frac{1}{2}\ln A_3=O(\log n)$. 

Next, we consider case $(b)$ with large number of T-states where $t=\omega(\log n)$. 
For $\ket{\psi_t}$, we get $-\frac{1}{2}\ln A_3(\ket{\psi_t})=\frac{1}{2}t\ln\frac{8}{5}=\omega(\log n)$ from explicit calculation~\cite{haug2023stabilizer} when $t=\omega(\log n)$.

Now, under mixed unital Clifford channel $\Lambda_\text{C}(\rho)$, we find that $-\frac{1}{2}\ln A_3(\Lambda_\text{C}(\ket{\psi_t}))$ can only increase: %
\begin{lemma}\label{lem:unitalcliff}
Applying arbitrary mixed unital Clifford channel  $\Lambda_\text{C}(\rho)$ onto state $\rho$ monotonously increases the $\alpha$-moment of Pauli spectrum for any integer $\alpha>1$, i.e.  
\begin{equation}
    -\ln A_\alpha(\Lambda_\text{C}(\rho))\geq -\ln A_\alpha(\rho)\,.
\end{equation}
\end{lemma} 
\begin{proof}

To show this, we have
\begin{gather}
    A_\alpha(\Lambda_\text{C}(\rho))=
    2^{-n}\sum_{P\in\mathcal{P}_n} \text{tr}(\Lambda_\text{C}(\rho) P)^{2\alpha}=
    2^{-n}\sum_{P\in\mathcal{P}_n} \text{tr}(\sum_i p_i \rho_i P)^{2\alpha}
\end{gather}
where we introduced the Clifford transformed state $\rho_i=U_\text{C}^{(i)}\rho {U_\text{C}^{(i)}}^\dagger$. Next, we expand the power as
\begin{gather}
    A_\alpha(\Lambda_\text{C}(\rho))
    =2^{-n} \sum_{i_1,i_2,\dots,i_{2\alpha}} \sum_{P\in\mathcal{P}_n}\text{tr}(p_{i_1} \rho_{i_1} P)\text{tr}(p_{i_2} \rho_{i_2} P)\dots \text{tr}(p_{i_{2\alpha}} \rho_{i_{2\alpha}} P)\\
    \leq 2^{-n} \sum_{i_1,i_2,\dots,i_{2\alpha}} p_{i_1}p_{i_2}\dots p_{i_{2\alpha}}\sum_{P\in\mathcal{P}_n}\vert\text{tr}( \rho_{i_1} P)\vert\vert\text{tr}( \rho_{i_2} P)\vert\dots \vert\text{tr}( \rho_{i_{2\alpha}} P)\vert\,.
\end{gather}
We now bound this using from H\"olders inequality~\cite{chen2014brief}
\begin{equation}
    \sum_j a_{1,j}^{\lambda_1}a_{2,j}^{\lambda_2}\dots a_{K,j}^{\lambda_K} \leq (\sum_j a_{1,j})^{\lambda_1}(\sum_j a_{2,j})^{\lambda_2}\dots (\sum_j a_{K,j})^{\lambda_{K}}
\end{equation}
with $\sum_j \lambda_j=1$, $\lambda_j\geq0$ and $a_{k,j}\geq0$. 
Now, we identify $j=P$ $\lambda_P=1/(2\alpha)$, $K=2\alpha$, and $a_{i_s,P}=\vert\text{tr}(\rho_{i_s} P)\vert^{2\alpha}$ with $s=1,\dots,2\alpha$.
Then, we apply H\"olders inequality to get
\begin{gather}
 A_\alpha(\Lambda_\text{C}(\rho)) 
 \leq 2^{-n} \sum_{i_1,i_2,\dots,i_{2\alpha}} p_{i_1} p_{i_2}\dots p_{i_{2\alpha}}\sum_{P\in\mathcal{P}_n}\vert\text{tr}( \rho_{i_1} P)\vert\vert\text{tr}( \rho_{i_2} P)\vert\dots \vert\text{tr}( \rho_{i_{2\alpha}} P)\vert \\
 \leq  2^{-n} \sum_{i_1,i_2,\dots,i_{2\alpha}} p_{i_1}p_{i_2}\dots p_{i_{2\alpha}}(\sum_{P\in\mathcal{P}_n} \vert\text{tr}( \rho_{i_1} P)\vert^{2\alpha})^{\frac{1}{2\alpha}}(\sum_{P\in\mathcal{P}_n} \vert\text{tr}( \rho_{i_2} P)\vert^{2\alpha})^{\frac{1}{2\alpha}}\dots (\sum_{P\in\mathcal{P}_n} \vert\text{tr}( \rho_{i_{2\alpha}} P)\vert^{2\alpha})^{\frac{1}{2\alpha}}\\
= \sum_{i_1,i_2,\dots,i_{2\alpha}} p_{i_1} p_{i_2}\dots p_{i_{2\alpha}} A_\alpha(\rho_{i_1})^{\frac{1}{2\alpha}}
 A_\alpha(\rho_{i_2})^{\frac{1}{2\alpha}}\dots
 A_\alpha(\rho_{2\alpha})^{\frac{1}{2\alpha}}\,.
\end{gather}
Now, we note that $A_\alpha(\rho_i)=A_\alpha(\rho)$ as Clifford unitaries leave $A_\alpha(\rho_i)$ invariant. 
Thus, we finally get
\begin{equation}
    A_\alpha(\Lambda_\text{C}(\rho)) \leq \sum_{i_1,i_2,\dots,i_{2\alpha}} p_{i_1} p_{i_2} \dots p_{i_{2\alpha}}A_\alpha(\rho) = A_\alpha(\rho)\,.
\end{equation}
where we recall $\sum_{i_s} p_{i_s}=1$.
\end{proof}

Thus, according to Lemma~\ref{lem:unitalcliff}, $-\ln A_3$ can only increase under mixed unital Clifford channels, and we have for $t=\omega(\log n)$
\begin{equation}
 -\frac{1}{2}\ln A_3(\Lambda_\text{C}(\ket{\psi_t})) \geq -\ln A_\alpha(\ket{\psi_t})=\omega(\log n)\,.
\end{equation}
Now, using Thm.~\ref{thm:witness_sup}, we can efficiently distinguish $(a)$ with $-\ln A_3=O(\log n)$ and $(b)$ with $-\ln A_3=\omega(\log n)$, which can be shown via Hoeffding's inequality similar to the proof of \SM{}~\ref{sec:testing}.

\end{proof}

\section{Witnessing magic in typical states under depolarizing noise} \label{sec:witness_noise}

In this section, we study our magic witness $\mathcal{W}_\alpha$ for typical states subject to global
depolarizing noise $\rho_\mathrm{dp}
 = (1-p)\ket{\psi}\bra{\psi} + p I/2^n$ 
 with a probability $p$. The mixed-state SRE $M_\alpha=1/(1-\alpha)(\ln A_\alpha+S_2)$ for $n\gg1$ is given by~\cite{turkeshi2023pauli}
 \begin{equation}
   M_\alpha^\mathrm{typ} = \frac{1}{1-\alpha}\ln \left[ \frac{(\eta-1) (2b)^\alpha   \Gamma \left(\alpha+\frac{1}{2}\right) (1-p)^{2\alpha}}{\sqrt{\pi }d }+\frac{1}{d}\right] + \frac{1}{1-\alpha} S_2[\rho_\mathrm{dp}],
\end{equation}
where the R\'enyi-2 entropy of the depolarized state is given by $S_2[\rho_\mathrm{dp}]= -\ln\left[(1-p)^2 + p(2-p)/2^n \right]$, $\Gamma(x)$ is the gamma function, $b=(d/2+1)^{-1}$, $d=2^n$, and  ${\eta=d^2}$. For any fixed $p<1$ and $n \to \infty$, we find that the leading term of the mixed-state SRE $M_\alpha = D_\alpha n + c_\alpha$ is unchanged from that of the pure state. Namely, $D_\alpha = 1/(\alpha-1)$ for $\alpha \geq 2$ and $D_\alpha=1$ for $\alpha<2$, while the probability $p$ only affects the subleading constant $c_\alpha$. As a consequence, the witnesses $\mathcal{W}_\alpha$ are remarkably effective in detecting magic in typical states even in the presence of noise, where the state becomes a mixed state. 
At $p=1$, the state becomes the maximally mixed state $\rho_\mathrm{dp}(p=1) = I / 2^n$, which is a mixed stabilizer state. In this case, the witnesses take the value $-2n \ln 2$, independent of the R\'enyi index $\alpha$. The negative value is expected, since the state is a mixed stabilizer state, so that the witnesses would not classify the state as nonstabilizer state.

Let us be more concrete when $\mathcal{W}_\alpha>0$, i.e. detects magic.
In particular, we consider exponential noise $p=1-2^{-\beta n}$ and determine for which $\beta$ the witness $\mathcal{W}_\alpha$ can successfully detect magic. 
We now compute the explicit scaling of $\mathcal{W}_\alpha$ with $n\gg1$, where we ignore constants and subleading terms.
First, when considering the entropy term, we find that $S_2(\rho_\text{dp})\approx 2\beta n\ln 2$ for $\beta\leq1/2$. 
Then, we have
\begin{equation}
\mathcal{W}_\alpha(\rho_\text{dp})\approx\frac{1}{1-\alpha}\ln(2^{n(1-\alpha)}2^{-2\alpha\beta n} +2^{-n})+\frac{1}{1-\alpha}2\beta n\ln 2 -4\beta n\ln2\,.
\end{equation}
We now ask for what $\beta$ we witness magic $\mathcal{W}_\alpha>0$.
For $\alpha>1$, we can drop the $2^{n(1-\alpha)}2^{-2\alpha\beta n}$ term in the logarithm as it is larger than $2^{-n}$. Then, we find $\beta<(4\alpha-2)^{-1}$. 
For $\frac{1}{2}\leq\alpha<1$, we can instead drop the $2^{-n}$ in the first logarithm as it decays faster than the other term. Then, a straightforward calculation yields $\beta<\frac{1}{2}$.

Next, we consider  the filtered witness $\Tilde{\mathcal{W}}_\alpha$ of~\eqref{eq:filtered_witness}, where we find
\begin{equation}    \Tilde{\mathcal{W}}_\alpha(\rho_{\text{dp}}) = \Tilde{M}_\alpha(\ket{\psi}) + 2\ln(1-p) = \Tilde{M}_\alpha(\ket{\psi}) - 2\beta n\ln2,
\end{equation}
where we recall that $\Tilde{M}_\alpha$ is the filtered SRE in~\eqref{eq:filtered_sre}. For Haar random states, $\Tilde{M}_\alpha=n \ln 2$ for any $\alpha$~\cite{turkeshi2023pauli}. Thus, we find $\Tilde{\mathcal{W}}_\alpha>1$ when $\beta<\frac{1}{2}$ for all $\alpha\geq1/2$. Remarkably, $\Tilde{\mathcal{W}}_\alpha$ for larger $\alpha$ can detect magic as well as $\Tilde{\mathcal{W}}_{1/2}$, while being efficiently measurable for odd integer $\alpha$~\cite{haug2023efficient}. Note that this is true for any state with flat Pauli spectrum subjected to depolarizing noise, since $\Tilde{M}_\alpha$ is independent of $\alpha$. Indeed, typical states were shown to have a flat Pauli spectrum in~\cite{tarabunga2025efficientmutualmagicmagic}.

\section{Witnessing magic in random circuits under local depolarizing noise} \label{sec:witness_noise_local}

\begin{figure}[htbp]
	\centering	
	\subfigimg[width=0.4\textwidth]{}{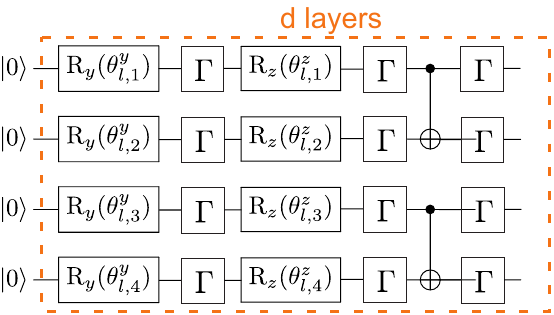}
    \caption{Hardware efficient circuit with local depolarising noise $\Gamma$ of $d$ layers. Per layer, we have random $x$ and $z$ rotations, and CNOT gates arrange in a one-dimensional configuration. Note that the layer of CNOT gates are shifted vertically by one qubit after each layer to achieve global entangling. 
	}
	\label{fig:depolcircuit}
\end{figure}

\begin{figure}[htbp]
	\centering	
	\subfigimg[width=0.24\textwidth]{a}{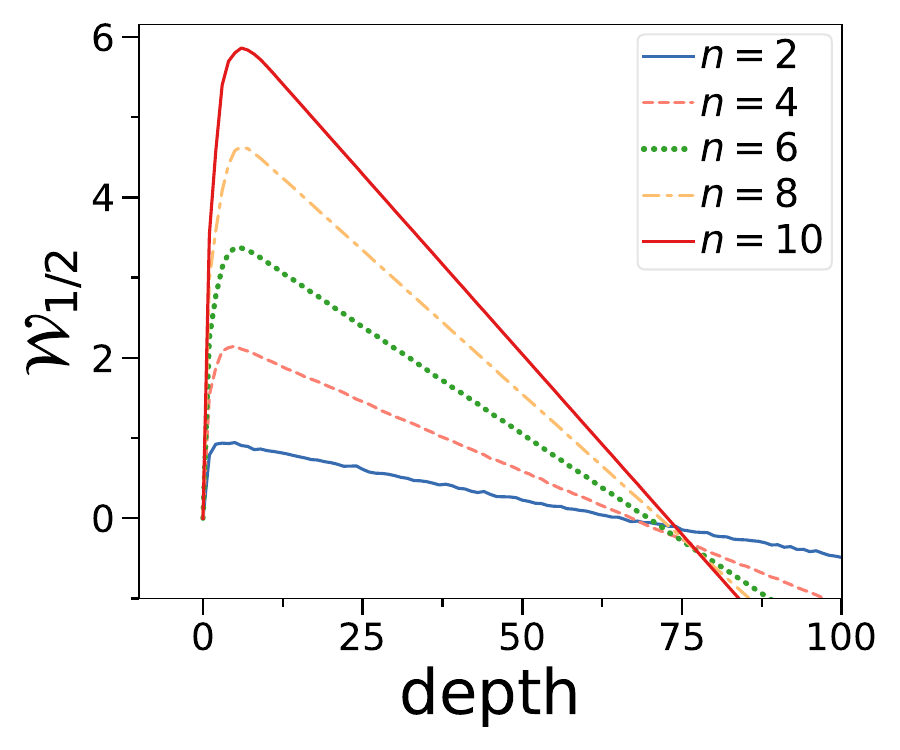}
 	\subfigimg[width=0.24\textwidth]{b}{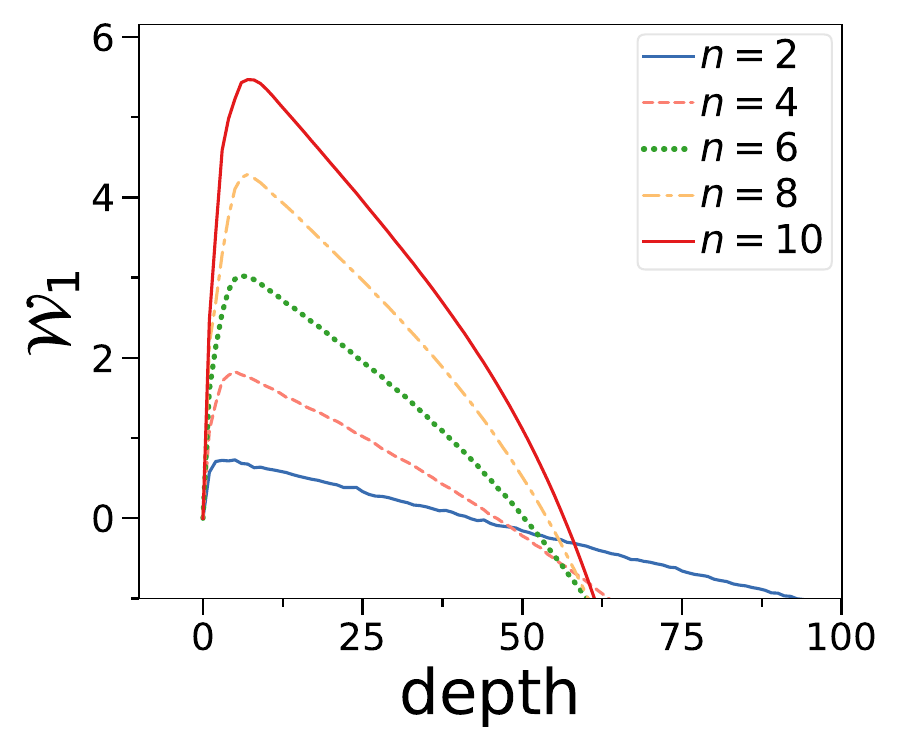}
\subfigimg[width=0.24\textwidth]{c}{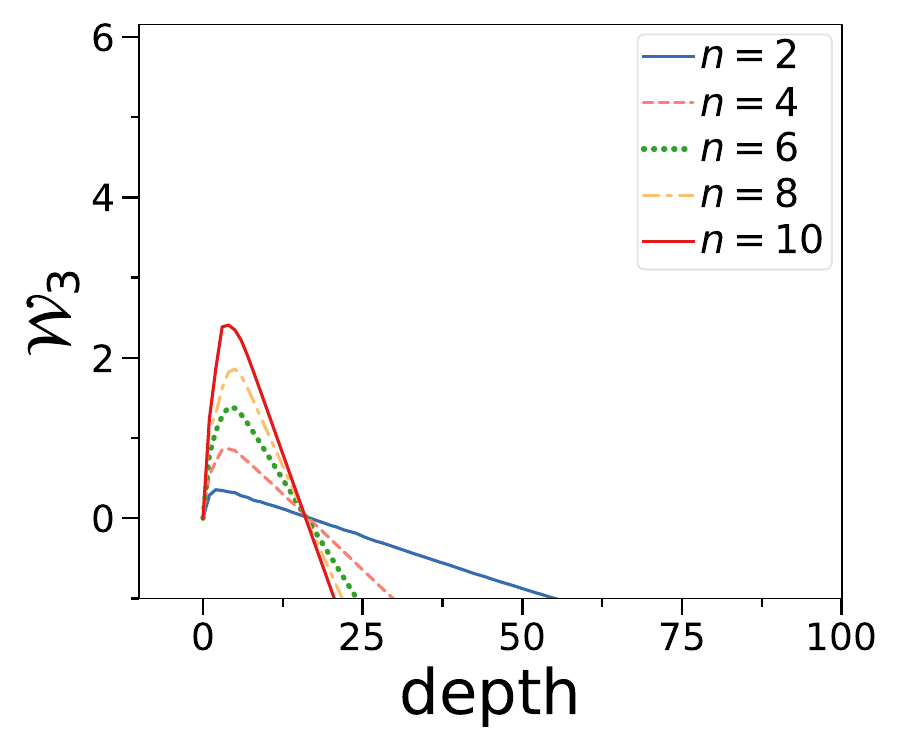}
\subfigimg[width=0.24\textwidth]{d}{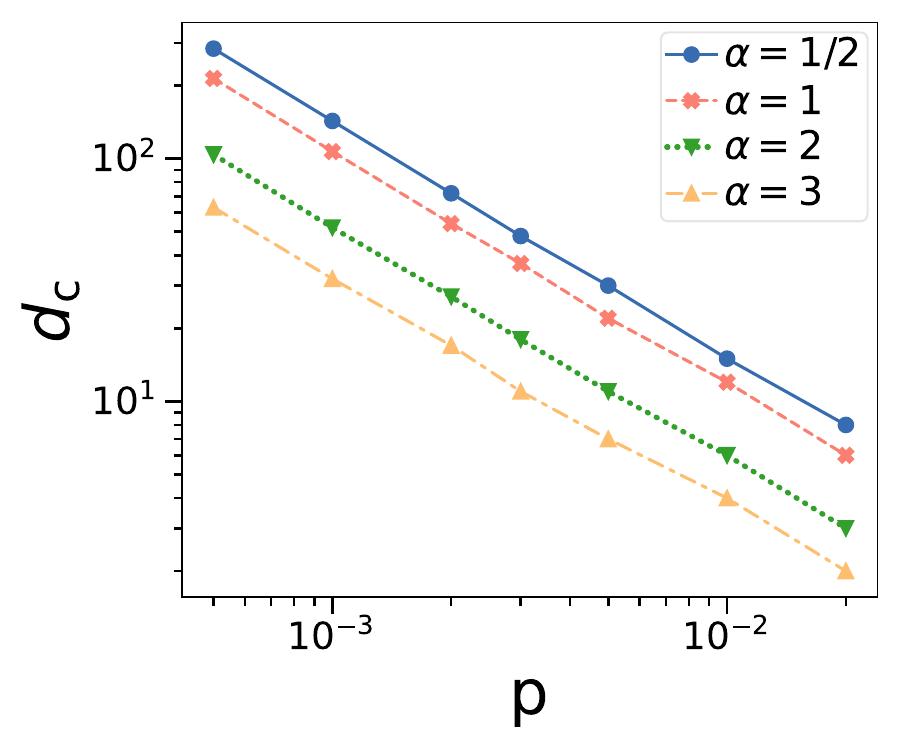}
    \caption{Magic witness $\mathcal{W}_\alpha$ against random circuit of depth $d$ for different qubit number $n$, where we choose local depolarising strength $p=0.002$ (see Fig.~\ref{fig:depolcircuit}). We show \idg{a} $\mathcal{W}_{1/2}$, \idg{b} $\mathcal{W}_{1}$, \idg{c} $\mathcal{W}_{3}$. \idg{d} Critical circuit depth $\tilde{d}_\text{c}$ where $\Tilde{\mathcal{W}}_\alpha$ becomes negative for different $p$ for $n=8$. By fitting we find $d_\text{c}\propto p^{-\eta}$ with $\eta\approx 0.96$. Each datapoint is averaged over 10 random circuit instances.
	}
	\label{fig:witnessLocalDepol}
\end{figure}

\begin{figure}[htbp]
	\centering	
	\subfigimg[width=0.24\textwidth]{a}{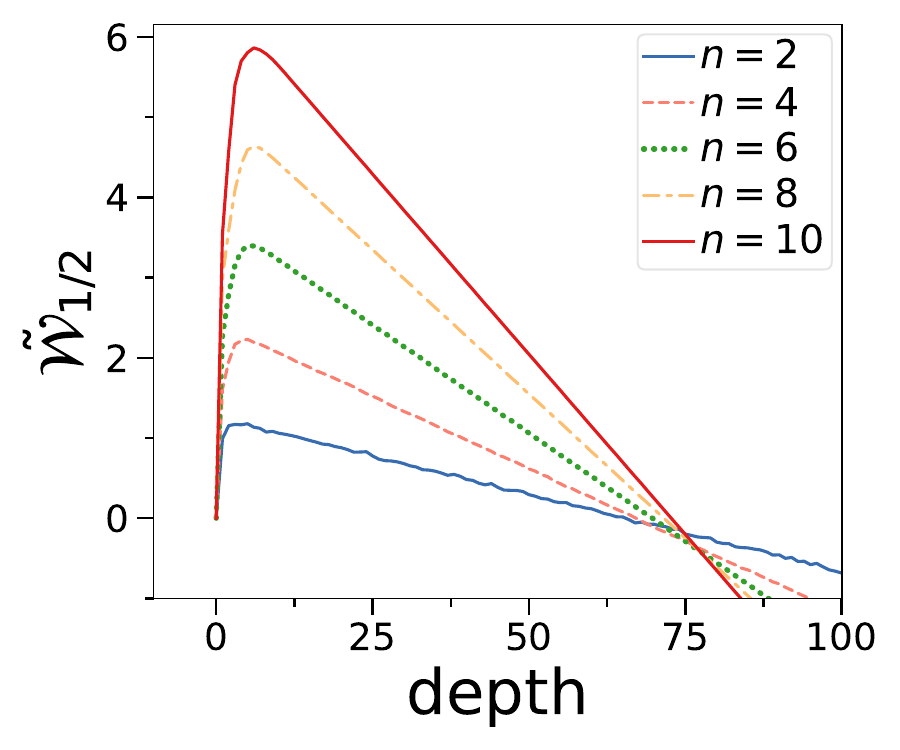}
 	\subfigimg[width=0.24\textwidth]{b}{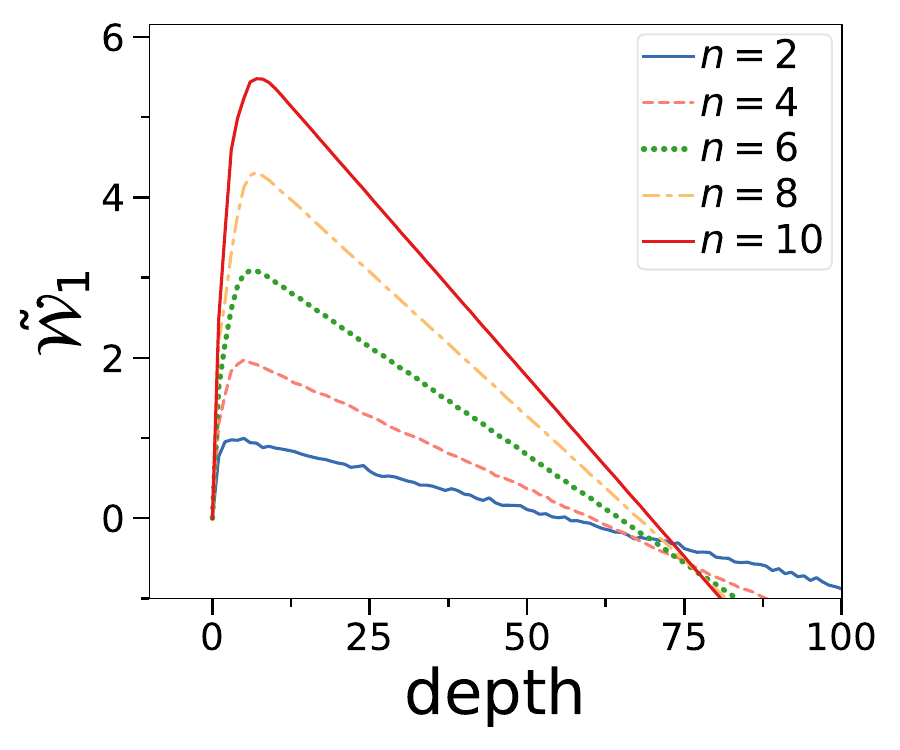}
\subfigimg[width=0.24\textwidth]{c}{filt_witness3WTN10d100r1c0s0_002.pdf}
    \subfigimg[width=0.24\textwidth]{d}{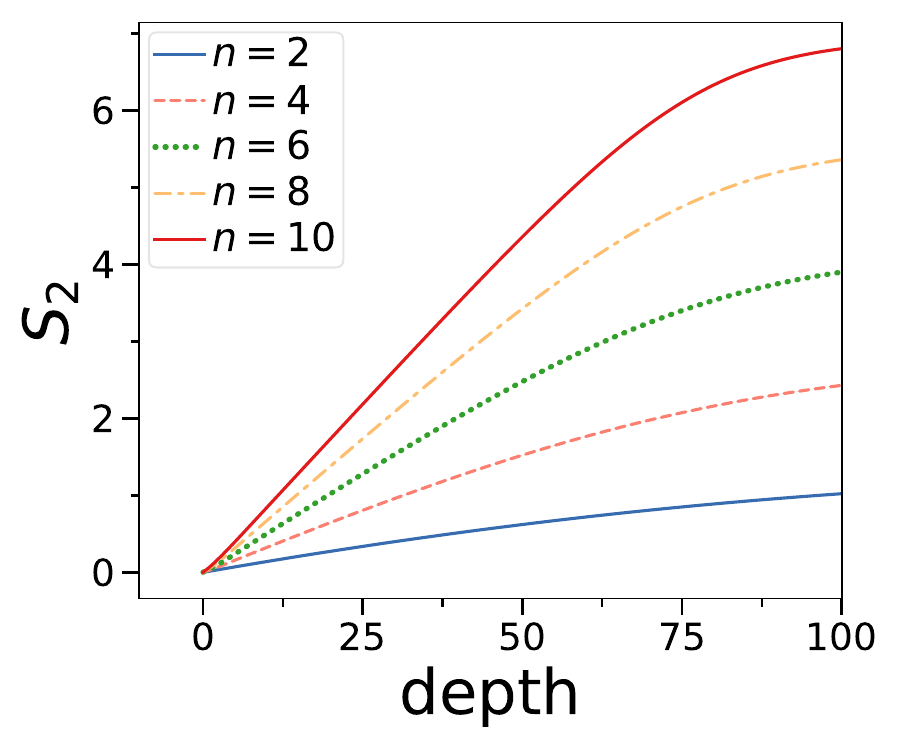}
    \caption{Filtered magic witness $\Tilde{\mathcal{W}}_\alpha$ against random circuit of depth $d$ for different qubit number $n$, where we choose local depolarising strength $p=0.002$ (see Fig.~\ref{fig:depolcircuit}). We show \idg{a} $\Tilde{\mathcal{W}}_{1/2}$, \idg{b} $\Tilde{\mathcal{W}}_{1}$, \idg{c} $\Tilde{\mathcal{W}}_{3}$, and \idg{d} $2$-R\'enyi entropy  $S_2$.
	}
	\label{fig:filteredwitnessLocalDepol}
\end{figure}

\begin{figure}[htbp]
	\centering	
	\subfigimg[width=0.24\textwidth]{a}{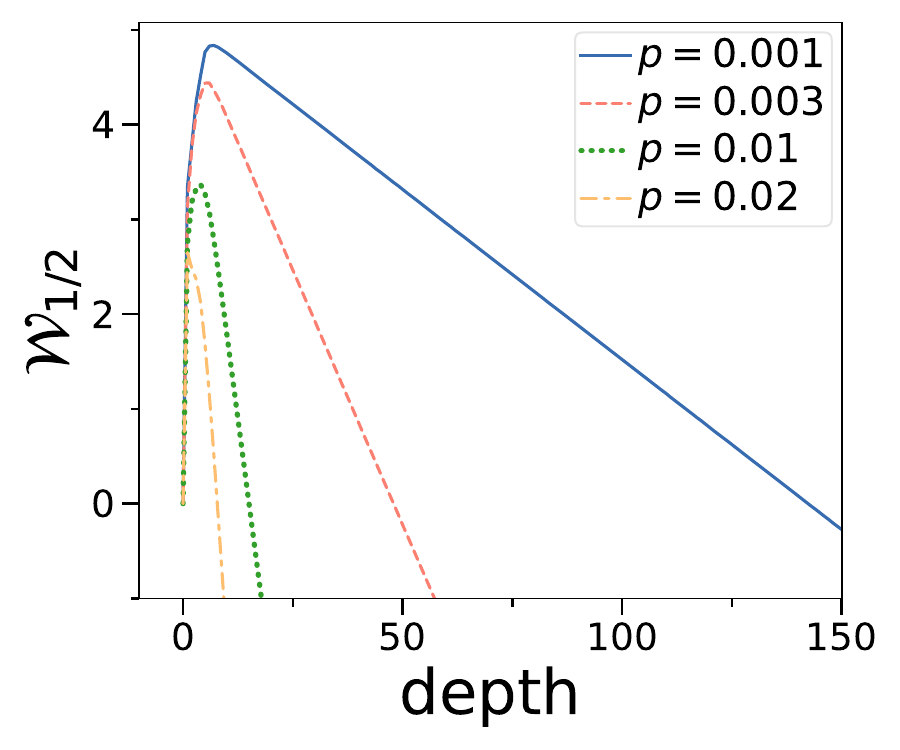}
 	\subfigimg[width=0.24\textwidth]{b}{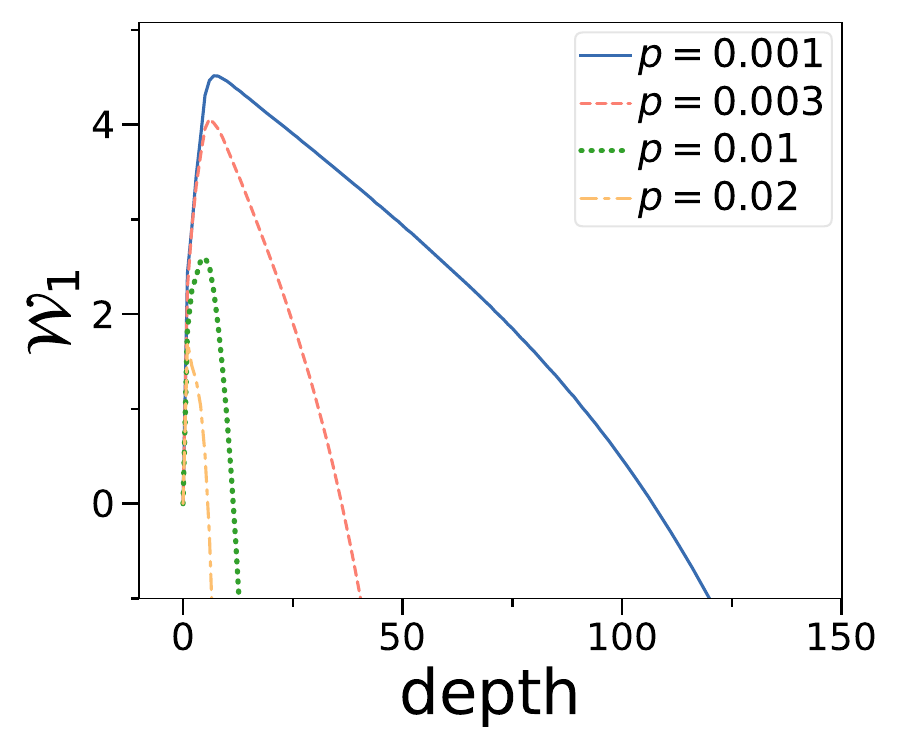}
\subfigimg[width=0.24\textwidth]{c}{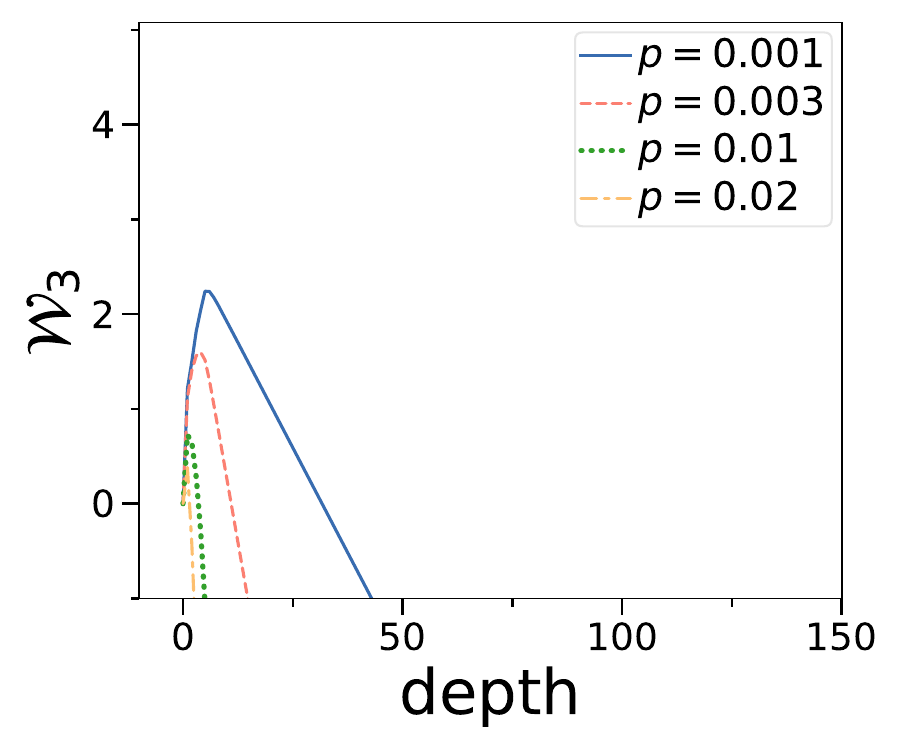}
    \subfigimg[width=0.24\textwidth]{d}{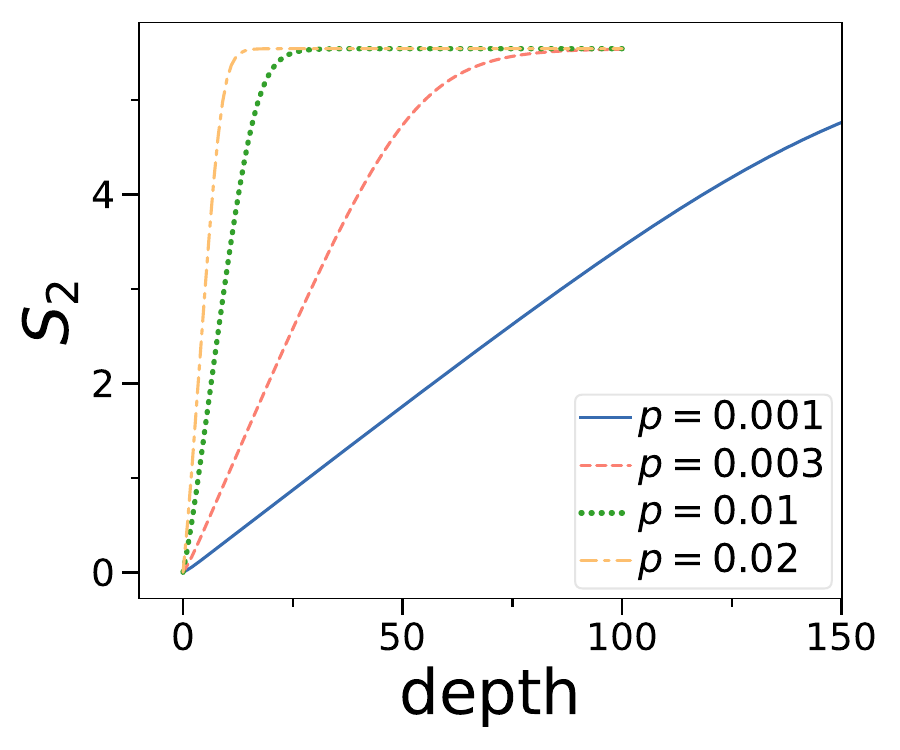}
    \caption{Magic witness $\mathcal{W}_\alpha$ against random circuit of depth $d$ for different local depolarising strength $p$  and $n=8$ (see Fig.~\ref{fig:depolcircuit}). We show \idg{a} $\mathcal{W}_{1/2}$, \idg{b} $\mathcal{W}_{1}$, \idg{c} $\mathcal{W}_{3}$, and \idg{d} $2$-R\'enyi entropy  $S_2$. 
	}
	\label{fig:witnessLocalDepolp}
\end{figure}

\begin{figure}[htbp]
	\centering	
	\subfigimg[width=0.24\textwidth]{a}{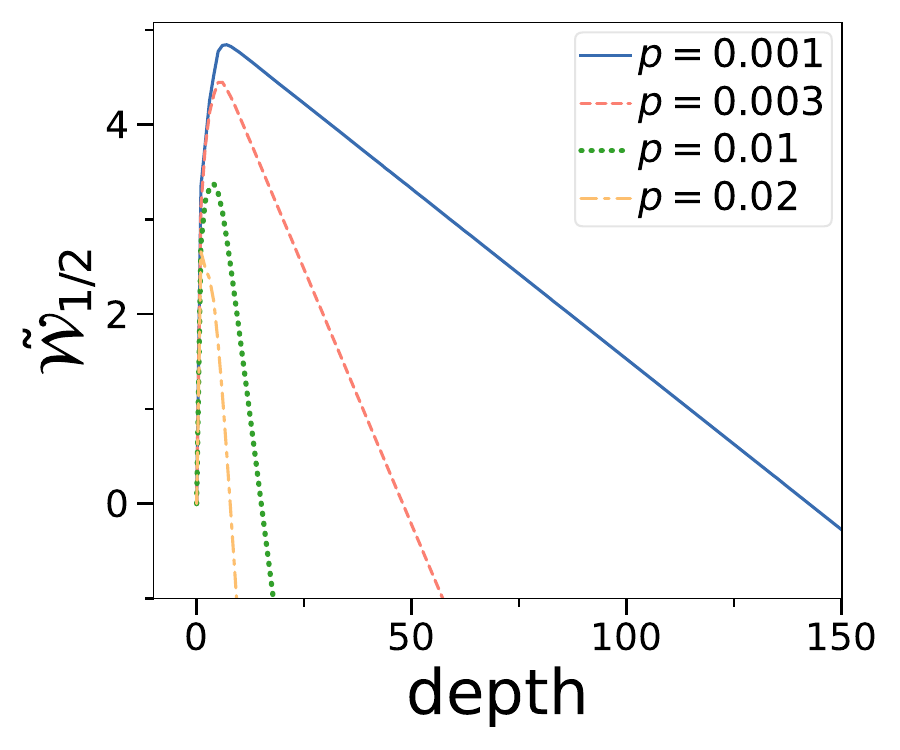}
 	\subfigimg[width=0.24\textwidth]{b}{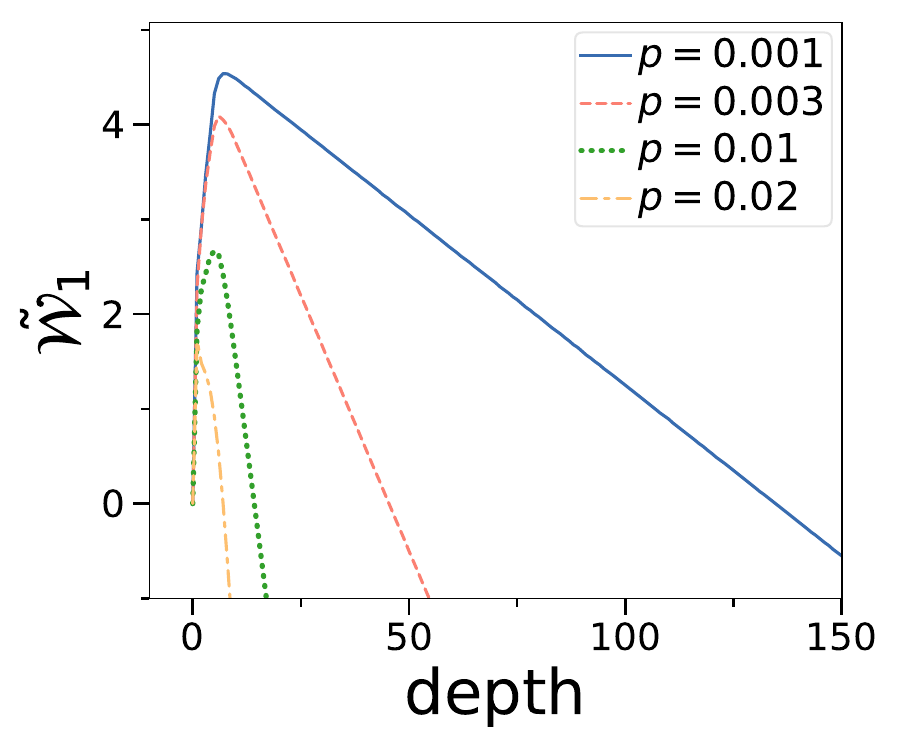}
\subfigimg[width=0.24\textwidth]{c}{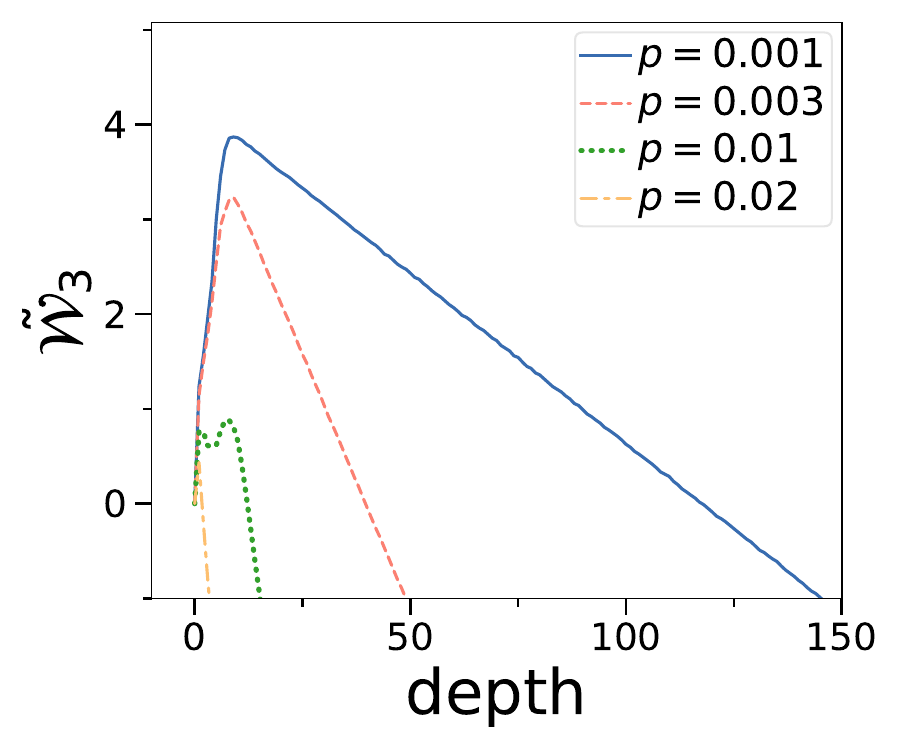}
    \caption{Filtered magic witness $\Tilde{\mathcal{W}}_\alpha$ against random circuit of depth $d$ for different local depolarising strength $p$  and $n=8$ (see Fig.~\ref{fig:depolcircuit}). We show \idg{a} $\Tilde{\mathcal{W}}_{1/2}$, \idg{b} $\Tilde{\mathcal{W}}_{1}$, \idg{c} $\Tilde{\mathcal{W}}_{3}$. 
	}
	\label{fig:filteredwitnessLocalDepolp}
\end{figure}

We now consider a model to study magic that can be directly implemented in experiment on quantum computers with realistic noise model.
We consider a circuit shown in Fig.~\ref{fig:depolcircuit}, consisting of $d$ layers of single qubit $y$ rotations $R_y(\theta)=\exp(-i\theta\sigma^y)$, $z$ rotations $R_y(\theta)=\exp(-i\theta\sigma^z)$, followed by CNOT gates arranged in a nearest-neighbour chain. 
We choose the rotation angles completely at random, i.e. $\theta\in(0,2\pi]$.
Then, after every gate we apply single-qubit depolarizing noise $\Gamma(\rho)=(1-p)\rho+p \text{tr}_1(\rho) I_1/2$ with noise strength $p$. This model is typical for NISQ computers~\cite{bharti2022noisy}.

First, we study our witness $\mathcal{W}_\alpha$ in Fig.~\ref{fig:witnessLocalDepol} for different $\alpha$ and qubit number $n$ for this noisy random local circuit model. We find that the witness increases with depth $d$, then decreases, becoming negative for large $d$. We find that the depth where the witness peaks is nearly independent of $n$. We find that for large depth $d$, the witness becomes negative. Notably, we find that all curves cross at the same depth when $\mathcal{W}_\alpha=0$. The critical point $d_\text{c}$ depends on $\alpha$ and $p$, but not on $n$. 

We also show the filtered witness $\Tilde{\mathcal{W}}_\alpha$ in Fig.~\ref{fig:filteredwitnessLocalDepol}. In contrast to $\mathcal{W}_\alpha$, the filtered witness is nearly independent of  $\alpha$, showing for all $\alpha$ nearly the same critical depth $\tilde{d}_\text{c}$ where the witness becomes negative.
We also note that for small $p$ we can witness magic even for quite high entropy $S_2=\text{tr}(\rho^2)$ as shown in Fig.~\ref{fig:witnessLocalDepol}d.

Next, we study our witness $\mathcal{W}_\alpha$ in Fig.~\ref{fig:witnessLocalDepolp} for different $\alpha$ and $p$, where we fix $n=8$. Again, we find that the witness increases with depth $d$, then decreases, becoming negative for large $d$.  For lower $p$, the crossing point shifts to larger $d$. 

We study the filtered witness $\Tilde{\mathcal{W}}_\alpha$ in Fig.~\ref{fig:filteredwitnessLocalDepolp}. The filtered witness behaves similar for all $\alpha$. In particular, the critical depth barely changes with $\alpha$. This contrasts $\mathcal{W}_\alpha$, which is less sensitive to detect magic for larger $\alpha$.

\section{Additional experiments}\label{sec:experiment_sup}
Here, we show an additional experiment on witnessing magic on the IonQ quantum computer.

First, in Fig~\ref{fig:Ionq_filtered}, we consider the experiment in the main text, but using the filtered witness $\Tilde{\mathcal{W}}_\alpha$, which similarly can be efficiently computed for odd $\alpha>1$.
We find similar behavior, but a smaller gap between the log-free robustness of magic and $\Tilde{\mathcal{W}}_\alpha$, indicating that the filtered witness $\Tilde{\mathcal{W}}_\alpha$ is more sensitive than the standard one $\mathcal{W}_\alpha$.
\begin{figure}[htbp]
	\centering	
	\subfigimg[width=0.28\textwidth]{a}{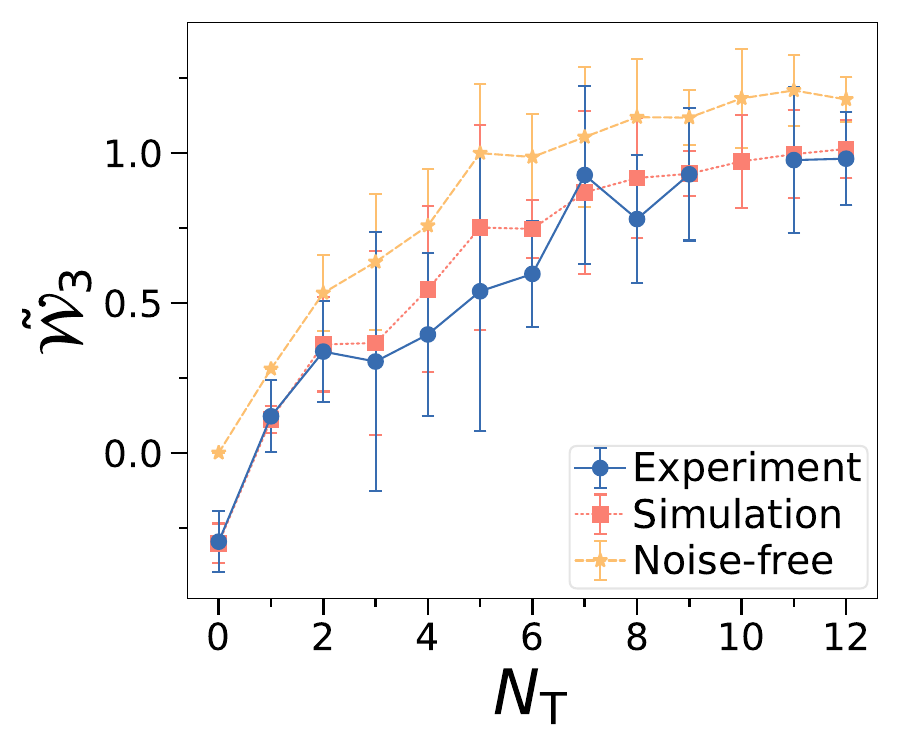}
    \subfigimg[width=0.28\textwidth]{b}{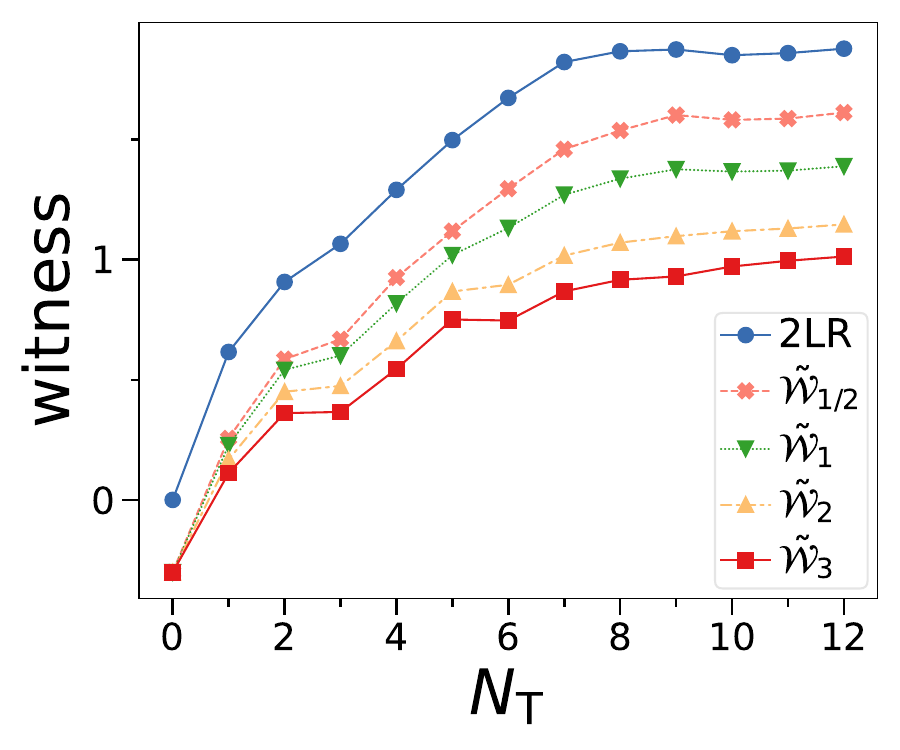}
    \caption{\idg{a} Filtered magic witness $\Tilde{\mathcal{W}}_3$ measured on the IonQ quantum computer for Clifford circuits doped with $N_\text{T}$ T-gates. We have $n=3$ qubits and estimate an effective average global depolarisation noise $p\approx 0.2$ from the purity. We show experimental result in blue and simulation with global depolarisation noise in orange, while yellow line is a simulation for noise-free states. We average over 10 random initialization of the circuit. 
    \idg{b} Simulation of $\Tilde{\mathcal{W}}_\alpha$ for different $\alpha$ and log-free robustness of magic $\text{LR}$ with global depolarisation noise with parameters of experiment.
	}
	\label{fig:Ionq_filtered}
\end{figure}

Next, we regard preparation of magic resource states which are usually used for magic state distillation.
In particular, we consider the single-qubit T-state 
\begin{equation}
    \ket{T}=\frac{1}{\sqrt{2}}(\ket{0}+e^{-i\pi/4}\ket{1})
\end{equation}
and the single-qubit state with maximal magic 
\begin{equation}
\ket{R}=\cos(\theta/2)\ket{0}+e^{-i\pi/4}\sin(\theta/2)\ket{1})\,,
\end{equation}
with $\theta=\arccos(1/\sqrt{3})$~\cite{bravyi2005universal}.
We also consider the stabilizer product state $\ket{+}=\frac{1}{\sqrt{2}}(\ket{0}+\ket{1})$ as reference. 
We show our witness $\mathcal{W}_3$ and filtered witness $\Tilde{\mathcal{W}}_3$ for the different product states and different qubit number $n$ in Fig.~\ref{fig:Ionq_sup}. 
We find that we can reliably produce valuable resource states with high amounts of magic, certifying that even under modest noise the mixed states indeed contain magic. Further, we correctly flag the  $\ket{+}$ as not containing any magic. 
We note that the filtered witness in Fig.~\ref{fig:Ionq_sup}b gives a stronger signal compared to the standard witness $\mathcal{W}_3$.
\begin{figure}[htbp]
	\centering	
	\subfigimg[width=0.3\textwidth]{a}{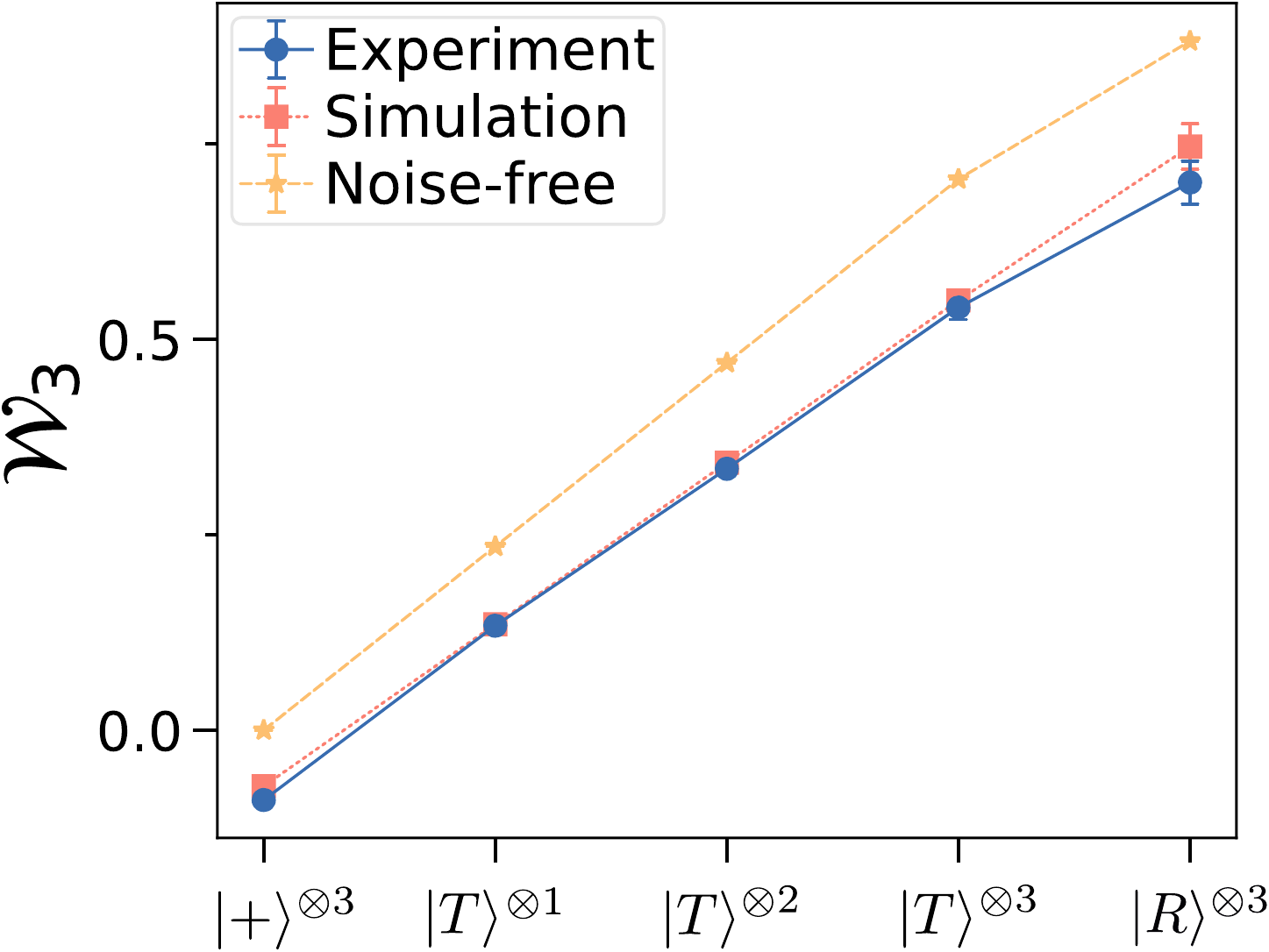}
    \subfigimg[width=0.3\textwidth]{b}{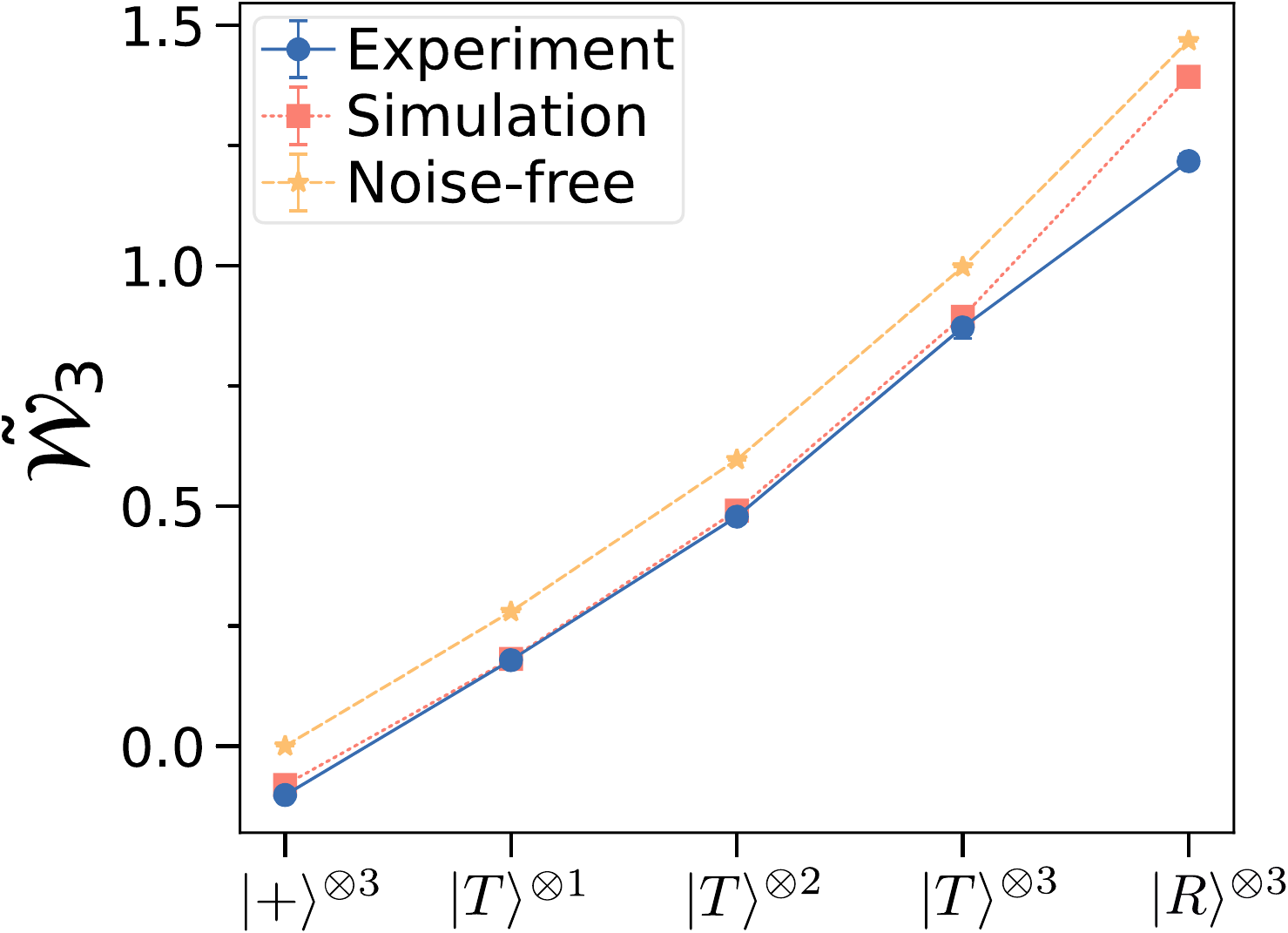}
    \caption{\idg{a} Magic witness $\mathcal{W}_3$ and \idg{b} filtered witness $\Tilde{\mathcal{W}}_3$  measured on the IonQ quantum computer different product states. We have $n=3$ qubits and estimate an effective average global depolarisation noise $p\approx 0.04$ from the purity. We show experimental result in blue and simulation with global depolarisation noise in orange, while yellow line is a simulation for noise-free states.
	}
	\label{fig:Ionq_sup}
\end{figure}

\section{T-gates for PRDM}\label{sec:Tgate}
In this section, we show that pseudorandom density matrices (PRDMs) require $\omega(\log n)$ T-gates to be prepared whenever $S_2=O(\log n)$.

PRDMs are efficiently preparable states that are indistinguishable for any efficient observer from truly random mixed states~\cite{bansal2024pseudorandomdensitymatrices}. 
They generalize pseudorandom states (PRSs), which are computationally indistinguishable from Haar random states~\cite{ji2018pseudorandom}. 
PRDMs are indistinguishable from random mixed states. Formally, these random mixed states are the so-called generalized Hilbert-Schmidt ensemble (GHSE)~\cite{braunstein1996geometry,hall1998random,Zyczkowski_2001,bansal2024pseudorandomdensitymatrices}
\begin{equation}\label{eq:GHSE}
    \eta_{n,m}=\{\operatorname{tr}_m(\ket{\psi}\bra{\psi})\}_{\ket{\psi}\in\mu_{n+m}}\,,
\end{equation}
which are states constructed by taking $n+m$ qubit random states from the Haar measure $\mu_{n+m}$ and tracing out $m$ qubits.

PRDMs are now defined as follows:
\begin{definition}[Pseudo-random density matrix (PRDM)~\cite{bansal2024pseudorandomdensitymatrices}]\label{def:PRDM}
    Let $\kappa=\operatorname{poly}(n)$ be the security parameter with keyspace $\mathcal{K}=\{0,1\}^{\kappa}$. A family of $n$-qubit density matrices $\{\rho_{k,m}\}_{k \in \mathcal{K}}$ are pseudorandom density matrices (PRDMs) with mixedness parameter $m$ if:
    \begin{enumerate}
        \item {Efficiently preparable}: There exists an efficient quantum algorithm $\mathcal{G}$ such that $\mathcal{G}(1^{\kappa}, k,m) = \rho_{k,m}$.
        \item {Computational indistinguishability}: $t=\mathrm{poly}(n)$ copies of $\rho_{k,m}$ are computationally indistinguishable (for any quantum polynomial time adversary $\mathcal{A}$) from the GHSE $\eta_{n,m}$
        \begin{equation}
            \Big{|}\Pr_{k \leftarrow \mathcal{K}}[\mathcal{A}(\rho_{k,m}^{\otimes t}) = 1] - \Pr_{\rho \leftarrow \eta_{n,m}}[\mathcal{A}(\rho^{\otimes t}) = 1]\Big{|} = \operatorname{negl}(n).
        \end{equation}
    \end{enumerate}
\end{definition}
For $m=0$, one recovers PRS~\cite{ji2018pseudorandom}, while for $m=\omega(\log n)$, PRDMs are computationally indistinguishable from the maximally mixed state~\cite{bansal2024pseudorandomdensitymatrices,haug2025pseudorandom}.
However, for the case $m=O(\log n)$ not many results have been known so far.

While PRDMs are efficient, can we give more concrete bounds on their preparation?
Here, we characterize the complexity in terms of number of T-gates.
For entropy $m=S_2=0$, one requires $\Omega(n)$ T-gates to prepare PRDMs~\cite{grewal2023improved}. 
We now show that for PRDMs with $m=S_2=O(\log n)$, at least $\omega(\log n)$ T-gates are required:
\begin{proposition}[T-gates for PRDM]
    Any family of circuits consisting of Clifford operations and $N_\text{T}$ T-gates requires $N_\text{T}=\omega(\log n)$ to prepare PRDMs with entropy $S_2=O(\log n)$.
\end{proposition}
\begin{proof}
This can be shown by contradiction: Assume PRDMs preparable with $N_\text{T}=O(\log n)$ and $S_2=O(\log n)$ exist. Circuits with $N_\text{T}$ T-gates can be prepared using $N_\text{T}$ T-states $\ket{T}$ and Clifford operations, which has $\text{LR}(\ket{T}^{\otimes N_\text{T}})=O(N_\text{T})$ due to sub-additivity and monotonicity~\cite{howard2017robustness}. Due to $\mathcal{W}_3\leq 2\text{LR}$, this implies $\mathcal{W}_3=O(\log n)$ and thus $-\ln A_3=O(\log n)$.

Next, we compute $-\ln A_3(\rho)$ for the GHSE $\rho\in\eta_{n,m}$ with $m=O(\log n)$.
For $m=0$, i.e. Haar random states, Ref.~\cite{turkeshi2023pauli} computed the SRE exactly, showing that $-\ln A_3(\ket{\psi})=\Theta(n)$ for $\ket{\psi}\in\eta_{n,0}$.  
Then, we note that one can write states in the Pauli basis as $\rho=2^{-n-m}\sum_{P\in\mathcal{P}_{n+m}}\beta_P P$, with coefficients $\beta_P$. This gives us
\begin{equation}
    -\ln A_3(\rho)=-\ln(2^{-n-m}\sum_{P\in\mathcal{P}_{n+m}}\text{tr}(\rho P)^6)=-\ln(\sum_{P\in\mathcal{P}_{n+m}}\text{tr}(\rho P)^6)+(n+m)\ln2=-\ln(\sum_{P\in\mathcal{P}_{n+m}}\beta_P^6)+(n+m)\ln2\,.
\end{equation}
When performing the partial trace over $m$ qubits for $\rho$, one only keeps the identity Pauli terms $I_m$ on the support of $m$. In particular, one has 
\begin{equation}
    \text{tr}_m(\rho)=2^{-n-m}\sum_{P\in \mathcal{P}_{n+m}} \beta_{P}\text{tr}_m(P)=2^{-n}\sum_{P'\in \mathcal{P}_{n}} \tilde{\beta}_{P'}P'\,,
\end{equation}
with $\tilde{\beta}_{P'}=\beta_{P'\otimes I_m}$.
This results in
\begin{equation}
    -\ln A_3(\text{tr}_m(\rho))=-\ln(\sum_{P\in\mathcal{P}_{n}}\tilde{\beta}_P^6)+n\ln2\,.
\end{equation}
Now, one can easily see
 This results in
\begin{equation}
    -\ln A_3(\rho)=-\ln(\sum_{P\in\mathcal{P}_{n+m}}\beta_P^6)+(n+m)\ln2 \leq -\ln(\sum_{P\in\mathcal{P}_{n}}\tilde{\beta}_P^6)+(n+m)\ln2 \equiv -\ln A_3(\text{tr}_m(\rho)) + m\ln2
\end{equation}
and
\begin{equation}
 -\ln A_3(\text{tr}_m(\rho)) \geq -\ln A_3(\rho) - m\ln2\,.
\end{equation}
Now, for $m=O(\log n)$ and $-\ln A_3(\rho)=\Theta(n)$, we have
\begin{equation}
 -\ln A_3(\text{tr}_m(\rho)) \geq \Theta(n)\,,
\end{equation}
i.e. the GHSE of~\eqref{eq:GHSE} with $m=O(\log n)$ has $-\ln A_3(\text{tr}_m(\rho))=\Theta(n)$. 
However, our test of Thm.~\ref{thm:testing} can efficiently distinguish $-\ln(A_3)=O(\log n)$ and $-\ln(A_3)=\Theta(n)$, and thus distinguish the assumed PRDM (with $N_\text{T}=O(\log n)$) from random mixed states. This contradicts the definition of PRDM and therefore PRDMs must have $N_\text{T}=\omega(\log n)$ whenever $S_2=O(\log n)$.
\end{proof}

\section{Pseudomagic}\label{sec:pseudomagic}

Here, we derive the pseudomagic gap of state ensembles with $S_2=O(\log n)$ entropy.
First, let us formally define pseudomagic state ensembles~\cite{bouland2022quantum,gu2023little,haug2023pseudorandom,bansal2024pseudorandomdensitymatrices}:
\begin{definition}[Pseudomagic]
A \emph{pseudomagic state ensemble} with gap $f(n)$ vs. $g(n)$ consists of two ensembles of $n$-qubit states $\rho_k$ and $\sigma_k$, indexed by a secret key $k \in \mathcal{K}$, $k\in\{0,1\}^{\mathrm{poly}(n)}$ with the following properties:

\begin{enumerate}
\item \emph{Efficient Preparation}: Given $k$, $\rho_k$ (or $\sigma_k$, respectively) is efficiently preparable by a uniform, poly-sized quantum circuit.

\item \emph{Pseudomagic}: With probability $\geq 1 - 1/\mathrm{poly}(n)$ over the choice of $k$, the log-free robustness of magic $\mathrm{LR}(\rho)=\ln(\min \vert c_\phi\vert \text{ s.t }  \rho=\sum_{\phi \in \mathrm{STAB}}c_\phi \phi)$ 
and stabilizer fidelity $D_\mathrm{F}(\rho)=\min_{\sigma\in \mathrm{STAB}}-\ln \mathcal{F}(\rho,\sigma)$
for $\rho_k$ (or $\sigma_k$, respectively) is $\Theta(f(n))$ (or $\Theta(g(n))$, respectively).

\item \emph{Indistinguishability}: For any polynomial $p(n)$, no poly-time quantum algorithm can distinguish between the ensembles of $\mathrm{poly}(n)$ copies 
with more than negligible probability. That is, for any poly-time quantum algorithm $A$, we have that
\[\left| \Pr_{k \gets \mathcal{K}} [A(\rho_k^{\otimes \mathrm{poly}(n)}) = 1] - \Pr_{k \gets \mathcal{K}} [A(\sigma_k^{\otimes \mathrm{poly}(n)}) = 1] \right| = \operatorname{negl}(n)\,.\]
\end{enumerate}
\end{definition}
When the ensemble is pure, i.e. has entropy $S_2(\rho)=0$, then the maximal possible pseudomagic gap is $f(n)=\Theta(n)$ vs $g(n)=\omega(\log n)$~\cite{gu2023little}. In contrast, for highly impure states $S_2=\omega(\log n)$, one has $f(n)=\Theta(n)$ vs $g(n)=0$. The gap for $S_2=O(\log n)$ was previously not known.

Note that for pseudomagic, both low and high magic ensemble must be indistinguishable for any efficient quantum algorithm. This directly puts bounds on the entropy $S_2$ that the ensemble can have:
\begin{lemma}[Entropy for pseudomagic]
Let us assume we have a pseudomagic ensemble with two ensembles, one with high magic $f(n)$ and low magic $g(n)$. We define the $2$-R\'enyi  entropy of the high-magic ensemble as $S_2^f$ and the low-magic ensemble as $S_2^g$. When $S_2^f=O(\log n)$, then $S_2^g=S_2^f+\text{negl}(n)$.
In contrast, when $S_2^f=\omega(\log n)$, then one must have $S_2^g=\omega(\log n)$.
\end{lemma}
\begin{proof}
This directly follows from  the SWAP test, which efficiently measures $\text{tr}(\rho^2)$ up to additive precision using $O(\epsilon^{-2})$ copies of $\rho$. Note that for entropy we take the logarithm $S_2=-\log (\text{tr}(\rho^2))$. Using Hoeffding's inequality, one can show that two ensembles with entropies $S_2^f$ and $S_2^g$  can be efficiently distinguished whenever $S_2^f=O(\log n)$ and $S_2^g=S_2^f+1/\text{poly}(n)$. 

In contrast, when $S_2^f=\omega(\log n)$, then one can efficiently distinguish it from $S_2^g=O(\log n)$ using polynomial copies via the SWAP test. %
\end{proof}

Now, we show that when both ensembles have $S_2=O(\log n)$, the pseudomagic gap is $f(n)=\Theta(n)$ vs $g(n)=\omega(\log n)$.

\begin{proposition}[Pseudomagic of mixed states]
Pseudomagic state ensembles with entropy $S_2=O(\log n)$ can have a pseudomagic gap of at most $f(n)=\Theta(n)$ vs $g(n)=\omega(\log n)$. 
\end{proposition}
\begin{proof}
Let us prove via contradiction: Let us assume there exist a pseudomagic state ensemble with $f(n)=\Theta(n)$ vs $g(n)=\Theta(\log n)$. Then, due to Thm.~\ref{thm:testing}, there exist an efficient quantum algorithm to distinguish states with $O(\log n)$ magic and $\omega(\log n)$ magic whenever $S_2=O(\log n)$. Thus, this pseudomagic state ensemble cannot exist as it can be efficiently distinguished. 
Thus, we must have a pseudomagic gap of $f(n)=\Theta(n)$ vs $g(n)=\omega(\log n)$.
\end{proof}
Note that to define pseudomagic, one must choose a resource monotone of magic. It has been noted that not all magic monotones have the same pseudomagic gap. In particular, highly non-robust magic measures such as stabilizer nullity (which can change by a factor $\Omega(n)$ by infinitesimally small perturbations) are not good measures for pseudomagic. Here, we choose the robustness of magic $\text{LR}$ and stabilizer fidelity $D_\text{F}$, both of which are well defined magic monotones~\cite{rubboli2024mixed}. 
Note that for technical reasons, we require that both $D_\text{F}$ and $\text{LR}$ must show the same scaling. However, possible future work could reduce testing to depend on $D_\text{F}$ only.

\section{Witnessing magic in many-body systems}\label{sec:manybody}
Here, we present additional data on witnessing magic in many-body systems.
We regard the ground state of the transverse-field Ising model as function of field $h$ with same parameters as in main text.

In Fig.~\ref{fig:ising_sup}a, we show the SRE $M_2$ for different $h$ of the ground state. We see a clear peak close to the critical point $h=1$.
In Fig.~\ref{fig:ising_sup}b, we show the filtered witness $\Tilde{\mathcal{W}}_2$ against size of the bipartition $\ell$ with $\rho=\text{tr}_{\bar{\ell}}(\ket{\psi})$. It shows similar behavior as $\mathcal{W}_2$ shown in main text. Notably,  we observe that $\Tilde{\mathcal{W}}_2$ is slightly more sensitive in detecting magic than $\mathcal{W}_2$. 
In Fig.~\ref{fig:ising_sup}c, we regard the $2$-R\'enyi entropy of the bipartition as function of $\ell$. $S_2$ decreases with $h$. Notably, we find that for $h>1$ the entropy is decaying at the boundaries of $\ell$. This explains the sharp rise of $\mathcal{W}_2$ that we observe for large $\ell$, as the second term in  $\mathcal{W}_2$ strongly depends on $S_2$.
In Fig.~\ref{fig:ising_sup}d, we regard the bipartition $\ell_\text{c}$ for which we detect magic, i.e. for $\ell\geq\ell_\text{c}$ we have $\mathcal{W}_2>0$. 
We find that for small $h$, one requires large $\ell_\text{c}$ to detect magic, which subsequently decreases with $h$ until $h\approx1$. Notably, the decrease appears to be exponentially fast with $h$, i.e. $\ell_\text{c}\approx n^{-h+1}$.
For $h>1$, we find $\ell_\text{c}=2$.

\begin{figure}[htbp]
	\centering	
\subfigimg[width=0.24\textwidth]{a}{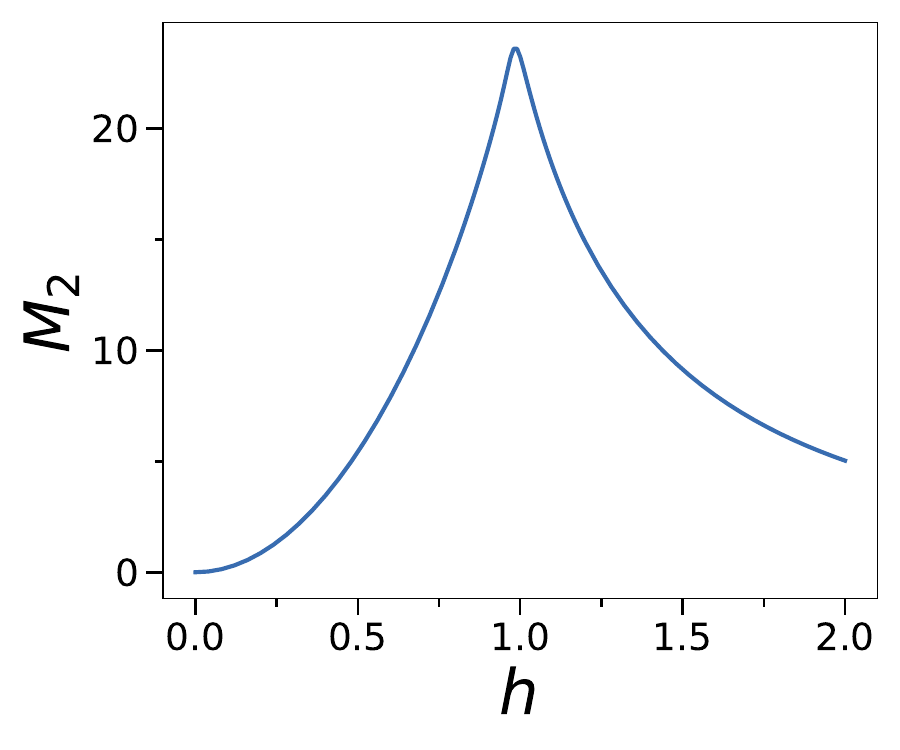}
\subfigimg[width=0.24\textwidth]{b}{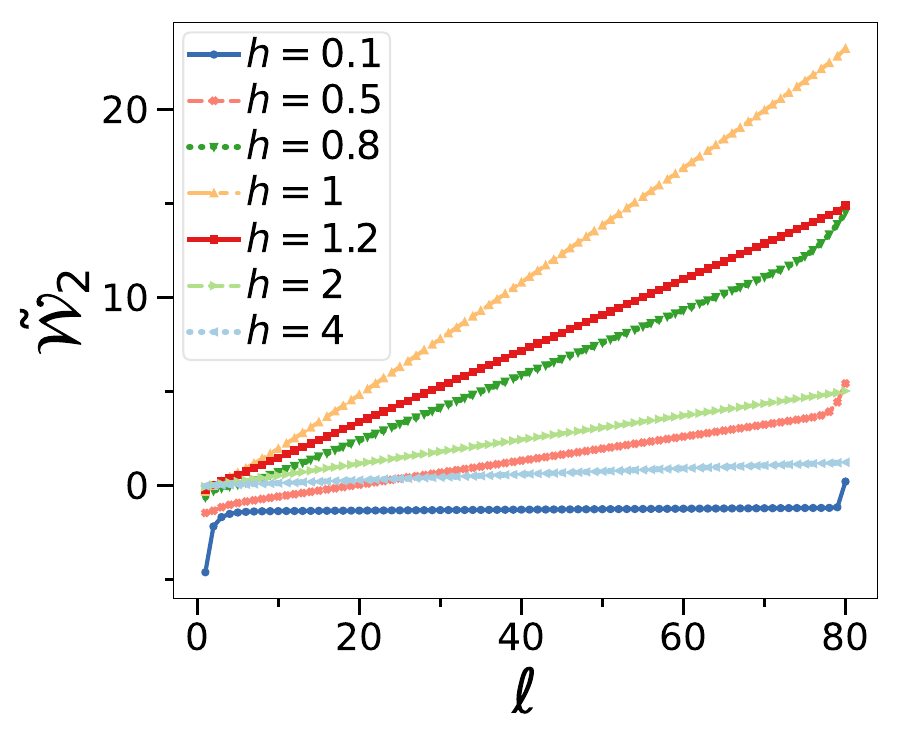}
\subfigimg[width=0.24\textwidth]{c}{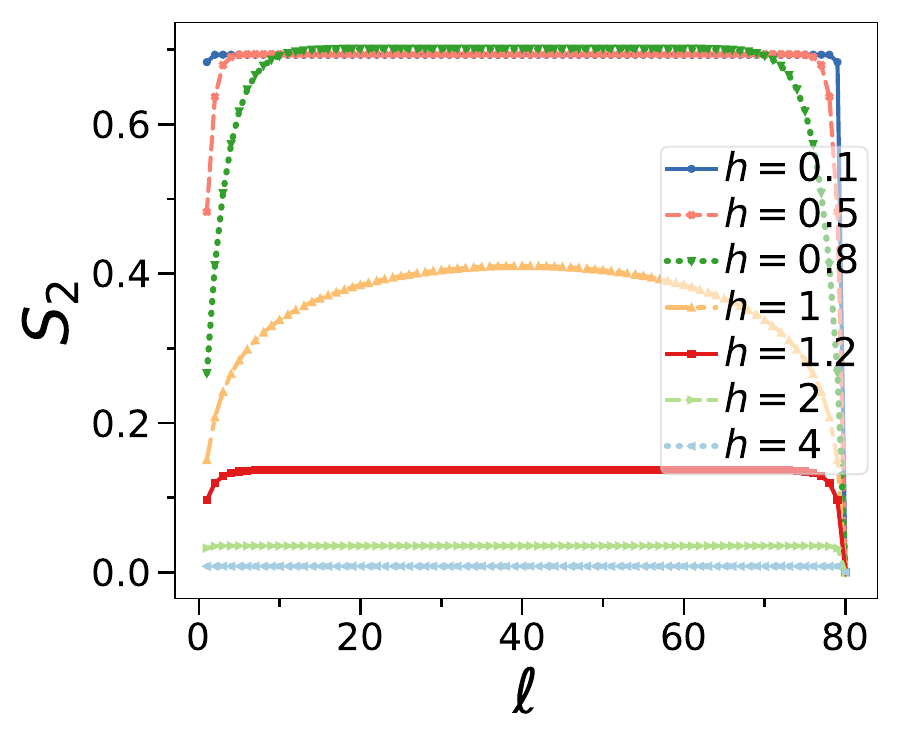}
\subfigimg[width=0.24\textwidth]{d}{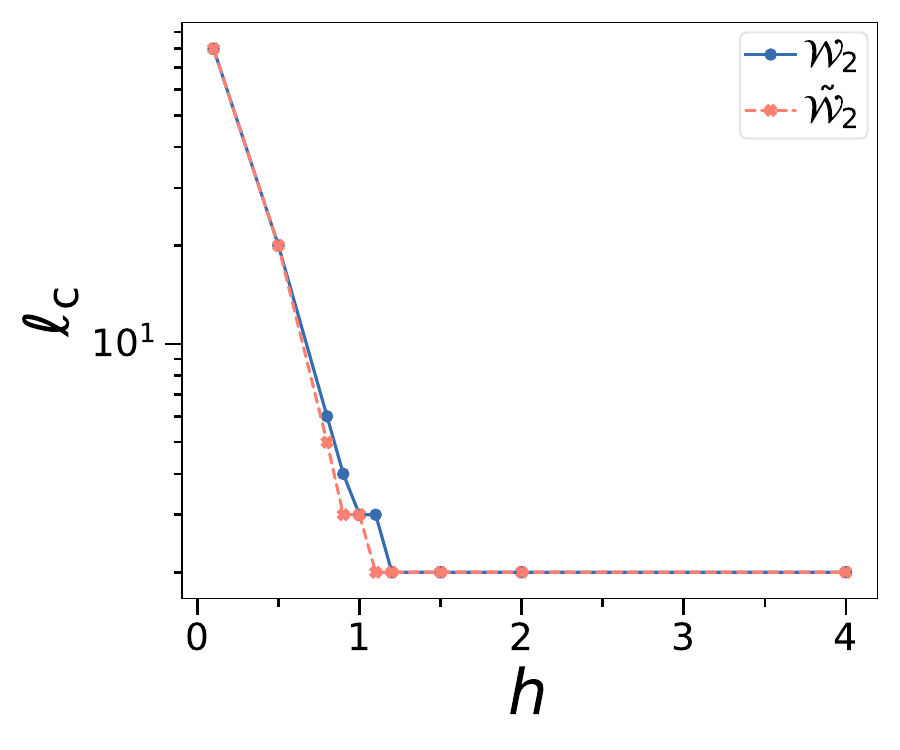}
    \caption{Ground state of the transverse-field Ising model with $n=80$ qubits.
    \idg{a} SRE $M_2$ as function of $h$.
    \idg{b} Filtered witness $\Tilde{\mathcal{W}}_2$ against bipartition size $\ell$ for different $h$. 
    \idg{c} Entropy $S_2(\ell)$ of bipartition as function of $\ell$ for different $h$.
    \idg{d} Bipartition size $\ell\geq \ell_\text{c}$ for which $\mathcal{W}_2>0$, i.e. we can detect magic. We show for different $h$ and standard witness $\mathcal{W}_2$ and filtered witness $\Tilde{\mathcal{W}}_2$.
	}
	\label{fig:ising_sup}
\end{figure}

\end{document}